\documentclass[preprint,11pt]{elsarticle}

\usepackage{amsmath,amssymb,amsthm}
\usepackage[margin=1.0in]{geometry}
\usepackage{tikz}
\usepackage{hyperref}
\usepackage[dvipsnames]{xcolor}
\usepackage{comment}
\usepackage{subcaption}
\usepackage{graphicx}
\usepackage[normalem]{ulem}
\usepackage[pagewise]{lineno}
\usepackage{mathtools}
\mathtoolsset{showonlyrefs}

\newtheorem{theorem}{Theorem}[section]
\newtheorem{lemma}[theorem]{Lemma}
\newtheorem{corollary}[theorem]{Corollary}
\newtheorem{definition}[theorem]{Definition}
\newtheorem{remark}[theorem]{Remark}

\newtheorem{proposition}[theorem]{Proposition}
\newtheorem{assumption}[theorem]{Assumption}

\numberwithin{equation}{section}

\journal{@@}

\begin{document}

\begin{frontmatter}

\title{Endogenous Reinsurance Pricing in Large Competitive Insurance Markets: Finite-Player and Mean Field Analysis}

\author[inst1]{Ruimeng Hu}
\ead{rhu@ucsb.edu}

\author[inst2]{Byungdoo Kong\corref{cor1}}
\ead{bkong@ucsb.edu}

\cortext[cor1]{Corresponding author}

\affiliation[inst1]{
  organization={Department of Mathematics, Department of Statistics and Applied Probability, University of California},
  city={Santa Barbara},
  country={USA}
}

\affiliation[inst2]{
  organization={Department of Statistics and Applied Probability, University of California},
  city={Santa Barbara},
  country={USA}
}

\begin{abstract}
We study endogenous reinsurance pricing in a competitive insurance market with one strategic reinsurer and many heterogeneous insurers. The reinsurer acts as a Stackelberg leader by choosing a common premium rate and an investment strategy, while insurers decide how much risk to retain and how to invest, taking into account their own performance, their performance relative to the insurer population, and common insurance-claim and financial-market noise. This creates a feedback loop absent from standard reinsurance models with exogenous premiums: a premium change affects insurers directly through the cost of reinsurance, and indirectly through the population's aggregate exposure to common insurance-claim risk.

For a fixed premium, we characterize the insurers' equilibrium retention through a scalar fixed point and establish its monotone premium response. This characterization reveals a spillover mechanism generated by relative performance concerns and leads to a threshold structure in which insurers move from full cession to partial retention and then to full retention as the premium increases. Using this structure, we reduce the reinsurer's premium problem to a one-dimensional optimization over a compact premium interval and characterize Stackelberg equilibria in both finite-player and mean field models. In the finite-player case, we develop an efficient threshold continuation procedure that determines equilibrium premiums without enumerating all retention configurations. We also prove convergence from finite-player equilibria to mean field equilibria without requiring the mean field equilibrium premium to be unique. Numerical illustrations show how relative performance concerns amplify spillover effects and can induce retention even when reinsurance remains actuarially favorable. They also demonstrate that Stackelberg equilibria need not be unique in either setting.
\end{abstract}


\end{frontmatter}

\section{Introduction} 
\label{sec:intro}
Insurers face underwriting risk from insurance liabilities and financial market risk from their investments. Reinsurance allows an insurer to transfer part of its claim risk to a reinsurer, while investment allocates surplus in financial markets to support profitability and solvency. The resulting optimal reinsurance and investment problem is a well-studied and active topic in actuarial science. A key feature of reinsurance, however, is that it involves interaction between two parties: the insurer chooses how much risk to cede, while the reinsurer sets the premium charged for the ceded risk. When the reinsurer prices strategically, this premium becomes an endogenous market variable rather than an external input.

Much of the classical literature on optimal reinsurance and investment takes the perspective of a single insurer and treats the reinsurance premium as exogenous. A natural way to endogenize premium formation is to model the interaction between the insurer and the reinsurer as a Stackelberg game. In this direction, \cite{ChenShen2019} study a stochastic Stackelberg differential reinsurance game under a time-inconsistent mean--variance criterion, while \cite{LiYoung2022} analyze a mean--variance Stackelberg reinsurance game with a random horizon. These works establish the insurer–reinsurer Stackelberg paradigm in continuous time and show how equilibrium reinsurance decisions depend on the premium policy chosen by the reinsurer. Their focus, however, is mainly on a bilateral market with one insurer and one reinsurer, rather than on strategic interaction among many competing insurers.

A separate line of research studies competition among insurers under relative performance concerns. For large populations, the mean field game (MFG) framework provides a natural limiting formulation; the modern theory was initiated independently by \cite{LasryLions2006a,LasryLions2006b,LasryLions2007} and by \cite{HuangMalhameCaines2006,HuangCainesMalhame2007}. Finite-player and mean field models have been developed to describe large insurance markets in which insurers interact through terminal wealth or related aggregate quantities. For example, \cite{BoWangZhou2024} analyze competitive insurers under exponential utility with relative performance and establish both a finite $n$-insurer constant Nash equilibrium and a corresponding mean field equilibrium, together with a convergence result. Related finite-player and mean field studies include \cite{GuanHu2022} and \cite{HeEtAl2023Robust}, which incorporate reinsurance-investment decisions, common and idiosyncratic sources of risk, and large-population limits. These papers show that relative performance and mean field interactions can substantially affect insurers' equilibrium behavior. In these models, however, premium formation is typically exogenous, or the reinsurer's Stackelberg pricing problem is not the main focus.

More closely related to our paper are studies in which both insurers and reinsurers are strategic players in the reinsurance market. In particular, \cite{BaiEtAl2022Hybrid} study a hybrid stochastic differential reinsurance-investment game with one reinsurer and two competing insurers, in which the reinsurer acts as a Stackelberg leader and the insurers interact through a non-zero-sum game with relative performance. This paper is especially relevant because it combines the reinsurer's endogenous premium choice with strategic interaction among insurers. Still, its focus is on a small finite-player market with one reinsurer and two competing insurers, together with bounded memory effects. More recently, \cite{LiangXiaZou2024} and \cite{CaoEtAl2025} study richer reinsurance market games with multiple reinsurers, Stackelberg contracting, and competition or cooperation among reinsurers, broadening the game-theoretic analysis of reinsurance markets. A closely related recent working paper, \cite{ZhangEtAl2024ReinsuranceMFG}, also studies a Stackelberg mean field formulation with one reinsurer and competitive insurers. Our paper differs in its joint finite-player and mean field analysis, explicit premium response and threshold structure, and forward convergence from finite-player equilibria to mean field ones.

Against this background, we study a reinsurance-investment market with one strategic reinsurer and many heterogeneous insurers in both finite-player and mean field settings. The insurers have relative performance concerns and are exposed to common insurance-claim and financial-market noise. This setting bridges a gap between two strands of the literature: Stackelberg reinsurance models, where the premium is chosen strategically but insurer competition is limited, and competitive insurer models, where relative performance and large population effects are present but premium formation is often exogenous.

Despite this richer market structure, the equilibrium remains analytically tractable under constant strategies. The key reduction is a scalar fixed point for the average common-insurance exposure, which determines the insurers' equilibrium retention response and enters the reinsurer's reduced premium problem. This reduction also reveals a market spillover channel: premium changes affect insurers directly through the premium-loss spread and indirectly through the induced change in average common-insurance exposure.

More specifically, our contributions are threefold. First, we formulate and analyze a Stackelberg reinsurance-investment model with one reinsurer and heterogeneous insurers in both finite-player and mean field settings. For a fixed premium, we characterize the follower (insurers') equilibrium by a scalar fixed point for the average common-insurance exposure. This gives explicit equilibrium retention responses and reveals a spillover mechanism generated by relative performance concerns: insurers respond not only to the direct premium-loss spread, but also to the induced change in the population's average common-insurance exposure. This mechanism explains how some insurers may begin to retain part of their risk even when reinsurance remains actuarially favorable.

Second, we reduce the reinsurer's problem to a one-dimensional optimization over a compact interval, and establish the existence and explicit characterization of Stackelberg equilibria in both finite-player and mean field models. This reduction is possible because the follower response has a threshold structure: as the premium increases, insurers move from full cession to partial retention and then to full retention. In the finite-player case, we further develop a threshold continuation procedure that computes equilibrium premiums without the exponential-time brute-force enumeration of retention configurations, reducing the computational cost from exponential to quadratic in the number of insurers.

Third, we establish finite-player-to-mean-field convergence. Under convergence of the empirical type distributions, the finite-player follower responses, the reinsurer's premium objective, and the optimal premium sets converge to their mean field counterparts. Since the Stackelberg mean field equilibrium need not be unique, we formulate convergence at the level of optimal premium sets and subsequential empirical type-control distributions: every subsequential limit of finite-player equilibrium type-control distributions is induced by a Stackelberg mean field equilibrium. This provides a rigorous forward convergence justification for the mean field approximation in the present reinsurance model.

The remainder of the paper is organized as follows. Section~\ref{sec:finite_game} introduces the finite-player Stackelberg--Nash reinsurance-investment game, characterizes the insurers' equilibrium response for a given premium, and analyzes the reinsurer's induced optimization problem. Section~\ref{sec:smfg} develops the corresponding Stackelberg mean field game and characterizes the mean field equilibrium. Section~\ref{sec:convergence} studies the large-population limit and establishes the connection between finite-player equilibria and their mean field counterparts. Section~\ref{sec:numerics} gives numerical illustrations of the effects of relative performance concerns on equilibrium retention and premium decisions. We conclude in Section~\ref{sec:conclusion}.

\section{Finite-Player Stackelberg--Nash Game}
\label{sec:finite_game}
\subsection{Model Formulation and Equilibrium Concept}
\label{subsec:finite_formulation}
We consider a finite-player Stackelberg--Nash game with one reinsurer and $N$ heterogeneous insurers. The reinsurer acts as the Stackelberg leader and first chooses a common reinsurance premium rate for all insurers, along with its own investment control. The insurers then act as followers: given the reinsurer's choice, each insurer chooses its retention level and investment control, taking into account its performance relative to the other insurers. Both the insurance and financial markets are subject to common shocks.

Let $T>0$ and let $(\Omega,\mathcal{F},\mathbb{F}=(\mathcal{F}_{t})_{0\le t\le T},\mathbb{P})$ support mutually independent standard Brownian motions $W^{0},W^{1},\dots,W^{N},B^{0},B^{1},\dots,B^{N},B^{L}$. Here, $W^{0}$ represents the common insurance-claim noise affecting all insurers, while $W^{i}$ represents the idiosyncratic insurance-claim noise of insurer $i$. The process $B^{0}$ represents the common financial-market noise shared by the insurers and the reinsurer, while $B^{i}$ and $B^{L}$ represent the idiosyncratic financial-market noises faced by insurer $i$ and the reinsurer, respectively.

For each insurer $i=1,\dots,N$, we model insurer $i$'s aggregate claim process by a diffusion approximation; see, for example, \cite{Iglehart1969, Browne1995, AsmussenTaksar1997}. Specifically,
\begin{equation*}
dC_{t}^{i}=a^{i}\,dt-v^{i,0}\,dW_{t}^{0}-v^{i}\,dW_{t}^{i},
\end{equation*}
where $a^{i}>0$ is the expected claim rate, $v^{i,0}>0$ and $v^{i}>0$ are the common and idiosyncratic insurance-claim volatility coefficients, respectively. The common insurance-claim noise, $W^{0}$, captures market-wide insurance shocks that affect all insurers.

Let $\eta^{i}>0$ be insurer $i$'s safety loading, so the direct insurance premium rate is $(1+\eta^{i})a^{i}$ under the expected value premium principle. At time $t$, insurer $i$ chooses a proportional retention level $q_{t}^{i}\in[0,1]$, so that $1-q_{t}^{i}$ is the ceded proportion. The insurer's insurance surplus is therefore
\begin{equation*}
dY_{t}^{i}=(1+\eta^{i})a^{i}\,dt-dC_{t}^{i}-(1-q_{t}^{i})\big(p_{t}\,dt-dC_{t}^{i}\big)=\big(\eta^{i}a^{i}-(1-q_{t}^{i})(p_{t}-a^{i})\big)\,dt+q_{t}^{i}v^{i,0}\,dW_{t}^{0}+q_{t}^{i}v^{i}\,dW_{t}^{i}.
\end{equation*}
Here, $p_t\ge0$ is the common reinsurance premium rate chosen by the reinsurer and charged to all insurers at time $t$. 

\begin{remark}
Although reinsurance treaties in practice may involve cedent-specific negotiations, we abstract from such bilateral negotiations and model the reinsurer as setting a common unit reinsurance rate applied to ceded risk. This common pricing rule does not eliminate insurer heterogeneity. Even at the common premium rate, differences in insurer-specific parameters can lead to different retention levels, which, in turn, result in different ceded proportions and total reinsurance payments across insurers. We use this formulation to study how a reinsurance premium set by the reinsurer affects a heterogeneous population of insurers and how the equilibrium premium is determined endogenously through a Stackelberg interaction.
\end{remark}

We now turn to the reinsurer's insurance surplus, which is determined by the reinsurance premiums paid by the insurers and the claims ceded by them:
\begin{equation}
\label{eq:surplus_reinsurer}
dY_{t}^{L}=\frac{1}{N}\sum_{i=1}^{N}(1-q_{t}^{i})\big(p_{t}\,dt-dC_{t}^{i}\big)=\frac{1}{N}\sum_{i=1}^{N}(1-q_{t}^{i})\big((p_{t}-a^{i})\,dt+v^{i,0}\,dW_{t}^{0}+v^{i}\,dW_{t}^{i}\big).
\end{equation}
Here, we normalize the reinsurer's insurance surplus by $1/N$. This normalization should be understood as a market-size convention, rather than a restriction on the amount of reinsurance accepted. The model keeps the total insurance market represented by the population at a fixed scale as $N$ grows; increasing $N$ means that this market is divided among more insurers. In the equal-size case, each insurer therefore has relative weight $1/N$ in the reinsurer's book. Each insurer still cedes the proportion $1-q_t^i$ of its own claim risk, while the reinsurer's surplus is measured in normalized market units. This convention keeps the reinsurer's surplus of order one in the large-population limit.

\begin{remark}
The normalization in \eqref{eq:surplus_reinsurer} does not affect the definition or solvability of the finite-player problem. In particular, removing the factor $1/N$ changes only the scale of the reinsurer-side quantities, and after the corresponding rescaling, the finite-player arguments in Section~\ref{sec:finite_game} remain valid. The choice of $1/N$ corresponds to an equal-weight normalization of the ceded portfolio, in which each insurer has the same relative size in the reinsurer's normalized book. More generally, one could replace $1/N$ with insurer-specific weights to model insurers with different
relative market sizes. We use equal-weight normalization because it is the simplest formulation and leads to a well-defined large-population limit. 
\end{remark}

We next specify the financial market in which insurers and the reinsurer invest. Without loss of generality, we assume that the risk-free rate is zero, and model the risky asset prices $S_t^i$, available to insurer $i$, and $S_t^L$, available to the reinsurer, by
\begin{equation*}
\frac{dS_{t}^{i}}{S_{t}^{i}}=\mu^{i}\,dt+\sigma^{i,0}\,dB_{t}^{0}+\sigma^{i}\,dB_{t}^{i},\qquad\frac{dS_{t}^{L}}{S_{t}^{L}}=\mu^{L}\,dt+\sigma^{L,0}\,dB_{t}^{0}+\sigma^{L}\,dB_{t}^{L},
\end{equation*}
where $\mu^{i}$ and $\mu^{L}$ are the expected (excess) returns, while $\sigma^{i,0},\sigma^{i}>0$ and $\sigma^{L,0},\sigma^{L}>0$ are the common and idiosyncratic financial-market volatility coefficients. The common financial-market noise $B^{0}$ captures market-wide financial shocks affecting all insurers and the reinsurer.

Let $\pi_{t}^{i}\in\mathbb{R}$ and $\pi_{t}^{L}\in\mathbb{R}$ be the dollar amounts invested in the corresponding risky assets by insurer $i$ and the reinsurer at time $t$. Then their wealth processes satisfy
\begin{align}
\label{eq:dynamics_insurer}
dX_{t}^{i} & = dY_{t}^{i}+\pi_{t}^{i}\,\frac{dS_{t}^{i}}{S_{t}^{i}}
\nonumber\\
& = \big(\eta^{i}a^{i}-(1-q_{t}^{i})(p_{t}-a^{i})+\pi_{t}^{i}\mu^{i}\big)\,dt+q_{t}^{i}v^{i,0}\,dW_{t}^{0}+q_{t}^{i}v^{i}\,dW_{t}^{i}+\pi_{t}^{i}\sigma^{i,0}\,dB_{t}^{0}+\pi_{t}^{i}\sigma^{i}\,dB_{t}^{i}, 
\\
\label{eq:dynamics_reinsurer}
dX_{t}^{L} & = dY_{t}^{L}+\pi_{t}^{L}\,\frac{dS_{t}^{L}}{S_{t}^{L}}
\nonumber\\
& = \Big(\frac{1}{N}\sum_{i=1}^{N}(1-q_{t}^{i})(p_{t}-a^{i})+\pi_{t}^{L}\mu^{L}\Big)\,dt+\frac{1}{N}\sum_{i=1}^{N}(1-q_{t}^{i})v^{i,0}\,dW_{t}^{0}+\frac{1}{N}\sum_{i=1}^{N}(1-q_{t}^{i})v^{i}\,dW_{t}^{i}
\nonumber\\
& \quad + \pi_{t}^{L}\sigma^{L,0}\,dB_{t}^{0}+\pi_{t}^{L}\sigma^{L}\,dB_{t}^{L},
\end{align}
with initial conditions $X_{0}^{i}=x^{i}$ and $X_{0}^{L}=x^{L}$. 

We now define the players' objective functions. Let
\begin{equation*}
\overline{X}_{T}^{N}:=\frac{1}{N}\sum_{j=1}^{N} X_{T}^{j}
\end{equation*}
be the average terminal wealth of the insurer population. In line with relative performance criteria commonly used in competitive portfolio and insurance games, each insurer evaluates its terminal wealth relative to this population average; see, for example, \cite{EspinosaTouzi2015,LackerZariphopoulos2019,BoWangZhou2024}. Thus, insurer $i$'s objective depends not only on its own terminal wealth but also on the performance of the other insurers:
\begin{equation}
\label{eq:Ji}
J^{i}\big((q^{i},\pi^{i});(q^{-i},\pi^{-i}),(p,\pi^{L})\big)=\mathbb{E}\Big[-\exp\Big(-\frac{1}{\gamma^{i}}\big(X_{T}^{i}-\theta^{i}\overline{X}_{T}^{N}\big)\Big)\Big],
\end{equation}
where $(q^{-i},\pi^{-i})$ denotes the collection of insurer controls excluding insurer $i$, $\gamma^{i}>0$ is insurer $i$'s risk-tolerance parameter, and $\theta^{i}\in[0,1]$ is the relative performance parameter. A larger value of $\theta^{i}$ means that insurer $i$ places more weight on performance relative to the other insurers, while $\theta^{i}=0$ corresponds to no relative performance concern. We exclude the degenerate case $N=\theta^1=1$. The reinsurer also evaluates terminal wealth using exponential utility:
\begin{equation}
\label{eq:JL}
J^{L}\big((p,\pi^{L});(q^{1},\pi^{1}),\dots,(q^{N},\pi^{N})\big)=\mathbb{E}\Big[-\exp\Big(-\frac{1}{\gamma^{L}}X_{T}^{L}\Big)\Big],
\end{equation}
where $\gamma^{L}>0$ is the reinsurer's risk-tolerance parameter.

We now describe the information structure and order of decisions. All model parameters and the decision structure are assumed to be common knowledge. The decision-making is formulated as a Stackelberg game. The reinsurer acts as the Stackelberg leader and first chooses a premium-investment pair $(p,\pi^L)$ to maximize its objective~\eqref{eq:JL}. After observing the announced leader control, the insurers simultaneously choose their retention-investment pairs $(q^i,\pi^i)$, $i=1,\ldots,N$, as functions of $(p, \pi^L)$, and play an $N$-player Nash game~\eqref{eq:Ji}. The reinsurer anticipates this follower equilibrium response when solving its own optimization problem.

To obtain an explicit characterization of the equilibrium, we restrict attention to constant controls:
\begin{equation}
\label{def:constcontrol}
q_{t}^{i}\equiv q^{i}\in[0,1],\qquad\pi_{t}^{i}\equiv\pi^{i}\in\mathbb{R},\qquad p_{t}\equiv p\ge0,\qquad\pi_{t}^{L}\equiv\pi^{L}\in\mathbb{R},\qquad i=1,\dots,N.
\end{equation}
This restriction ensures the tractability of the equilibrium and allows for the closed-form characterization developed in the subsequent sections. We now define the corresponding notion of equilibrium.

\begin{definition}[Constant Stackelberg--Nash equilibrium]
\label{def:CSNE}
Fix a constant leader control $(p,\pi^{L})\in[0,\infty)\times\mathbb{R}$. A follower profile $(q^{1},\pi^{1}),\dots,(q^{N},\pi^{N})$ is called a \emph{constant follower Nash equilibrium (NE) response to $(p,\pi^{L})$} if, for each insurer $i=1,\dots,N$,
\begin{equation*}
J^{i}\big((q^{i},\pi^{i});(q^{-i},\pi^{-i}),(p,\pi^{L})\big)\ge J^{i}\big((\tilde{q}^{i},\tilde{\pi}^{i});(q^{-i},\pi^{-i}),(p,\pi^{L})\big)
\end{equation*}
for all constant controls $\tilde q^{i}\in[0,1]$ and $\tilde \pi^{i}\in\mathbb{R}$. 

A constant control profile $(p^{*},\pi^{L,*},q^{1,*},\pi^{1,*},\dots,q^{N,*},\pi^{N,*})$ is called a \emph{constant Stackelberg--Nash equilibrium} if $(q^{1,*},\pi^{1,*}),\dots,(q^{N,*},\pi^{N,*})$ is a constant follower Nash equilibrium response to $(p^{*},\pi^{L,*})$, and the leader control $(p^{*},\pi^{L,*})$ is optimal for the reinsurer's objective \eqref{eq:JL} given this follower equilibrium response.
\end{definition}

\begin{remark}
The restriction to a constant premium control $p$ and constant retention controls $q^{1},\dots,q^{N}$ can also be interpreted as a contract-period benchmark: the reinsurance premium rate and the ceding policies are fixed at treaty inception or renewal and are revised only periodically, rather than continuously over time. In this sense, the formulation can be viewed as a stylized benchmark for treaty-style reinsurance and renewal markets. The restriction to constant investment controls $\pi^{i}$ and $\pi^{L}$, on the other hand, is essentially without loss of generality: the optimal investment in a Merton-type problem with exponential utility and constant coefficients is well known to be time-independent.
\end{remark}

The analysis therefore proceeds in two steps. We first fix the leader's premium and characterize the induced follower Nash equilibrium. We then substitute this follower response into the reinsurer's objective and solve the leader's reduced premium problem.

\subsection{Follower Nash Equilibrium Response}
\label{subsec:finite_followers}
We characterize the constant follower Nash equilibrium response for a fixed leader control $(p,\pi^{L})$. Note that the leader's investment control $\pi^{L}$ does not enter the followers' wealth dynamics \eqref{eq:dynamics_insurer} or their objective \eqref{eq:Ji}, and hence the follower equilibrium response depends only on the premium level $p$. We therefore proceed in two steps. First, we characterize the follower equilibrium response for a fixed premium level $p$. We then study how the equilibrium response varies with $p$, which will later serve as a key input for the reinsurer's optimization problem.

Assuming constant controls \eqref{def:constcontrol}, the terminal wealth vector $(X_{T}^{1},\dots,X_{T}^{N})$ is jointly Gaussian. Hence, for each insurer $i$, the random variable $X_{T}^{i}-\theta^{i}\overline{X}_{T}^{N}$ is also Gaussian, and 
\begin{equation*}
\mathbb{E}\Big[-\exp\Big(-\frac{1}{\gamma^{i}}\big(X_{T}^{i}-\theta^{i}\overline{X}_{T}^{N}\big)\Big)\Big]=-\exp\Big(-\frac{1}{\gamma^{i}}\mathbb{E}[X_{T}^{i}-\theta^{i}\overline{X}_{T}^{N}]+\frac{1}{2(\gamma^{i})^{2}}\mathrm{Var}(X_{T}^{i}-\theta^{i}\overline{X}_{T}^{N})\Big).
\end{equation*}
Thus, for fixed $p\ge0$, maximizing the objective \eqref{eq:Ji} is equivalent to maximizing
\begin{equation}
\label{eq:Ji_tilde}
\widetilde{J}^{i}\big((q^{i},\pi^{i});(q^{-i},\pi^{-i})\big):=\mathbb{E}[X_{T}^{i}-\theta^{i}\overline{X}_{T}^{N}]-\frac{1}{2\gamma^{i}}\mathrm{Var}(X_{T}^{i}-\theta^{i}\overline{X}_{T}^{N}),
\end{equation}
and the follower-side game reduces to a static mean--variance game. This is the standard reduction for exponential utility under Gaussian terminal wealth; see, e.g., \cite{GerberPafumi1998Utility}. A direct computation using dynamics \eqref{eq:dynamics_insurer} shows that
\begin{align*}
\mathbb{E}[X_{T}^{i}] & =x^{i}+\big(\eta^{i}a^{i}-(1-q^{i})(p-a^{i})+\pi^{i}\mu^{i}\big)T, \\
\mathrm{Var}(X_{T}^{i})& =\big[(q^{i})^{2}\big((v^{i,0})^{2}+(v^{i})^{2}\big)+(\pi^{i})^{2}\big((\sigma^{i,0})^{2}+(\sigma^{i})^{2}\big)\big]T, \\
\mathrm{Cov}(X_{T}^{i},X_{T}^{j}) & =\big(q^{i}q^{j}v^{i,0}v^{j,0}+\pi^{i}\pi^{j}\sigma^{i,0}\sigma^{j,0}\big)T, \quad \forall i\neq j.
\end{align*}
Hence, the reduced criterion \eqref{eq:Ji_tilde} takes the form
\begin{align}
\label{eq:Jtilde_expanded}
\widetilde{J}^{i} = &\ \mathrm{Const.}+T\Big(1-\frac{\theta^{i}}{N}\Big)\,\Big[(p-a^{i})q^{i}+\frac{1}{\gamma^{i}}\theta^{i}v^{i,0}V_{-i}^{N}\,q^{i}-\frac{1}{2\gamma^{i}}\Big(1-\frac{\theta^{i}}{N}\Big)\big((v^{i,0})^{2}+(v^{i})^{2}\big)(q^{i})^{2}\Big]
\nonumber\\
& + T\Big(1-\frac{\theta^{i}}{N}\Big)\,\Big[\Big(\mu^{i}+\frac{1}{\gamma^{i}}\theta^{i}\sigma^{i,0}\Sigma_{-i}^{N}\Big)\,\pi^{i}-\frac{1}{2\gamma^{i}}\Big(1-\frac{\theta^{i}}{N}\Big)\big((\sigma^{i,0})^{2}+(\sigma^{i})^{2}\big)(\pi^{i})^{2}\Big],
\end{align}
where  ``Const.'' summarizes all terms independent of insurer $i$'s own control $(q^i,\pi^i)$, and 
\begin{equation*}
V_{-i}^{N}:=\frac{1}{N}\sum_{j\neq i}q^{j}v^{j,0},\qquad\Sigma_{-i}^{N}:=\frac{1}{N}\sum_{j\neq i}\pi^{j}\sigma^{j,0}.
\end{equation*}

For fixed controls of the other insurers, the reduced criterion \eqref{eq:Jtilde_expanded} separates into a retention term in $q^i$ and an investment term in $\pi^i$, and the premium level $p$ appears only in the retention term. Hence, insurer $i$'s investment response is independent of both the reinsurance premium and its retention decision. We now record the unique follower equilibrium response in the following proposition.

\begin{proposition}[Follower NE for fixed premium]
\label{prop:follower_response}
For fixed $p\ge0$, the followers' game admits a unique constant follower Nash equilibrium response $(\widehat{q}^{1}(p),\widehat\pi^{1}),\dots,(\widehat{q}^{N}(p),\widehat\pi^{N})$. Define the average common-insurance and common-financial exposures associated with the equilibrium by
\begin{equation}
\label{eq:Vhat_Sigmahat_def}
\widehat{V}^{N}(p):=\frac{1}{N}\sum_{j=1}^{N}\widehat{q}^{j}(p)\,v^{j,0},\qquad\widehat{\Sigma}^{N}:=\frac{1}{N}\sum_{j=1}^{N}\widehat{\pi}^{j}\sigma^{j,0},
\end{equation}
and define, for each insurer $i$,
\begin{equation}
\label{eq:retention_incentive}
I^{i}(p,V):=\frac{(p-a^{i})+\frac{1}{\gamma^{i}}\theta^{i}v^{i,0}V}{\frac{1}{\gamma^{i}}\big[(v^{i,0})^{2}+(v^{i})^{2}\big(1-\theta^{i}/N\big)\big]},\qquad p\ge0,\ V\in\mathbb{R}.
\end{equation}
Then, for each insurer $i=1,\dots,N$,
\begin{equation}
\label{eq:pihat}
\widehat{\pi}^{i}=\frac{\mu^{i}+\frac{1}{\gamma^{i}}\theta^{i}\sigma^{i,0}\widehat{\Sigma}^{N}}{\frac{1}{\gamma^{i}}\big[(\sigma^{i,0})^{2}+(\sigma^{i})^{2}(1-\theta^{i}/N)\big]},\qquad\widehat{\Sigma}^{N}=\frac{\frac{1}{N}\sum_{j=1}^{N}\frac{\gamma^{j}\sigma^{j,0}\mu^{j}}{(\sigma^{j,0})^{2}+(\sigma^{j})^{2}(1-\theta^{j}/N)}}{1-\frac{1}{N}\sum_{j=1}^{N}\frac{\theta^{j}(\sigma^{j,0})^{2}}{(\sigma^{j,0})^{2}+(\sigma^{j})^{2}(1-\theta^{j}/N)}},
\end{equation}
and
\begin{equation}
\label{eq:qhat_projection}
\widehat{q}^{i}(p)=\Pi_{[0,1]}\big(I^{i}(p,\,\widehat{V}^{N}(p))\big),
\end{equation}
where $\Pi_{[0,1]}:\mathbb{R}\to[0,1]$ denotes projection onto $[0,1]$. The average common-insurance exposure $\widehat V^{N}(p)$ is determined as the unique solution of the fixed-point equation for $V$:
\begin{equation}
\label{eq:fixed_point_V}
V=\frac{1}{N}\sum_{j=1}^{N}v^{j,0}\Pi_{[0,1]}\big(I^{j}(p,V)\big).
\end{equation}
\end{proposition}

\begin{proof}
See \ref{app:2.1}.
\end{proof}

Proposition~\ref{prop:follower_response} gives a unique follower equilibrium response for each fixed premium level $p$, reducing this to the one-dimensional fixed-point equation \eqref{eq:fixed_point_V} for the average common-insurance exposure $\widehat V^N(p)$. Once $\widehat V^N(p)$ is determined, the equilibrium retention responses are given by \eqref{eq:qhat_projection}. The equilibrium investment responses are independent of $p$ and are given explicitly by \eqref{eq:pihat}.

For insurer $i$, the quantity
\begin{equation}
\label{eq:retention_incentive_equilibrium}
I^{i}(p,\widehat{V}^{N}(p))=\frac{(p-a^{i})+\frac{1}{\gamma^{i}}\theta^{i}v^{i,0}\widehat{V}^{N}(p)}{\frac{1}{\gamma^{i}}\big[(v^{i,0})^{2}+(v^{i})^{2}(1-\theta^{i}/N)\big]}
\end{equation}
determines the equilibrium retention response $\widehat{q}^{i}(p)$ through projection onto $[0,1]$, which we therefore interpret as \emph{insurer $i$'s equilibrium retention incentive}. A larger $I^i$ leads to a larger $\widehat q^i(p)$, with a nonpositive incentive giving full cession $(\widehat{q}^{i}(p)=0)$, an incentive between $0$ and $1$ giving partial retention $(0<\widehat{q}^{i}(p)<1)$, and an incentive at least $1$ giving full retention $(\widehat{q}^{i}(p)=1)$.

Equation \eqref{eq:retention_incentive_equilibrium} shows how insurer $i$'s equilibrium retention incentive depends on both its own characteristics and the aggregate behavior of the insurer population. The term  $p-a^{i}$  is the premium-loss spread: when it is positive, retaining risk becomes more attractive; when it is negative, full cession is actuarially favorable and lowers insurer $i$'s equilibrium retention incentive.

The insurance-claim volatility coefficients $v^{i}$ and $v^{i,0}$ measure the uncertainty in insurer $i$'s claim process. Since insurer $i$ is risk-averse (cf. \eqref{eq:Ji}), larger insurance-claim volatilities lower the incentive to retain risk, as reflected by the denominator in \eqref{eq:retention_incentive_equilibrium}. The risk sensitivity $1/\gamma^{i}$ scales these terms, so that a more risk-sensitive insurer reacts more strongly to the uncertainty in its claim process. The common insurance-claim volatility $v^{i,0}$, however, also enters the numerator through $\theta^{i}v^{i,0}\widehat{V}^{N}(p)/\gamma^{i}$, and can therefore increase the retention incentive.

The insurer population effect is captured by
\begin{equation*}
\frac{1}{\gamma^{i}}\theta^{i}v^{i,0}\widehat{V}^{N}(p)=\frac{1}{\gamma^{i}}\theta^{i}v^{i,0}\frac{1}{N}\sum_{j=1}^{N}\widehat{q}^{j}(p)v^{j,0}.
\end{equation*}
Thus, when insurer $i$ has relative performance concerns ($\theta^i>0$) and the insurer population retains common-insurance exposure ($\widehat V^N(p)>0$), this term is positive and raises insurer $i$'s equilibrium retention incentive. The effect is stronger when $1/\gamma^i$, $\theta^i$, $v^{i,0}$, or the aggregate common-insurance exposure $\widehat V^N(p)$ is larger. 

If $\theta^i=0$, the population-dependent term disappears, and insurer $i$'s retention response \eqref{eq:qhat_projection} reduces to
\begin{equation}
\label{eq:qhat_reduced}
\widehat{q}^{i}(p)=\Pi_{[0,1]}\Big(\frac{p-a^{i}}{\frac{1}{\gamma^{i}}\big[(v^{i,0})^{2}+(v^{i})^{2}\big]}\Big).
\end{equation}
When $\theta^i>0$ and $\widehat V^N(p)>0$, the response in \eqref{eq:qhat_projection} is no smaller than the response in \eqref{eq:qhat_reduced}. Thus, at the same premium level, relative performance concerns can increase insurer $i$'s equilibrium retention response.

So far, we have characterized the follower equilibrium response for each fixed premium level $p$. We now study how the equilibrium retention responses vary with $p$. This premium-response structure is needed for the leader's premium-optimization problem, since the reinsurer anticipates the follower equilibrium response when choosing the premium level. Because the follower investment responses $\widehat{\pi}^{i}$, $i=1,\ldots,N$, are independent of $p$, it remains only to characterize the premium-response structure of the equilibrium retention responses $\widehat{q}^{i}(p)$, $i=1,\ldots,N$.

\begin{proposition}[Premium response of equilibrium retentions]
\label{prop:premium_response}
For each insurer $i=1,\dots,N$, the map $p\mapsto \widehat{q}^{i}(p)$ is continuous and monotonically increasing on $[0,\infty)$.

Define the population-wide full-cession and full-retention thresholds by
\begin{equation*}
p^{\min}:=\max\{p\ge0:\widehat{q}^{i}(p)=0,\ i=1,\dots,N\},\qquad p^{\max}:=\min\{p\ge0:\widehat{q}^{i}(p)=1,\ i=1,\dots,N\}.
\end{equation*}
Then $p^{\min}<p^{\max}$, and for each insurer $i=1,\dots,N$,
\begin{equation*}
\widehat{q}^{i}(p)=0\quad\text{for }p\le p^{\min},\qquad\widehat{q}^{i}(p)=1\quad\text{for }p\ge p^{\max},\qquad i=1,\dots,N.
\end{equation*}
Define $\bar{V}^{N}=\frac{1}{N}\sum_{j=1}^{N}v^{j,0}$, then these thresholds are explicitly given by
\begin{equation}
\label{eq:pminmax}
p^{\min}=\min_{1\le j\le N}a^{j},\qquad p^{\max}=\max_{1\le j\le N}\Big\{ a^{j}+\frac{1}{\gamma^{j}}\big[(v^{j,0})^{2}+(v^{j})^{2}\bigl(1-\theta^{j}/N\bigr)]-\frac{1}{\gamma^{j}}\theta^{j}v^{j,0}\bar{V}^{N}\Big\}.
\end{equation}

For each insurer $i=1,\dots,N$, if we define the insurer-specific full-cession and full-retention thresholds by
\begin{equation}
\label{eq:premium_thresholds}
p^{i,\min}:=\max\{p\ge0:\widehat{q}^{i}(p)=0\},\qquad p^{i,\max}:=\min\{p\ge0:\widehat{q}^{i}(p)=1\},
\end{equation}
then $p^{i,\min},p^{i,\max}\in[p^{\min},p^{\max}]$ with $p^{i,\min}<p^{i,\max}$, and 
\begin{equation}
\label{eq:qhat_threshold_structure}
\widehat{q}^{i}(p)=0\quad\text{for }p\le p^{i,\min},\qquad0<\widehat{q}^{i}(p)<1\quad\text{for }p^{i,\min}<p<p^{i,\max},\qquad\widehat{q}^{i}(p)=1\quad\text{for }p\ge p^{i,\max}.
\end{equation}
\end{proposition}

\begin{proof}
See \ref{app:2.2}.
\end{proof}

Proposition~\ref{prop:premium_response} gives the basic premium-retention relation in the follower equilibrium response. As the reinsurance premium increases, each insurer moves progressively from full cession to partial retention and eventually to full retention. This is consistent with the usual economic intuition that more expensive reinsurance induces insurers to retain more risk. 

This monotonicity further reveals an implicit channel through which the premium level affects insurers' retention incentives. By \eqref{eq:Vhat_Sigmahat_def}, the monotonicity of $p\mapsto\widehat q^i(p)$ implies that $p\mapsto\widehat V^N(p)$ is also monotonically increasing. Hence, for insurers with $\theta^i>0$, the population-dependent term $\theta^{i}v^{i,0}\widehat V^N(p)/\gamma^i$ in the equilibrium retention incentive \eqref{eq:retention_incentive_equilibrium} becomes stronger as the premium rises. We call this effect \emph{common-insurance exposure spillover}: as the premium level increases, larger retention responses raise the average common-insurance exposure $\widehat V^N(p)$, which in turn strengthens the retention incentives of insurers with relative performance concerns. This contrasts with the effect of the premium-loss spread $p-a^i$ in \eqref{eq:retention_incentive_equilibrium}, which directly raises every insurer's incentive regardless of relative performance concerns.

We next show how this common-insurance exposure spillover appears at the level of insurer-specific thresholds. Without relative performance concerns ($\theta^{1}=\cdots=\theta^{N}=0$), insurer $i$'s full-cession and full-retention thresholds are $a^i$, insurer $i$'s expected claim rate and the actuarially fair reinsurance premium rate, and $a^i+\frac1{\gamma^i}\big[(v^{i,0})^2+(v^i)^2\big]$, corresponding to the risk-adjusted full-retention threshold. The next corollary compares these thresholds with the thresholds under relative performance concerns.

\begin{corollary}
\label{cor:relative_performance_threshold_shift}
For each insurer $i=1,\ldots,N$, the insurer-specific thresholds under relative performance concerns satisfy
\begin{equation}
\label{eq:threshold_shift}
p^{i,\min}=a^{i}-\frac{1}{\gamma^{i}}\theta^{i}v^{i,0}\widehat{V}^{N}(p^{i,\min}),\qquad p^{i,\max}=a^{i}+\frac{1}{\gamma^{i}}\big[(v^{i,0})^{2}+(v^{i})^{2}(1-\theta^{i}/N)\big]-\frac{1}{\gamma^{i}}\theta^{i}v^{i,0}\widehat{V}^{N}(p^{i,\max}).
\end{equation}
If $\theta^{i}=0$, then insurer $i$'s thresholds coincide with the no-relative-performance thresholds:
\begin{equation}
\label{eq:equality_threshold_shift}
p^{i,\min}=a^{i},\qquad p^{i,\max}=a^{i}+\frac{1}{\gamma^{i}}\big[(v^{i,0})^{2}+(v^{i})^{2}\big].
\end{equation}
If $\theta^{i}>0$, then
\begin{equation}
\label{eq:strict_upper_threshold_shift}
p^{i,\max}<a^{i}+\frac{1}{\gamma^{i}}\big[(v^{i,0})^{2}+(v^{i})^{2}\big].
\end{equation}
Moreover, if $\theta^{i}>0$ and $a^{i}>p^{\min}$, then
\begin{equation}
\label{eq:strict_lower_threshold_shift}
p^{i,\min}<a^{i}.
\end{equation}
If $a^{i}=p^{\min}$, then $p^{i,\min}=a^{i}$.
\end{corollary}

\begin{proof}
By \eqref{eq:qhat_projection} and \eqref{eq:retention_incentive_equilibrium}, together with the continuity of $p\mapsto \widehat q^i(p)$ and the threshold structure in Proposition~\ref{prop:premium_response}, the thresholds $p^{i,\min}$ and $p^{i,\max}$ are determined by
\begin{equation*}
\gamma^{i}(p^{i,\min}-a^{i})+\theta^{i}v^{i,0}\widehat{V}^{N}(p^{i,\min})=0,\qquad\gamma^{i}(p^{i,\max}-a^{i})+\theta^{i}v^{i,0}\widehat{V}^{N}(p^{i,\max})=(v^{i,0})^{2}+(v^{i})^{2}(1-\theta^{i}/N).
\end{equation*}
These equations give \eqref{eq:threshold_shift}. If $\theta^i=0$, then \eqref{eq:threshold_shift} becomes \eqref{eq:equality_threshold_shift}. If $\theta^i>0$, then the second identity in \eqref{eq:threshold_shift} yields \eqref{eq:strict_upper_threshold_shift}.

If $a^i=p^{\min}$, then Proposition~\ref{prop:premium_response} gives $\widehat q^i(p)=0$ for all $p\le a^i$, so $p^{i,\min}\ge a^i$. 
On the other hand, the first identity in \eqref{eq:threshold_shift} gives $p^{i,\min}\le a^i$. 
Therefore $p^{i,\min}=a^i$. 

If $a^i>p^{\min}$, then by the definition of $p^{\min}$, at least one insurer has already begun retaining risk at $p=a^i$, so $\widehat V^N(a^i)>0$. If further $\theta^i>0$, then insurer $i$'s retention incentive at $p=a^i$ is
\begin{equation*}
I^{i}(a^{i},\widehat{V}^{N}(a^{i}))=\frac{\theta^{i}v^{i,0}\widehat{V}^{N}(a^{i})}{(v^{i,0})^{2}+(v^{i})^{2}(1-\theta^{i}/N)}>0.
\end{equation*}
Hence $\widehat q^i(a^i)>0$, and this implies $p^{i,\min}<a^{i}$. Thus, \eqref{eq:strict_lower_threshold_shift} holds.
\end{proof}

Corollary~\ref{cor:relative_performance_threshold_shift} shows that relative performance concerns can lower the insurer-specific thresholds relative to the no-relative-performance case. If $\theta^{i}=0$, insurer~$i$'s thresholds are unchanged. If $\theta^{i}>0$, its full-retention threshold lies strictly below the risk-adjusted full-retention threshold from the no-relative-performance case, $a^{i}+\frac{1}{\gamma^{i}}[(v^{i,0})^{2}+(v^{i})^{2}]$; and, unless $a^{i}=p^{\min}$, its full-cession threshold lies strictly below $a^{i}$. Thus, relative performance concerns can induce partial retention while full cession is still actuarially favorable, and can also lead to earlier full retention. Because the threshold identities in \eqref{eq:threshold_shift} involve $\theta^i$, $\gamma^i$, $v^{i,0}$, $v^i$, and the values of $\widehat V^N$ at insurer~$i$'s own thresholds, the shifts vary across insurers. Hence, the order in which insurers enter partial or full retention need not match the ordering in the no-relative-performance case.

Regardless of how relative performance concerns shift or reorder the insurer-specific thresholds, Proposition~\ref{prop:premium_response} implies that the \emph{follower retention configuration}---the partition of insurers into full-cession, partial-retention, and full-retention groups---can change only at the thresholds $p^{i,\min}$ and $p^{i,\max}$, $i=1,\ldots,N$. Once these thresholds are ordered, they partition $[p^{\min},p^{\max}]$ into finitely many intervals, on each of which the follower retention configuration is fixed.

\begin{corollary}
\label{cor:premium_partition}
Let $p_0,\ldots,p_K$ be the ordered distinct elements of the insurer-specific thresholds:
\begin{equation}
\label{eq:premium_partition_points}
p^{\min}=p_{0}<p_{1}<\cdots<p_{K}=p^{\max},\qquad\{p_{0},p_{1},\ldots,p_{K}\}=\bigcup_{i=1}^{N}\{p^{i,\min},p^{i,\max}\}.
\end{equation}
Then $[p^{\min},p^{\max}]$ is partitioned by $p_0,\ldots,p_K$, with $K\le 2N-1$.

On each open interval $(p_{k-1},p_k)$, $k=1,\ldots,K$, the follower retention configuration is fixed. More precisely, there exists a retention configuration $(\mathcal C_k,\mathcal P_k,\mathcal R_k)$, where $\mathcal C_k$, $\mathcal P_k$, and $\mathcal R_k$ denote the sets of full-cession, partial-retention, and full-retention insurers, respectively, such that for every $p\in(p_{k-1},p_k)$,
\begin{equation*}
\widehat{q}^{i}(p)=0\quad\text{for }i\in\mathcal{C}_{k},\qquad0<\widehat{q}^{i}(p)<1\quad\text{for }i\in\mathcal{P}_{k},\qquad\widehat{q}^{i}(p)=1\quad\text{for }i\in\mathcal{R}_{k}.
\end{equation*}
\end{corollary}

\begin{proof}
By \eqref{eq:premium_thresholds} and \eqref{eq:qhat_threshold_structure},
\begin{equation*}
p^{\min}=\min_{1\le i\le N}p^{i,\min},\qquad p^{\max}=\max_{1\le i\le N}p^{i,\max}.
\end{equation*}
Hence, $p^{\min}$ and $p^{\max}$ are included among the insurer-specific thresholds. Between two consecutive ordered thresholds, no insurer crosses either its full-cession threshold or its full-retention threshold. Therefore, each insurer's retention status, and hence the follower retention configuration, is fixed on that open interval.
\end{proof}

Corollary~\ref{cor:premium_partition} gives the interval structure needed for the leader's problem. Relative performance concerns can lower these thresholds and alter their ordering across insurers, but once the insurer-specific thresholds are sorted, the premium interval $[p^{\min},p^{\max}]$ is divided into finitely many regions. On each open region $(p_{k-1},p_k)$, the set of full-cession, partial-retention, and full-retention insurers is fixed, which allows the leader's reduced premium objective to be written explicitly interval by interval.

\subsection{Leader's Optimization Problem and Constant Stackelberg--Nash Equilibrium}
\label{subsec:finite_leader}

We now solve the leader's optimization problem using the constant follower Nash equilibrium response characterized in Section~\ref{subsec:finite_followers}. Under the constant-control formulation in \eqref{def:constcontrol}, substituting this follower equilibrium response into the reinsurer's wealth dynamics \eqref{eq:dynamics_reinsurer} reduces the leader's stochastic control problem \eqref{eq:JL} to an equivalent static optimization problem in $(p,\pi^L)$. We first simplify the reinsurer's objective \eqref{eq:JL} under this follower equilibrium response, solve the resulting leader problem, and then construct constant Stackelberg--Nash equilibria.

For each premium level $p\ge0$, let $(\widehat q^1(p),\widehat\pi^1),\dots,(\widehat q^N(p),\widehat\pi^N)$ be the constant follower Nash equilibrium response from Proposition~\ref{prop:follower_response}. Under this follower equilibrium response and a constant leader investment control $\pi^L$, the reinsurer's terminal wealth follows from \eqref{eq:dynamics_reinsurer} as
\begin{align*}
X_{T}^{L} = &\ x^{L}+T\Big(\frac{1}{N}\sum_{i=1}^{N}(1-\widehat{q}^{i}(p))(p-a^{i})+\pi^{L}\mu^{L}\Big)
\nonumber\\
& + \frac{1}{N}\sum_{i=1}^{N}(1-\widehat{q}^{i}(p))v^{i,0}W_{T}^{0}+\frac{1}{N}\sum_{i=1}^{N}(1-\widehat{q}^{i}(p))v^{i}W_{T}^{i} + \pi^{L}\sigma^{L,0}B_{T}^{0}+\pi^{L}\sigma^{L}B_{T}^{L}.
\end{align*}
Since $X_T^L$ is Gaussian, maximizing the reinsurer's objective \eqref{eq:JL} is equivalent to maximizing
\begin{equation*}
\widetilde{J}^{L}(p,\pi^{L}):=\mathbb{E}[X_{T}^{L}]-\frac{1}{2\gamma^{L}}\mathrm{Var}(X_{T}^{L}),
\end{equation*}
and a direct computation yields
\begin{align}
\label{eq:Jtilde_leader}
\widetilde{J}^{L}(p,\pi^{L}) = &\ x^{L} + T\Big[\mu^{L}\pi^{L}-\frac{1}{2\gamma^{L}}\big((\sigma^{L,0})^{2}+(\sigma^{L})^{2}\big)(\pi^{L})^{2}\Big] + \frac{T}{N}\sum_{i=1}^{N}(1-\widehat{q}^{i}(p))(p-a^{i})
\nonumber\\
& -\frac{T}{2\gamma^{L}}\Big[\Big(\frac{1}{N}\sum_{i=1}^{N}(1-\widehat{q}^{i}(p))v^{i,0}\Big)^{2}+\frac{1}{N^{2}}\sum_{i=1}^{N}(1-\widehat{q}^{i}(p))^{2}(v^{i})^{2}\Big].
\end{align}
The expression in \eqref{eq:Jtilde_leader} separates the reinsurer's objective into an investment part and a premium part. The investment part is independent of the follower equilibrium responses, whereas the premium part depends on the follower equilibrium retention responses $\widehat q^1(p),\ldots,\widehat q^N(p)$. We therefore define the reinsurer's reduced premium objective by
\begin{equation}
\label{eq:leader_premium_objective}
\mathcal{J}^{L}(p):=\frac{1}{N}\sum_{i=1}^{N}(1-\widehat{q}^{i}(p))(p-a^{i})-\frac{1}{2\gamma^{L}}\Big[\Big(\frac{1}{N}\sum_{i=1}^{N}(1-\widehat{q}^{i}(p))v^{i,0}\Big)^{2}+\frac{1}{N^{2}}\sum_{i=1}^{N}(1-\widehat{q}^{i}(p))^{2}(v^{i})^{2}\Big].
\end{equation}

The next proposition identifies the leader's optimal controls and constructs constant Stackelberg--Nash equilibria.

\begin{proposition}[Constant Stackelberg--Nash equilibria]
\label{prop:leader_equilibrium}
The leader's optimal investment control is unique and is given by
\begin{equation}\label{eq:pistar_finite}
\pi^{L,*}=\frac{\gamma^{L}\mu^{L}}{(\sigma^{L,0})^{2}+(\sigma^{L})^{2}}.
\end{equation}
There exists at least one optimal premium level $p^*$ satisfying
\begin{equation}
\label{eq:pstar_argmax}
p^{*}\in\underset{p\in[p^{\min},p^{\max}]}{\arg\max}\,\mathcal{J}^{L}(p),
\end{equation}
where $p^{\min}$ and $p^{\max}$ are given in \eqref{eq:pminmax}, and $\mathcal J^L$ is defined in \eqref{eq:leader_premium_objective}. For any such $p^*$, define
\begin{equation*}
q^{i,*}:=\widehat{q}^{i}(p^{*}),\qquad\pi^{i,*}:=\widehat{\pi}^{i},\qquad i=1,\ldots,N.
\end{equation*}
Then $(p^*,\pi^{L,*},q^{1,*},\pi^{1,*},\ldots,q^{N,*},\pi^{N,*})$ is a constant Stackelberg--Nash equilibrium. The optimal premium $p^*$ may not be unique, so the constant Stackelberg--Nash equilibrium may not be unique.
\end{proposition}

\begin{proof}
See \ref{app:2.3}.
\end{proof}

Although the leader's premium is chosen from $[0,\infty)$, it is enough to maximize $\mathcal J^L(p)$ over $[p^{\min},p^{\max}]$. For $p\le p^{\min}$, all insurers choose full cession, so $\widehat q^i(p)=0$ for every $i$, and $\mathcal J^L(p)$ is increasing in $p$. Hence no premium below $p^{\min}$ can be optimal. For $p\ge p^{\max}$, all insurers choose full retention, so $\widehat q^i(p)=1$ for every $i$, and $\mathcal J^L(p)=0$. Thus, even if a premium above $p^{\max}$ also gives the maximal value, the same value is attained at $p^{\max}$. Therefore, an optimal premium can always be selected from the compact interval $[p^{\min},p^{\max}]$. This matches the economic interpretation of the two boundary regimes: premiums below $p^{\min}$ are too low for the reinsurer, while premiums above $p^{\max}$ are too high for insurers to cede risk.

It remains to make the maximization of $\mathcal J^L(p)$ explicit. The essential difficulty is that $\mathcal J^L(p)$ depends on $\widehat q^1(p),\ldots,\widehat q^N(p)$, and these equilibrium retention responses are defined implicitly through the fixed-point equation \eqref{eq:fixed_point_V} and the projection formula \eqref{eq:qhat_projection}. On each partition interval $(p_{k-1},p_k)$ from Corollary~\ref{cor:premium_partition}, the follower retention configuration $(\mathcal C_k,\mathcal P_k,\mathcal R_k)$ is fixed, and equilibrium retention responses and the reduced premium objective can be written explicitly. Hence
\begin{equation}
\label{eq:qhat_interval_configuration}
\widehat{q}^{i}(p)=0\quad\text{for }i\in\mathcal{C}_{k},\qquad\widehat{q}^{i}(p)=I^{i}(p,\widehat{V}^{N}(p))\quad\text{for }i\in\mathcal{P}_{k},\qquad\widehat{q}^{i}(p)=1\quad\text{for }i\in\mathcal{R}_{k}.
\end{equation}
Substituting \eqref{eq:qhat_interval_configuration} into the definition of $\widehat V^N(p)$ in \eqref{eq:Vhat_Sigmahat_def} gives
\begin{equation*}
\widehat{V}^{N}(p)=\frac{1}{N}\sum_{j\in\mathcal{P}_{k}}v^{j,0}I^{j}(p,\widehat{V}^{N}(p))+\frac{1}{N}\sum_{j\in\mathcal{R}_{k}}v^{j,0}.
\end{equation*}
Using the expression for $I^{j}(p,\widehat{V}^{N}(p))$ in \eqref{eq:retention_incentive_equilibrium}, this becomes a linear equation in $\widehat V^N(p)$. Solving this equation yields
\begin{equation}
\label{eq:Vhat_interval}
\widehat{V}^{N}(p)=\frac{\frac{1}{N}\sum_{j\in\mathcal{R}_{k}}v^{j,0}+\frac{1}{N}\sum_{j\in\mathcal{P}_{k}}\frac{\gamma^{j}v^{j,0}(p-a^{j})}{(v^{j,0})^{2}+(v^{j})^{2}(1-\theta^{j}/N)}}{1-\frac{1}{N}\sum_{j\in\mathcal{P}_{k}}\frac{\theta^{j}(v^{j,0})^{2}}{(v^{j,0})^{2}+(v^{j})^{2}(1-\theta^{j}/N)}},\qquad p\in(p_{k-1},p_k).
\end{equation}
Thus, the partial-retention responses $\widehat q^i(p)$, $i\in\mathcal P_k$, are affine in $p$ on this interval, and $\mathcal J^L(p)$ is an explicit quadratic function if $\mathcal P_k\ne\emptyset$, and linear with positive slope if $\mathcal P_k=\emptyset$.

The calculation above shows that $\mathcal J^L(p)$ can be optimized interval by interval once the thresholds $p_0,\ldots,p_K$ in \eqref{eq:premium_partition_points} and the associated follower retention configurations $(\mathcal C_k,\mathcal P_k,\mathcal R_k)$, $k=1,\ldots,K$, are determined. Although each follower retention configuration is fixed only on the open interval $(p_{k-1},p_k)$, Proposition~\ref{prop:premium_response} implies that the equilibrium retention responses, and hence $\mathcal J^L(p)$, are continuous in $p$. Therefore, maximizing the intervalwise expression for $\mathcal J^L(p)$ over $[p_{k-1},p_k]$ includes possible maximizers at the endpoints. We now present our threshold continuation procedure, which determines these thresholds and configurations sequentially and records one premium candidate from each closed interval $[p_{k-1},p_k]$ for \eqref{eq:pstar_argmax}.

\medskip
\noindent
\begin{enumerate}
\item \textbf{Initialize the first retention configuration.} Set $p_0=p^{\min}$ and $k=1$. Set
\begin{equation*}
\mathcal{C}_{1}:=\{i:a^{i}>p^{\min}\},\qquad\mathcal{P}_{1}:=\{i:a^{i}=p^{\min}\},\qquad\mathcal{R}_{1}:=\emptyset.
\end{equation*}
Thus, $(\mathcal C_1,\mathcal P_1,\mathcal R_1)$ is the follower retention configuration for premiums just above $p_0=p^{\min}$.

\item \textbf{Advance to the next threshold.} At the beginning of iteration $k$, the threshold $p_{k-1}$ and the follower retention configuration $(\mathcal C_k,\mathcal P_k,\mathcal R_k)$ have been determined.

\begin{enumerate}
\item \emph{Compute the intervalwise response.} Let $\widehat V_k^N(p)$ denote the expression in \eqref{eq:Vhat_interval} obtained from the configuration $(\mathcal C_k,\mathcal P_k,\mathcal R_k)$. Then \eqref{eq:qhat_interval_configuration} gives the corresponding intervalwise equilibrium retention responses. Moreover, $\widehat V_k^N(p)$ is affine and nondecreasing in $p$, so $p\mapsto I^i(p,\widehat V_k^N(p))$ is affine with strictly positive slope for each insurer $i$.

\item \emph{Locate the next threshold.} For each $i\in\mathcal C_k$, define the candidate full-cession threshold under the current configuration, denoted by $p_k^{i,\min}$, as the unique solution of
\begin{equation*}
I^i(p,\widehat V_k^N(p))=0.
\end{equation*}
For each $i\in\mathcal P_k$, define the candidate full-retention threshold under the current configuration, denoted by $p_k^{i,\max}$, as the unique solution of
\begin{equation*}
I^i(p,\widehat V_k^N(p))=1.
\end{equation*}
The next threshold is then obtained by
\begin{equation*}
p_{k}=\min\Big\{ p^{\max},\,\min_{i\in\mathcal{C}_{k}}p_{k}^{i,\min},\,\min_{i\in\mathcal{P}_{k}}p_{k}^{i,\max}\Big\},
\end{equation*}
with the convention that $\min\emptyset = +\infty$.

\item \emph{Optimize on the current interval.} Substituting the intervalwise equilibrium retention responses obtained in Step~2(a) into \eqref{eq:leader_premium_objective} gives the intervalwise expression for $\mathcal J^L(p)$. Choose
\begin{equation*}
p_k^*:=\underset{p\in[p_{k-1},p_k]}{\arg\max}\,\mathcal J^L(p),
\end{equation*}
and record it as the premium candidate associated with $[p_{k-1},p_k]$.

\item \emph{Update the retention configuration.} If $p_k=p^{\max}$, proceed to Step 3. Otherwise, define
\begin{equation*}
E_{k}^{\min}:=\{i\in\mathcal{C}_{k}:p_{k}^{i,\min}=p_{k}\},\qquad E_{k}^{\max}:=\{i\in\mathcal{P}_{k}:p_{k}^{i,\max}=p_{k}\}.
\end{equation*}
The sets $E_k^{\min}$ and $E_k^{\max}$ collect the insurers that, as the premium crosses $p_k$, move from full cession to partial retention and from partial retention to full retention, respectively. The next follower retention configuration is then obtained by
\begin{equation*}
\mathcal{C}_{k+1}=\mathcal{C}_{k}\setminus E_{k}^{\min},\qquad\mathcal{P}_{k+1}=(\mathcal{P}_{k}\cup E_{k}^{\min})\setminus E_{k}^{\max},\qquad\mathcal{R}_{k+1}=\mathcal{R}_{k}\cup E_{k}^{\max}.
\end{equation*}
Increase $k$ by one and return to Step~2(a).
\end{enumerate}

\item \textbf{Select the optimal premium.} Compare the values of $\mathcal J^L$ at the intervalwise premium candidates $p_1^*,\ldots,p_K^*$. Any candidate with the largest value is an optimal premium level $p^*$ in \eqref{eq:pstar_argmax}.
\end{enumerate}

The threshold continuation procedure makes the problem \eqref{eq:pstar_argmax} computationally tractable. A naive approach would have to identify the follower retention configuration at each $p\in[p^{\min},p^{\max}]$, potentially by checking up to $3^N$ assignments of insurers to full-cession, partial-retention, and full-retention groups, and then evaluate $\mathcal J^L(p)$ under the selected configuration. Since $p$ ranges over a continuum, this is not a tractable premium-optimization procedure. By contrast, the continuation procedure determines at most $2N-1$ partition intervals (see Corollary~\ref{cor:premium_partition}), updates the follower retention configuration only when a threshold is reached, and obtains one premium candidate on each closed interval $[p_{k-1},p_k]$ by maximizing a linear or strictly concave quadratic expression for $\mathcal J^L(p)$. Hence, neither enumeration of possible retention configurations nor continuous search over premium levels is required, and the total computational complexity is $O(N^2)$.

Once an optimal premium level $p^*$ (which need not be unique) is selected from the procedure, Proposition~\ref{prop:leader_equilibrium} gives the corresponding constant Stackelberg--Nash equilibrium by setting $q^{i,*}=\widehat q^i(p^*)$ and $\pi^{i,*}=\widehat\pi^i$, $i=1,\ldots,N$, together with the leader's optimal investment control $\pi^{L,*}$. This completes the characterization of finite-player equilibria and provides the benchmark for the mean field formulation in the next section.

\section{Stackelberg Mean Field Game}
\label{sec:smfg}
\subsection{Mean Field Formulation and Equilibrium Concept}
\label{subsec:smfg_formulation}

We now formulate a mean field counterpart of the finite-player Stackelberg--Nash game studied in Section~\ref{sec:finite_game}. The market consists of one reinsurer and a continuum of infinitesimal heterogeneous insurers. As in the finite-player game, the reinsurer acts as the Stackelberg leader and first announces a common reinsurance premium rate together with its own investment control. Given this announcement, the insurers choose their retention and investment controls through a follower mean
field game. The main difference from Section~\ref{sec:finite_game} is that finite averages over insurers are replaced by expectations with respect to the type distribution.

Heterogeneity across insurers is represented by a random type vector $\zeta$ 
\begin{equation*}
\zeta=(x,\eta,a,\gamma,\theta,v^{0},v,\mu,\sigma^{0},\sigma),
\end{equation*}
taking values in a compact set $\mathcal Z$:
\begin{equation*}
\mathcal Z \subset
\mathbb{R}\times(0,\infty)\times(0,\infty)\times(0,\infty)\times[0,1]\times(0,\infty)\times(0,\infty)\times\mathbb{R}\times(0,\infty)\times(0,\infty).
\end{equation*}
Here $x$ is the initial wealth, $\eta$ is the safety loading, $a$ is the expected claim rate, $\gamma$ is the risk-tolerance parameter, $\theta$ is the relative performance parameter, $v^{0}$ and $v$ are the common and idiosyncratic insurance-claim volatility coefficients, and $\mu$, $\sigma^{0}$, and $\sigma$ are the financial-market coefficients. The law of $\zeta$ describes the distribution of insurer types in the population.

Let $(\Omega,\mathcal F,\mathbb P)$ support the random type vector $\zeta$ and mutually independent standard Brownian motions $W^{0}$, $W$, $B^{0}$, $B$, and $B^{L}$. We assume that $\zeta$ is independent of these Brownian motions. Here, $W^{0}$ represents the common insurance-claim noise affecting the insurer population, while $W$ represents the idiosyncratic insurance-claim noise faced by the representative insurer. The process $B^{0}$ represents the common financial-market noise shared by the insurer population and the reinsurer, while $B$ and $B^{L}$ represent the idiosyncratic financial-market noises faced by the representative insurer and the reinsurer, respectively. We write $\mathbb F^{0}=(\mathcal F_t^0)_{0\le t\le T}$ for the filtration generated by the common noises $(W^{0},B^{0})$, and $\mathbb F^{\mathrm{MF}}=(\mathcal F_t^{\mathrm{MF}})_{0\le t\le T}$ for the smallest filtration satisfying the usual conditions such that $\zeta$ is $\mathcal F_0^{\mathrm{MF}}$-measurable and $W^0$, $W$, $B^0$, $B$, and $B^L$ are adapted.

As in the finite-player analysis, we focus on the constant-control formulation. The reinsurer's controls consist of a common reinsurance premium rate and its investment control:
\begin{equation*}
p_{t}\equiv p\ge0,\qquad\pi_{t}^{L}\equiv\pi^{L}\in\mathbb{R},
\end{equation*}
whereas the representative insurer's controls are its retention level and investment control:
\begin{equation*}
q_{t}^{\zeta}\equiv q^{\zeta}\in[0,1],\qquad\pi_{t}^{\zeta}\equiv\pi^{\zeta}\in\mathbb{R}.
\end{equation*}
The reinsurer's controls $p$ and $\pi^L$ are deterministic constants. The representative insurer's controls $q^\zeta$ and $\pi^\zeta$ are constant in time, but may depend on the random type vector $\zeta$. We define the admissible set for the representative insurer's constant controls by
\begin{equation*}
\mathcal{A}^{\mathrm{MF}}:=\big\{ (q^{\zeta},\pi^{\zeta}):q^{\zeta},\pi^{\zeta}\text{ are }\sigma(\zeta)\text{-measurable},\ q^{\zeta}\in[0,1]\ \text{a.s.},\ \mathbb{E}[(\pi^{\zeta})^{2}]<\infty\big\}.
\end{equation*}

As in the finite-player formulation, insurers manage their wealth by ceding part of their insurance risk through proportional reinsurance and by investing in a risky asset. The reinsurer's insurance surplus is generated by the ceded insurance business of the insurer population, and the reinsurer also invests in its own risky asset. 

In the followers' mean field formulation, we describe the insurers' retention-investment problem through a representative insurer with random type $\zeta$. With control $(q^\zeta,\pi^\zeta)\in\mathcal A^{\mathrm{MF}}$, we model its wealth process by
\begin{equation}
\label{eq:dynamics_insurer_MF}
dX_{t}^{\zeta}=\big(\eta a-(1-q^{\zeta})(p-a)+\pi^{\zeta}\mu\big)\,dt+q^{\zeta}v^{0}\,dW_{t}^{0}+q^{\zeta}v\,dW_{t}+\pi^{\zeta}\sigma^{0}\,dB_{t}^{0}+\pi^{\zeta}\sigma\,dB_{t},\quad X_{0}^{\zeta}=x,
\end{equation}
where $(p,\pi^{L})\in[0,\infty)\times\mathbb{R}$ is the reinsurer's control announced to the insurer population. The reinsurer's wealth process is then modeled by
\begin{equation}
\label{eq:dynamics_reinsurer_MF}
dX_{t}^{L,\mathrm{MF}}=\big(\mathbb{E}[(1-q^{\zeta})(p-a)]+\pi^{L}\mu^{L}\big)\,dt+\mathbb{E}[(1-q^{\zeta})v^{0}]\,dW_{t}^{0}+\pi^{L}\sigma^{L,0}\,dB_{t}^{0}+\pi^{L}\sigma^{L}\,dB_{t}^{L},\quad X_{0}^{L,\mathrm{MF}}=x^{L}.
\end{equation}

\begin{remark}
The reinsurer's mean field wealth dynamics \eqref{eq:dynamics_reinsurer_MF} is the continuum population analogue of \eqref{eq:dynamics_reinsurer}. The expectations $\mathbb{E}[(1-q^{\zeta})(p-a)]$ and $\mathbb{E}[(1-q^{\zeta})v^{0}]$ correspond to the $1/N$-scaled finite-player aggregates in \eqref{eq:dynamics_reinsurer}, while the terms involving the idiosyncratic insurance-claim noises $W^{i}$ diversify away in the continuum limit. Since $q^{\zeta}$ is $\sigma(\zeta)$-measurable and $\zeta$ is independent of the common noises, the conditional expectations given $\mathcal{F}_{T}^{0}$, which usually appear in the standard framework of MFG with common noise, reduce to unconditional ones here. 
\end{remark}

We now define the players' objective functions. As in the finite-player formulation, insurers evaluate terminal wealth relative to the insurer population average. In the mean field formulation, each insurer is infinitesimal relative to the continuum population, so this population average is treated as given in the representative insurer's problem. We represent this given candidate population average terminal wealth by an $\mathcal F_T^0$-measurable random variable $\overline X_T$. Accordingly, given $\overline X_T$, the representative insurer's problem is to maximize
\begin{equation}
\label{eq:J_insurer_MF}
J^{\mathrm{MF}}(q^{\zeta},\pi^{\zeta};\overline{X}_{T},p, \pi^{L}):=\mathbb{E}\Big[-\exp\Big(-\frac{1}{\gamma}\big(X_{T}^{\zeta}-\theta\overline{X}_{T}\big)\Big)\Big]
\end{equation}
over $(q^\zeta,\pi^\zeta)\in\mathcal A^{\mathrm{MF}}$, where $X^\zeta$ solves \eqref{eq:dynamics_insurer_MF} with the control pair $(q^\zeta,\pi^\zeta)$. The reinsurer's problem is to maximize
\begin{equation}
\label{eq:J_reinsurer_MF}
J^{L,\mathrm{MF}}(p,\pi^{L};q^{\zeta},\pi^{\zeta}):=\mathbb{E}\Big[-\exp\Big(-\frac{1}{\gamma^{L}}X_{T}^{L,\mathrm{MF}}\Big)\Big]
\end{equation}
over $(p,\pi^L)\in[0,\infty)\times\mathbb R$, where $X^{L,\mathrm{MF}}$ solves \eqref{eq:dynamics_reinsurer_MF}.

\begin{remark}
\label{rem:typewise_MF_objective}
The representative insurer's objective \eqref{eq:J_insurer_MF} is formulated as an unconditional expectation over the random type $\zeta$. As in \cite[Section 2.2.2]{LackerZariphopoulos2019}, this admits an equivalent type-by-type interpretation: for a fixed $\overline X_T$, maximizing \eqref{eq:J_insurer_MF} over $\sigma(\zeta)$-measurable controls is equivalent to maximizing, for almost every type $\zeta$, the typewise conditional objective
\begin{equation}
\label{eq:J_insurer_conditional_MF}
J^{\zeta}(q^{\zeta},\pi^{\zeta};\overline{X}_{T},p,\pi^{L}):=\mathbb{E}\Big[-\exp\Big(-\frac{1}{\gamma}(X_{T}^{\zeta}-\theta\overline{X}_{T})\Big)\,\Big|\,\zeta\Big].
\end{equation}
\end{remark}

The order of decisions is the same as in Section~\ref{sec:finite_game}. The reinsurer first announces $(p,\pi^L)$. Given this leader control, the insurer population plays a follower mean field game: for a candidate population average $\overline X_T$, the representative insurer solves \eqref{eq:J_insurer_MF}, and the candidate average must coincide with the average generated by the resulting control. The reinsurer anticipates this response when solving its own problem~\eqref{eq:J_reinsurer_MF}.

We now formalize the follower-side mean field equilibrium response and the corresponding Stackelberg mean field equilibrium.

\begin{definition}[Constant Stackelberg mean field equilibrium]
\label{def:CSMFE}
Fix a constant leader control $(p,\pi^{L})\in[0,\infty)\times\mathbb{R}$. A pair $(q^{\zeta},\pi^{\zeta})\in\mathcal{A}^{\mathrm{MF}}$ is called a \emph{constant follower mean field equilibrium (MFE) response to $(p,\pi^{L})$} if there exists an $\mathcal{F}_{T}^{0}$-measurable random variable $\overline{X}_{T}$ such that, given $\overline{X}_{T}$, the control pair $(q^{\zeta},\pi^{\zeta})$ maximizes the representative insurer's objective \eqref{eq:J_insurer_MF} over $\mathcal{A}^{\mathrm{MF}}$, and $\overline{X}_{T}$ coincides with the population average generated by $(q^{\zeta},\pi^{\zeta})$:
\begin{equation}
\label{eq:MF_consistency_condition}
\overline{X}_{T}=\mathbb{E}[X_{T}^{\zeta}\mid\mathcal{F}_{T}^{0}],
\end{equation}
where $X^{\zeta}$ solves \eqref{eq:dynamics_insurer_MF} with the control pair $(q^{\zeta},\pi^{\zeta})$.

A constant control profile $(p^{\mathrm{MF},*},\pi^{L,\mathrm{MF},*},q^{\zeta,*},\pi^{\zeta,*})$ is called a \emph{constant Stackelberg mean field equilibrium} if $(q^{\zeta,*},\pi^{\zeta,*})$ is a constant follower mean field equilibrium response to $(p^{\mathrm{MF},*},\pi^{L,\mathrm{MF},*})$, and the leader control $(p^{\mathrm{MF},*},\pi^{L,\mathrm{MF},*})$ is optimal for the reinsurer's objective \eqref{eq:J_reinsurer_MF} given this follower equilibrium response.
\end{definition}

\subsection{Follower Mean Field Equilibrium Response}
\label{subsec:smfg_followers}

We next characterize the follower mean field equilibrium for a fixed premium $p$. The leader's investment $\pi^L$ does not enter the representative insurer's wealth dynamics~\eqref{eq:dynamics_insurer_MF} or objective \eqref{eq:J_insurer_MF}, so the follower response depends only on $p$.

For a constant follower control $(q^\zeta,\pi^\zeta)\in\mathcal A^{\mathrm{MF}}$, the independence of $\zeta$, $W$, and $B$ from the common noises, together with the $\sigma(\zeta)$-measurability of $(q^\zeta,\pi^\zeta)$, implies that the population average mandated by the consistency condition \eqref{eq:MF_consistency_condition} takes the affine common-noise form
\begin{equation}
\label{eq:MF_average_affine_form}    
\overline{X}_{T}=\mathbb{E}[X_{T}^{\zeta}\mid\mathcal{F}_{T}^{0}]=\mathbb{E}[x]+T\big(\mathbb{E}[\eta a]-\mathbb{E}[(1-q^{\zeta})(p-a)]+\mathbb{E}[\pi^{\zeta}\mu]\big)+\mathbb{E}[q^{\zeta}v^{0}]W_{T}^{0}+\mathbb{E}[\pi^{\zeta}\sigma^{0}]B_{T}^{0}.
\end{equation}
Hence, although Definition~\ref{def:CSMFE} requires $\overline X_T$ to be $\mathcal F_T^0$-measurable, the consistency condition restricts attention to candidates of the affine form
\begin{equation*}
\overline{X}_{T}=C+VW_{T}^{0}+\Sigma B_{T}^{0},\qquad C,\,V,\,\Sigma\in\mathbb{R},
\end{equation*}
whose coefficients $C, V, \Sigma$ are determined by the controls $(q^\zeta,\pi^\zeta)$ through the expectations in \eqref{eq:MF_average_affine_form}.

By Remark~\ref{rem:typewise_MF_objective}, the maximization of \eqref{eq:J_insurer_MF} reduces to a type-by-type maximization through the conditional objective \eqref{eq:J_insurer_conditional_MF}. For a candidate $\overline X_T$ of the affine form above, conditional on $\zeta$, the random variable $X_T^\zeta-\theta\overline X_T$ is Gaussian, so this in turn is equivalent, for almost every type $\zeta$, to maximizing the mean--variance criterion
\begin{equation*}
\widetilde{J}^{\zeta}(q^{\zeta},\pi^{\zeta};\overline X_T):=\mathbb{E}[X_{T}^{\zeta}-\theta\overline{X}_{T}\mid\zeta]-\frac{1}{2\gamma}\mathrm{Var}(X_{T}^{\zeta}-\theta\overline{X}_{T}\mid\zeta).
\end{equation*}
A direct computation using \eqref{eq:dynamics_insurer_MF} and the candidate affine form shows that
\begin{align}
\label{eq:Jtilde_MF_expanded}
\widetilde{J}^{\zeta}(q^{\zeta},\pi^{\zeta};\overline X_T) = & \ \mathrm{Const.}+T\Big[(p-a)q^{\zeta}+\frac{1}{\gamma}\theta v^{0}V q^{\zeta}-\frac{1}{2\gamma}\big((v^{0})^{2}+v^{2}\big)(q^{\zeta})^{2}\Big]
\nonumber\\
& + T\Big[\Big(\mu+\frac{1}{\gamma}\theta\sigma^{0}\Sigma\Big)\pi^{\zeta}-\frac{1}{2\gamma}\big((\sigma^{0})^{2}+\sigma^{2}\big)(\pi^{\zeta})^{2}\Big],
\end{align}
where ``Const.'' denotes terms independent of the control values $(q^\zeta,\pi^\zeta)$.

The reduced criterion \eqref{eq:Jtilde_MF_expanded} separates into a retention term in $q^\zeta$ and an investment term in $\pi^\zeta$, with the premium level $p$ appearing only in the retention term. The two scalars $V$ and $\Sigma$ are treated as fixed in the typewise maximization and are then determined by the mean field consistency conditions. 

\begin{proposition}[Follower MFE for a fixed premium]
\label{prop:CMFER}
Fix $p\ge0$. Then the follower-side game admits a unique constant follower mean field equilibrium response $(\widehat q^\zeta(p),\widehat\pi^\zeta)$. Define the average common-insurance and common-financial exposures associated with this equilibrium response by
\begin{equation}
\label{eq:Vhat_Sigmahat_MF}
\widehat{V}(p):=\mathbb{E}[\widehat{q}^{\zeta}(p)v^{0}],\qquad\widehat{\Sigma}:=\mathbb{E}[\widehat{\pi}^{\zeta}\sigma^{0}],
\end{equation}
and define
\begin{equation}
\label{eq:retention_incentive_MF}
I^{\zeta}(p,V):=\frac{(p-a)+\frac{1}{\gamma}\theta v^{0}V}{\frac{1}{\gamma}\big((v^{0})^{2}+v^{2}\big)},\qquad p\ge0,\ V\in\mathbb{R}.
\end{equation}
Then
\begin{equation}
\label{eq:pihat_MF}
\widehat{\pi}^{\zeta}=\frac{\mu+\frac{1}{\gamma}\theta\sigma^{0}\widehat{\Sigma}}{\frac{1}{\gamma}\big((\sigma^{0})^{2}+\sigma^{2}\big)},\qquad\widehat{\Sigma}=\frac{\mathbb{E}\big[\gamma\mu\sigma^{0}/((\sigma^{0})^{2}+\sigma^{2})\big]}{1-\mathbb{E}\big[\theta(\sigma^{0})^{2}/((\sigma^{0})^{2}+\sigma^{2})\big]},
\end{equation}
and
\begin{equation}
\label{eq:qhat_MF}
\widehat{q}^{\zeta}(p)=\Pi_{[0,1]}\big(I^{\zeta}(p,\widehat{V}(p))\big).
\end{equation}
The average common-insurance exposure $\widehat V(p)$ is determined as the unique solution of the fixed-point equation for $V$:
\begin{equation}
\label{eq:fixed_point_V_MF}
V=\mathbb{E}\big[v^{0}\,\Pi_{[0,1]}\big(I^{\zeta}(p,V)\big)\big].
\end{equation}
\end{proposition}

\begin{proof}
See \ref{app:3.1}.
\end{proof}

Proposition~\ref{prop:CMFER} reduces the follower MFG to the scalar fixed point~\eqref{eq:fixed_point_V_MF} for $\widehat V(p)$.  Once $\widehat V(p)$ is determined, the equilibrium retention response is given typewise by the projection formula \eqref{eq:qhat_MF}, while the equilibrium investment is explicit and independent of $p$.

For each type $\zeta$, the equilibrium retention incentive is 
\begin{equation}
\label{eq:retention_incentive_equilibrium_MF}
I^\zeta(p,\widehat V(p))=\frac{(p-a)+\frac{1}{\gamma}\theta v^0\widehat V(p)}{\frac{1}{\gamma}\big((v^0)^2+v^2\big)}.
\end{equation}
As in the finite-player expression \eqref{eq:retention_incentive_equilibrium}: a nonpositive value gives full cession, a value between $0$ and $1$ gives partial retention, and a value at least $1$ gives full retention (cf. \eqref{eq:qhat_MF}). The term $p-a$ is the premium-loss spread, while $\theta v^0\widehat V(p)/\gamma$ captures the effect of relative performance concerns through the population's average common-insurance exposure.   If $\theta=0$, this population term disappears and the response reduces to
\begin{equation}
\label{eq:qhat_MF_no_relative_performance}
\widehat q^\zeta(p)=\Pi_{[0,1]}\Big(\frac{p-a}{\frac{1}{\gamma}\big((v^0)^2+v^2\big)}\Big).
\end{equation}
For types with $\theta>0$ and $\widehat V(p)>0$, the response in \eqref{eq:qhat_MF} is no smaller than the response in \eqref{eq:qhat_MF_no_relative_performance}, so relative performance concerns can increase the equilibrium retention response through the term $\theta v^0\widehat V(p)/\gamma$.

We now study how the follower-side mean field equilibrium response varies with the premium level. Since the equilibrium investment response $\widehat\pi^\zeta$ is independent of $p$, it remains to characterize the premium-response structure of $\widehat q^\zeta(p)$ and $\widehat V(p)$.

\begin{proposition}[Premium response of mean field retentions]
\label{prop:premium_response_CMFER}
For almost every type $\zeta$, the map $p\mapsto\widehat q^\zeta(p)$ is continuous and monotonically increasing on $[0,\infty)$.

Define the population-wide full-cession and full-retention thresholds by
\begin{equation*}
p^{\mathrm{MF},\min}:=\sup\{p\ge0:\widehat{q}^{\zeta}(p)=0\text{ for a.e. }\zeta\},\qquad p^{\mathrm{MF},\max}:=\inf\{p\ge0:\widehat{q}^{\zeta}(p)=1\text{ for a.e. }\zeta\}.
\end{equation*}
Then $p^{\mathrm{MF},\min}<p^{\mathrm{MF},\max}$, and for every $p\le p^{\mathrm{MF},\min}$ and almost every type $\zeta$, $\widehat q^\zeta(p)=0$, and for every $p\ge p^{\mathrm{MF},\max}$ and almost every type $\zeta$, $\widehat q^\zeta(p)=1$. These thresholds are explicitly given by
\begin{equation}
\label{eq:pminmax_MF}
p^{\mathrm{MF},\min}=\mathrm{ess\,inf}\,a,\qquad p^{\mathrm{MF},\max}=\mathrm{ess\,sup}\Big(a+\frac{1}{\gamma}((v^{0})^{2}+v^{2})-\frac{1}{\gamma}\theta v^{0}\bar{V}\Big),
\end{equation}
where $\bar{V}:=\mathbb E[v^0]$.

Consequently, for almost every type $\zeta$, if we define the type-dependent full-cession and full-retention thresholds by
\begin{equation*}
p^{\zeta,\min}:=\sup\{p\ge0:\widehat q^\zeta(p)=0\}, \qquad p^{\zeta,\max}:=\inf\{p\ge0:\widehat q^\zeta(p)=1\},
\end{equation*}
then these are $\sigma(\zeta)$-measurable random variables in $[p^{\mathrm{MF},\min},p^{\mathrm{MF},\max}]$ with $p^{\zeta,\min}< p^{\zeta,\max}$ a.s., and for almost every type $\zeta$, 
\begin{equation*}
\widehat q^\zeta(p)=0\quad\text{for }p\le p^{\zeta,\min}, \quad 0<\widehat q^\zeta(p)<1\quad\text{for }p^{\zeta,\min}<p<p^{\zeta,\max}, \quad \widehat q^\zeta(p)=1\quad\text{for }p\ge p^{\zeta,\max}.
\end{equation*}
\end{proposition}

\begin{proof}
See \ref{app:3.2}.
\end{proof}

Since insurer heterogeneity is represented by the random type $\zeta$, the population-wide thresholds $p^{\mathrm{MF},\min}$ and $p^{\mathrm{MF},\max}$ are deterministic but expressed through the essential infimum and supremum, while the type-dependent thresholds $p^{\zeta,\min}$ and $p^{\zeta,\max}$ are $\sigma(\zeta)$-measurable random variables.

Proposition~\ref{prop:premium_response_CMFER} establishes the monotonicity of $p\mapsto\widehat q^\zeta(p)$, and this implies that $p\mapsto\widehat V(p)=\mathbb E[\widehat q^\zeta(p)v^0]$ is also monotonically increasing. Hence, for types with $\theta>0$, the population-dependent term $\theta v^0\widehat V(p)/\gamma$ in the equilibrium retention incentive \eqref{eq:retention_incentive_equilibrium_MF} becomes stronger as the premium rises, providing the common-insurance exposure spillover through $\widehat V(p)$. This reinforcement is absent for types with $\theta=0$. As in the finite-player setting (Corollary~\ref{cor:relative_performance_threshold_shift}), this spillover lowers the type-dependent thresholds; the next corollary records the shift.

\begin{corollary}
\label{cor:MF_threshold_shift}
For almost every type $\zeta$, the type-dependent thresholds satisfy
\begin{equation}
\label{eq:MF_threshold_shift}
p^{\zeta,\min}=a-\frac{1}{\gamma}\theta v^{0}\widehat{V}(p^{\zeta,\min}),\qquad p^{\zeta,\max}=a+\frac{1}{\gamma}((v^{0})^{2}+v^{2})-\frac{1}{\gamma}\theta v^{0}\widehat{V}(p^{\zeta,\max}).
\end{equation}
For types with $\theta=0$,
\begin{equation}
\label{eq:MF_threshold_no_relative_performance}
p^{\zeta,\min}=a,\qquad p^{\zeta,\max}=a+\frac{1}{\gamma}((v^{0})^{2}+v^{2}).
\end{equation}
\end{corollary}
\begin{proof}
By \eqref{eq:qhat_MF}, \eqref{eq:retention_incentive_equilibrium_MF}, and the threshold structure in Proposition~\ref{prop:premium_response_CMFER}, the lower and upper thresholds are determined by
\begin{equation*}
I^{\zeta}(p^{\zeta,\min},\widehat{V}(p^{\zeta,\min}))=0,\qquad I^{\zeta}(p^{\zeta,\max},\widehat{V}(p^{\zeta,\max}))=1.
\end{equation*}
Substituting \eqref{eq:retention_incentive_MF} gives \eqref{eq:MF_threshold_shift}. For types with $\theta=0$, \eqref{eq:MF_threshold_shift} reduces to \eqref{eq:MF_threshold_no_relative_performance}.
\end{proof}

In the no-relative-performance case \eqref{eq:MF_threshold_no_relative_performance}, the lower threshold is the actuarially fair level $a$ and the upper threshold is the risk-adjusted level $a+((v^0)^2+v^2)/\gamma$. Under relative performance concerns, both thresholds are reduced by $\theta v^0\widehat V(\cdot)/\gamma$, so types with $\theta>0$ can begin partial retention and reach full retention at premiums below these no-relative-performance thresholds.

\subsection{Leader's Optimization Problem and Constant Stackelberg Mean Field Equilibrium}
\label{subsec:smfg_leader}

We now turn to the leader's problem. Substituting the follower MFE from Proposition~\ref{prop:CMFER} into the reinsurer's wealth dynamics \eqref{eq:dynamics_reinsurer_MF} reduces the leader's problem \eqref{eq:J_reinsurer_MF} to a static mean--variance optimization in $(p,\pi^L)$. The investment part is quadratic and can be solved explicitly, while the premium part
depends on the follower retention response $\widehat q^\zeta(p)$.

Using Proposition~\ref{prop:CMFER} with a constant leader investment control $\pi^L$, the reinsurer's terminal wealth in \eqref{eq:dynamics_reinsurer_MF} becomes
\begin{equation*}
X_{T}^{L,\mathrm{MF}}=x^{L}+T\big(\mathbb{E}[(1-\widehat{q}^{\zeta}(p))(p-a)]+\pi^{L}\mu^{L}\big)+\mathbb{E}[(1-\widehat{q}^{\zeta}(p))v^{0}]\,W_{T}^{0}+\pi^{L}\sigma^{L,0}B_{T}^{0}+\pi^{L}\sigma^{L}B_{T}^{L}.
\end{equation*}
Since $X_T^{L,\mathrm{MF}}$ is Gaussian, the mean--variance reduction gives that maximizing \eqref{eq:J_reinsurer_MF} is equivalent to maximizing
\begin{equation*}
\widetilde J^{L,\mathrm{MF}}(p,\pi^{L}):=\mathbb{E}[X_{T}^{L,\mathrm{MF}}]-\frac{1}{2\gamma^{L}}\mathrm{Var}(X_{T}^{L,\mathrm{MF}}),
\end{equation*}
and a direct computation yields
\begin{align}
\label{eq:Jtilde_leader_MF}
\widetilde{J}^{L,\mathrm{MF}}(p,\pi^{L}) = & \ x^{L}+T\Big[\mu^{L}\pi^{L}-\frac{1}{2\gamma^{L}}\big((\sigma^{L,0})^{2}+(\sigma^{L})^{2}\big)(\pi^{L})^{2}\Big]
\nonumber\\
& + T\Big[\mathbb{E}[(1-\widehat{q}^{\zeta}(p))(p-a)]-\frac{1}{2\gamma^{L}}\big(\mathbb{E}[(1-\widehat{q}^{\zeta}(p))v^{0}]\big)^{2}\Big].
\end{align}
The expression in \eqref{eq:Jtilde_leader_MF} separates the reinsurer's objective into an investment part and a premium part, with the premium part depending on the follower retention response $\widehat q^\zeta(p)$. We accordingly define the reinsurer's reduced premium objective by
\begin{equation}
\label{eq:leader_premium_objective_MF}
\mathcal{J}^{L,\mathrm{MF}}(p):=\mathbb{E}[(1-\widehat{q}^{\zeta}(p))(p-a)]-\frac{1}{2\gamma^{L}}\big(\mathbb{E}[(1-\widehat{q}^{\zeta}(p))v^{0}]\big)^{2}.
\end{equation}

\begin{proposition}[Constant Stackelberg mean field equilibria]
\label{prop:CSMFE_existence}
The leader's optimal investment control is unique and is given by
\begin{equation}\label{eq:pistar_mf}
\pi^{L,\mathrm{MF},*}=\frac{\gamma^{L}\mu^{L}}{(\sigma^{L,0})^{2}+(\sigma^{L})^{2}}.
\end{equation}
There exists at least one optimal premium level $p^{\mathrm{MF},*}$ satisfying
\begin{equation}
\label{eq:pstar_argmax_MF}
p^{\mathrm{MF},*}\in\underset{p\in[p^{\mathrm{MF},\min},\,p^{\mathrm{MF},\max}]}{\arg\max}\,\mathcal{J}^{L,\mathrm{MF}}(p),
\end{equation}
where $p^{\mathrm{MF},\min}$ and $p^{\mathrm{MF},\max}$ are given in \eqref{eq:pminmax_MF}, and $\mathcal J^{L,\mathrm{MF}}$ is defined in \eqref{eq:leader_premium_objective_MF}. For any such $p^{\mathrm{MF},*}$, define
\begin{equation*}
q^{\zeta,*}:=\widehat q^\zeta(p^{\mathrm{MF},*}),\qquad\pi^{\zeta,*}:=\widehat\pi^\zeta.
\end{equation*}
Then $(p^{\mathrm{MF},*},\pi^{L,\mathrm{MF},*},q^{\zeta,*},\pi^{\zeta,*})$ is a constant Stackelberg mean field equilibrium. The optimal premium $p^{\mathrm{MF},*}$ may not be unique, and therefore the constant Stackelberg mean field equilibrium may not be unique.
\end{proposition}

\begin{proof}
See \ref{app:3.3}.
\end{proof}

Proposition~\ref{prop:CSMFE_existence} gives an explicit unique leader's optimal investment $\pi^{L,\mathrm{MF},*}$, and reduces the optimal premium $p^{\mathrm{MF},*}$ to a one-dimensional optimization over the compact interval  $[p^{\mathrm{MF},\min},p^{\mathrm{MF},\max}]$, determined by the follower threshold structure in  Proposition~\ref{prop:premium_response_CMFER}. For $p\le p^{\mathrm{MF},\min}$, almost every type chooses full cession, and $\mathcal J^{L,\mathrm{MF}}(p)$ is increasing in $p$; for $p\ge p^{\mathrm{MF},\max}$, almost every type chooses full retention, and $\mathcal J^{L,\mathrm{MF}}(p)=0$. Hence an optimal premium can always be selected from $[p^{\mathrm{MF},\min},p^{\mathrm{MF},\max}]$.

In contrast to the finite-player case, a continuum-type distribution generally does not produce finitely many retention regimes. Therefore, the threshold continuation procedure developed in Section~\ref{subsec:finite_leader} cannot be applied, and the optimal premium $p^{\mathrm{MF},*}$ must be obtained directly from the one-dimensional optimization \eqref{eq:pstar_argmax_MF}, for example by grid search or other standard numerical methods. The exception is when $\zeta$ has a discrete distribution supported on finitely many types, in which the threshold continuation procedure in Section~\ref{subsec:finite_leader}  extends directly, with type-weighted sums replacing finite-player averages. Once an optimal premium $p^{\mathrm{MF},*}$ is selected, Proposition~\ref{prop:CSMFE_existence} gives the corresponding constant Stackelberg mean field equilibrium.

In summary, the Stackelberg mean field game reduces to two scalar objects: the follower fixed point $\widehat V(p)$ and the leader's reduced premium objective $\mathcal J^{L,\mathrm{MF}}(p)$. This explicit reduction is the basis for the finite-player-to-mean-field convergence analysis in Section~\ref{sec:convergence}.

\section{Convergence from the Finite-Player Game to the Mean Field Game}
\label{sec:convergence}
\subsection{Type-Form Setup}
\label{subsec:convergence_setup}
In this section, we justify the Stackelberg mean field game as the large-population limit of the finite-player Stackelberg--Nash game in Section~\ref{sec:finite_game}. Under weak convergence of the empirical distribution of the insurers' types, we establish the convergence of the finite-player follower equilibrium, the reinsurer's reduced premium objective, the finite-player optimal premiums, and the empirical distributions of finite-player Stackelberg--Nash equilibria to their mean field counterparts.

In the finite-player game of Section~\ref{sec:finite_game}, we collect each insurer $i$'s parameters into the type vector
\begin{equation*}
\zeta^{i,N}:=(x^{i,N},\eta^{i,N},a^{i,N},\gamma^{i,N},\theta^{i,N},v^{i,N,0},v^{i,N},\mu^{i,N},\sigma^{i,N,0},\sigma^{i,N}),\qquad i=1,\ldots,N.
\end{equation*}
The empirical distribution of the insurer types is defined as
\begin{equation*}
\mathfrak{m}^{N}:=\frac{1}{N}\sum_{i=1}^{N}\delta_{\zeta^{i,N}},
\end{equation*}
and we write $\mathfrak{m}:=\mathcal{L}(\zeta)$ for the type distribution in the mean field formulation of Section~\ref{sec:smfg}. 

For a generic type $z\in\mathcal Z$ (defined in Section~\ref{subsec:smfg_formulation}), we write
\begin{equation*}
z=(x_z,\eta_z,a_z,\gamma_z,\theta_z,v_z^0,v_z,\mu_z,\sigma_z^0,\sigma_z).
\end{equation*}

\begin{assumption}\label{assump}
Throughout this section, we impose the following: 
\begin{enumerate}
\item $N\ge2$, and the space $\mathcal Z$ is compact.

\item Each type $\zeta^{i,N}\in\mathcal Z$ for every $i=1,\ldots,N$.

\item There exist positive constants $\underline{a},\bar{a},\underline{\gamma},\bar{\gamma},\underline{v}^{0},\bar{v^{0}},\underline{v},\bar{v},\underline{\sigma}^{0},\bar{\sigma^{0}},\underline{\sigma},\bar{\sigma}$, and a finite constant $\bar\mu$ such that, for every $z\in\mathcal Z$,
\begin{align*}
0<\underline{a}\le a_{z}\le\bar{a},\qquad0<\underline{\gamma}\le\gamma_{z}\le\bar{\gamma},\qquad0<\underline{v}^{0}\le v_{z}^{0}\le\bar{v^{0}},\qquad0<\underline{v}\le v_{z}\le\bar{v},\\
0<\underline{\sigma}^{0}\le\sigma_{z}^{0}\le\bar{\sigma^{0}},\qquad0<\underline{\sigma}\le\sigma_{z}\le\bar{\sigma},\qquad|\mu_{z}|\le\bar{\mu},\qquad\theta_{z}\in[0,1].
\end{align*}

\item The empirical type distributions converge weakly:
\begin{equation}
\label{eq:empirical_type_convergence}
\mathfrak m^N\Rightarrow\mathfrak m\qquad\text{as }N\to\infty.
\end{equation}
\end{enumerate}
\end{assumption}

We set
\begin{equation*}
\underline p:=\underline a,\qquad \bar p:=\bar a+\frac{1}{\underline\gamma}\big((\bar{v^{0}})^2+\bar v^2\big).
\end{equation*}
By Propositions~\ref{prop:premium_response} and~\ref{prop:premium_response_CMFER}, the compact interval $[\underline{p},\bar{p}]$ contains the finite-player population-wide thresholds $p^{\min}$, $p^{\max}$ for all $N$ and the mean field counterparts $p^{\mathrm{MF},\min}$, $p^{\mathrm{MF},\max}$. By Propositions~\ref{prop:leader_equilibrium} and~\ref{prop:CSMFE_existence}, $[\underline{p},\bar{p}]$ then serves as a uniform search interval for the finite-player and mean field optimal premiums.

We rewrite the finite-player follower equilibrium response in Section~\ref{subsec:finite_followers} and the mean field counterpart in Section~\ref{subsec:smfg_followers} on a common type-space notation, with the generic type $z\in\mathcal{Z}$ replacing the insurer-specific subscripts and the random type. For $p\in[\underline{p},\bar{p}]$ and $V\in[0,\bar{v^{0}}]$, the type-form retention incentives are
\begin{equation}
\label{eq:type_retention_incentive}
I^{N}(z;p,V):=\frac{(p-a_{z})+\frac{1}{\gamma_{z}}\theta_{z}v_{z}^{0}V}{\frac{1}{\gamma_{z}}\big((v_{z}^{0})^{2}+(v_{z})^{2}(1-\theta_{z}/N)\big)},\qquad I(z;p,V):=\frac{(p-a_{z})+\frac{1}{\gamma_{z}}\theta_{z}v_{z}^{0}V}{\frac{1}{\gamma_{z}}\big((v_{z}^{0})^{2}+(v_{z})^{2}\big)},
\end{equation}
which are the retention incentives $I^{i}(p,V)$ from \eqref{eq:retention_incentive} and $I^{\zeta}(p,V)$ in \eqref{eq:retention_incentive_MF} viewed as functions of a generic type $z$ so that $I^{N}(\zeta^{i,N};p,V)=I^{i}(p,V)$ and $I(\zeta;p,V)=I^{\zeta}(p,V)$. The fixed-point equations \eqref{eq:fixed_point_V} and \eqref{eq:fixed_point_V_MF} then can be written as $V=F^{N}(p,V)$ and $V=F(p,V)$, where
\begin{equation*}
F^{N}(p,V):=\int_{\mathcal{Z}}v_{z}^{0}\,\Pi_{[0,1]}\big(I^{N}(z;p,V)\big)\,\mathfrak{m}^{N}(dz),\qquad F(p,V):=\int_{\mathcal{Z}}v_{z}^{0}\,\Pi_{[0,1]}\big(I(z;p,V)\big)\,\mathfrak{m}(dz),
\end{equation*}
and their unique fixed points are the average common-insurance exposures $\widehat{V}^{N}(p)$ in \eqref{eq:Vhat_Sigmahat_def} and $\widehat{V}(p)$ in \eqref{eq:Vhat_Sigmahat_MF}. We accordingly write the equilibrium retention responses \eqref{eq:qhat_projection} and \eqref{eq:qhat_MF} in type-form as
\begin{equation*}
\widehat{q}^{N}(z;p):=\Pi_{[0,1]}\big(I^{N}(z;p,\widehat{V}^{N}(p))\big),\qquad\widehat{q}(z;p):=\Pi_{[0,1]}\big(I(z;p,\widehat{V}(p))\big),
\end{equation*}
so that $\widehat{q}^{N}(\zeta^{i,N};p)=\widehat{q}^{i}(p)$ and $\widehat{q}(\zeta;p)=\widehat{q}^{\zeta}(p)$.

Analogously, the average common-financial exposures $\widehat{\Sigma}^{N}$ in \eqref{eq:pihat} and $\widehat{\Sigma}$ in \eqref{eq:Vhat_Sigmahat_MF}, viewed as integrals against $\mathfrak{m}^{N}$ and $\mathfrak{m}$, are
\begin{equation}
\label{eq:Sigmahat_type_form}
\widehat{\Sigma}^{N}=\frac{\int_{\mathcal{Z}}\frac{\gamma_{z}\sigma_{z}^{0}\mu_{z}}{(\sigma_{z}^{0})^{2}+(\sigma_{z})^{2}(1-\theta_{z}/N)}\,\mathfrak{m}^{N}(dz)}{1-\int_{\mathcal{Z}}\frac{\theta_{z}(\sigma_{z}^{0})^{2}}{(\sigma_{z}^{0})^{2}+(\sigma_{z})^{2}(1-\theta_{z}/N)}\,\mathfrak{m}^{N}(dz)},\qquad\widehat{\Sigma}=\frac{\int_{\mathcal{Z}}\frac{\gamma_{z}\sigma_{z}^{0}\mu_{z}}{(\sigma_{z}^{0})^{2}+(\sigma_{z})^{2}}\,\mathfrak{m}(dz)}{1-\int_{\mathcal{Z}}\frac{\theta_{z}(\sigma_{z}^{0})^{2}}{(\sigma_{z}^{0})^{2}+(\sigma_{z})^{2}}\,\mathfrak{m}(dz)},
\end{equation}
and we write the equilibrium investment responses $\widehat{\pi}^{i}$ in \eqref{eq:pihat} and $\widehat{\pi}^{\zeta}$ in \eqref{eq:pihat_MF} in type-form as
\begin{equation}
\label{eq:pihat_type_form}
\widehat{\pi}^{N}(z):=\frac{\mu_{z}+\frac{1}{\gamma_{z}}\theta_{z}\sigma_{z}^{0}\widehat{\Sigma}^{N}}{\frac{1}{\gamma_{z}}\big((\sigma_{z}^{0})^{2}+(\sigma_{z})^{2}(1-\theta_{z}/N)\big)},\qquad\widehat{\pi}(z):=\frac{\mu_{z}+\frac{1}{\gamma_{z}}\theta_{z}\sigma_{z}^{0}\widehat{\Sigma}}{\frac{1}{\gamma_{z}}\big((\sigma_{z}^{0})^{2}+(\sigma_{z})^{2}\big)},
\end{equation}
so that $\widehat{\pi}^{N}(\zeta^{i,N})=\widehat{\pi}^{i}$ and $\widehat{\pi}(\zeta)=\widehat{\pi}^{\zeta}$.

\subsection{Convergence of the Follower Equilibrium Response}
\label{subsec:convergence_follower}
Building on the type-form notation of Section~\ref{subsec:convergence_setup}, we now establish the convergence of the finite-player follower equilibrium response to its mean field counterpart, treating the retention and investment responses separately before combining them at the level of type-control empirical distributions.

\begin{proposition}
\label{prop:qhat_uniform_convergence}
We have
\begin{equation}
\label{eq:VhatN_to_Vhat_uniform}
\sup_{p\in[\underline{p},\bar{p}]}\big|\widehat{V}^{N}(p)-\widehat{V}(p)\big|\longrightarrow0,
\end{equation}
and
\begin{equation}
\label{eq:qhat_uniform_convergence}
\sup_{(z,p)\in\mathcal{Z}\times[\underline{p},\bar{p}]}\big|\widehat{q}^{N}(z;p)-\widehat{q}(z;p)\big|\longrightarrow0.
\end{equation}
\end{proposition}

\begin{proof}
See \ref{app:4.1}.
\end{proof}

Proposition~\ref{prop:qhat_uniform_convergence} gives the uniform convergence of the finite-player retention response to its mean field counterpart. The uniformity in $(z,p)$ will be used in Section~\ref{subsec:convergence_leader} to transfer this convergence to the leader's reduced premium objective, which involves both optimization over $p\in[\underline{p},\bar{p}]$ and integration over $z\in\mathcal{Z}$.

We next turn to the follower investment response. Since $\widehat{\pi}^{N}(z)$ in \eqref{eq:pihat_type_form} depends on the insurer population only through the scalar average common-financial exposure $\widehat{\Sigma}^{N}$, its uniform convergence on $\mathcal{Z}$ follows from the scalar convergence $\widehat{\Sigma}^{N}\to\widehat{\Sigma}$.

\begin{proposition}
\label{prop:pihat_uniform_convergence}
We have
\begin{equation}
\label{eq:SigmahatN_to_Sigmahat}
\widehat{\Sigma}^{N}\longrightarrow\widehat{\Sigma},
\end{equation}
and
\begin{equation}
\label{eq:pihat_uniform_convergence}
\sup_{z\in\mathcal{Z}}\big|\widehat{\pi}^{N}(z)-\widehat{\pi}(z)\big|\longrightarrow0.
\end{equation}
\end{proposition}

\begin{proof}
See \ref{app:4.2}.
\end{proof}

Combining the retention and investment convergence (Propositions~\ref{prop:qhat_uniform_convergence} and \ref{prop:pihat_uniform_convergence}), we next obtain the weak convergence of the type-control empirical distributions induced by the follower equilibrium response. Thus, for any converging premium sequence, the joint distribution of insurer types and their equilibrium responses across the population converges to the mean field type-control distribution.

\begin{proposition}[Convergence of follower type-control distributions]
\label{prop:profile_distribution_convergence}
Let $(p^{N})_{N\ge1}$ be a sequence in $[\underline{p},\bar{p}]$ such that $p^{N}\to p^{\infty}$. Then 
\begin{equation}
\label{eq:profile_distribution_convergence}
\frac{1}{N}\sum_{i=1}^{N}\delta_{(\zeta^{i,N},\,\widehat{q}^{N}(\zeta^{i,N};p^{N}),\,\widehat{\pi}^{N}(\zeta^{i,N}))}\Rightarrow\mathcal{L}(\zeta,\widehat{q}(\zeta;p^{\infty}),\widehat{\pi}(\zeta)).
\end{equation}
\end{proposition}

\begin{proof}
See \ref{app:4.3}.
\end{proof}

\subsection{Convergence of the Leader Problem and Stackelberg Equilibria}
\label{subsec:convergence_leader}

Building on the follower-side convergence of Section~\ref{subsec:convergence_follower}, we now establish the convergence of the leader's reduced premium objective, the optimal premium sets, and the constant Stackelberg--Nash equilibria, all from the finite-player setting to their mean field counterparts.

We rewrite the finite-player and mean field reduced premium objectives \eqref{eq:leader_premium_objective} and \eqref{eq:leader_premium_objective_MF} on the type-form introduced in Section~\ref{subsec:convergence_setup}. The finite-player reduced premium objective $\mathcal{J}^{L}(p)$ in \eqref{eq:leader_premium_objective}, which we denote by $\mathcal{J}^{L,N}(p)$ in this section to mark its $N$-dependence, reads
\begin{align}
\label{eq:leader_premium_objective_typeform}
\mathcal{J}^{L,N}(p) = & \int_{\mathcal{Z}}(1-\widehat{q}^{N}(z;p))(p-a_{z})\,\mathfrak{m}^{N}(dz)\nonumber\\
& -\frac{1}{2\gamma^{L}}\Big[\Big(\int_{\mathcal{Z}}(1-\widehat{q}^{N}(z;p))v_{z}^{0}\,\mathfrak{m}^{N}(dz)\Big)^{2}+\frac{1}{N}\int_{\mathcal{Z}}(1-\widehat{q}^{N}(z;p))^{2}(v_{z})^{2}\,\mathfrak{m}^{N}(dz)\Big],
\end{align}
and the mean field reduced premium objective \eqref{eq:leader_premium_objective_MF} reads
\begin{equation}
\label{eq:leader_premium_objective_MF_typeform}
\mathcal{J}^{L,\mathrm{MF}}(p)=\int_{\mathcal{Z}}(1-\widehat{q}(z;p))(p-a_{z})\,\mathfrak{m}(dz)-\frac{1}{2\gamma^{L}}\Big(\int_{\mathcal{Z}}(1-\widehat{q}(z;p))v_{z}^{0}\,\mathfrak{m}(dz)\Big)^{2}.
\end{equation}
The only structural difference between \eqref{eq:leader_premium_objective_typeform} and \eqref{eq:leader_premium_objective_MF_typeform} is the additional last term in $\mathcal{J}^{L,N}$, which arises from the insurer-specific idiosyncratic insurance-claim noises $W^{i}$ at finite $N$.

\begin{proposition}
\label{prop:J_uniform_convergence}
We have
\begin{equation}
\label{eq:J_uniform_convergence}
\sup_{p\in[\underline{p},\bar{p}]}\big|\mathcal{J}^{L,N}(p)-\mathcal{J}^{L,\mathrm{MF}}(p)\big|\longrightarrow0.
\end{equation}
\end{proposition}
\begin{proof}
See \ref{app:4.4}.
\end{proof}

Proposition~\ref{prop:J_uniform_convergence} transfers the follower-side convergence of Proposition~\ref{prop:qhat_uniform_convergence} to the leader's reduced premium objective. The proof reveals that the last term in \eqref{eq:leader_premium_objective_typeform} vanishes in the large-population limit, reflecting the diversification of insurer-specific idiosyncratic insurance-claim noise across a large insurer population. This is consistent with the mean field reinsurer dynamics \eqref{eq:dynamics_reinsurer_MF}, in which the idiosyncratic insurance-claim noises $W^{i}$ are absent.

We next consider the finite-player and mean field optimal premium sets,
\begin{equation*}
\mathcal{P}^{N,*}:=\underset{p\in[\underline{p},\bar{p}]}{\arg\max}\,\mathcal{J}^{L,N}(p),\qquad\mathcal{P}^{\mathrm{MF},*}:=\underset{p\in[\underline{p},\bar{p}]}{\arg\max}\,\mathcal{J}^{L,\mathrm{MF}}(p),
\end{equation*}
which are nonempty compact subsets of $[\underline{p},\bar{p}]$, since $\mathcal{J}^{L,N}$ and $\mathcal{J}^{L,\mathrm{MF}}$ are continuous in $p$. We optimize both finite-player and mean field premium problems over the common compact interval $[\underline{p},\bar{p}]$. This fixed domain contains all effective premium intervals and is convenient for the convergence analysis. It may include full-retention regions: $\mathcal J^{L,N}(p)=0$ for $p\in[p^{\max},\bar p]$ and $\mathcal J^{L,\mathrm{MF}}(p)=0$ for $p\in[p^{\mathrm{MF},\max},\bar p]$. Hence, if this zero value is optimal, the argmax may contain an interval of economically equivalent premiums. One may remove this artificial nonuniqueness by selecting the smallest maximizer.

\begin{proposition}[Convergence of finite-player optimal premiums]
\label{prop:optimal_premium_convergence}
Every limit point of a sequence of finite-player optimal premiums is a mean field optimal premium.

More precisely, if $(N_{k})_{k\ge1}$ is a strictly increasing sequence of positive integers, $p^{N_{k},*}\in\mathcal{P}^{N_{k},*}$ for each $k$, and $p^{N_{k},*}\to p^{\infty,*}$, then $p^{\infty,*}\in\mathcal{P}^{\mathrm{MF},*}$. Equivalently,
\begin{equation}
\label{eq:optimal_premium_convergence}
\sup_{p\in\mathcal{P}^{N,*}}\mathrm{dist}\big(p,\mathcal{P}^{\mathrm{MF},*}\big)\longrightarrow0,
\end{equation}
where $\mathrm{dist}(p,A):=\inf_{q\in A}|p-q|$ denotes the Euclidean point-to-set distance.
\end{proposition}
\begin{proof}
See \ref{app:4.5}.
\end{proof}

Proposition~\ref{prop:optimal_premium_convergence} converts the uniform convergence of the leader objectives (cf. Proposition~\ref{prop:J_uniform_convergence}) into the convergence of optimal premium sets. Since $\mathcal{P}^{\mathrm{MF},*}$ may have multiple maximizers, the result is stated in set-distance form, with the additional information that every limit point of finite-player optimal premiums is mean field optimal. 

\begin{corollary}
\label{cor:optimal_premium_unique}
If $\mathcal{P}^{\mathrm{MF},*}=\{p^{\mathrm{MF},*}\}$, then
\begin{equation*}
\sup_{p\in\mathcal{P}^{N,*}}|p-p^{\mathrm{MF},*}|\longrightarrow0.
\end{equation*}
In particular, every sequence $(p^{N,*})_{N\ge1}$ with $p^{N,*}\in\mathcal{P}^{N,*}$ converges to $p^{\mathrm{MF},*}$.
\end{corollary}

We now combine the follower distributional convergence in Proposition~\ref{prop:profile_distribution_convergence} with the optimal premium convergence in Proposition~\ref{prop:optimal_premium_convergence} to obtain the convergence of the constant Stackelberg--Nash equilibria. Because the reduced mean field premium objective may have multiple maximizers, one should not expect every finite-player equilibrium sequence to converge to a single preselected mean field equilibrium. The next result identifies all possible weak limit points.

\begin{theorem}[Convergence of Stackelberg--Nash equilibrium type-control distributions]
\label{thm:empirical_profile_convergence}
Every weak limit point of the empirical type-control distributions of finite-player Stackelberg--Nash equilibria is induced by a Stackelberg mean field equilibrium. 

More precisely, let $(p^{N,*})_{N\ge1}$ be any sequence with $p^{N,*}\in\mathcal{P}^{N,*}$, and define
\begin{equation}
\label{eq:selected_equilibrium_controls}
q^{i,N,*}:=\widehat{q}^{N}(\zeta^{i,N};p^{N,*}),\qquad\pi^{i,N,*}:=\widehat{\pi}^{N}(\zeta^{i,N}),\qquad i=1,\ldots,N.
\end{equation}
Suppose that, along a strictly increasing sequence $(N_{k})_{k\ge1}$,
\begin{equation}
\label{eq:empirical_profile_assumption}
\frac{1}{N_{k}}\sum_{i=1}^{N_{k}}\delta_{(\zeta^{i,N_{k}},q^{i,N_{k},*},\pi^{i,N_{k},*})}\Rightarrow\Lambda
\end{equation}
for some probability distribution $\Lambda$ on $\mathcal{Z}\times[0,1]\times\mathbb{R}$. Then there exists $p^{\infty,*}\in\mathcal{P}^{\mathrm{MF},*}$ such that
\begin{equation}
\label{eq:empirical_weak_limit}
\Lambda=\mathcal{L}\big(\zeta,\widehat{q}(\zeta;p^{\infty,*}),\widehat{\pi}(\zeta)\big).
\end{equation}
\end{theorem}
\begin{proof}
See \ref{app:4.6}.
\end{proof}

Theorem~\ref{thm:empirical_profile_convergence} identifies every weak limit of empirical type-control distributions of finite-player Stackelberg--Nash equilibria as the type-control distribution induced by a Stackelberg mean field equilibrium. The next two corollaries record the standard subsequential and unique-case consequences.

\begin{corollary}
\label{cor:subsequential_equilibrium}
Let $(p^{N,*})_{N\ge1}$ be any sequence with $p^{N,*}\in\mathcal{P}^{N,*}$, and define $q^{i,N,*}$ and $\pi^{i,N,*}$ as in \eqref{eq:selected_equilibrium_controls}. Then every subsequence admits a further subsequence $(N_{k_{\ell}})_{\ell\ge1}$ and some $p^{\infty,*}\in\mathcal{P}^{\mathrm{MF},*}$ such that
\begin{equation*}
p^{N_{k_{\ell}},*}\to p^{\infty,*},
\end{equation*}
and
\begin{equation*}
\frac{1}{N_{k_{\ell}}}\sum_{i=1}^{N_{k_{\ell}}}\delta_{(\zeta^{i,N_{k_{\ell}}},q^{i,N_{k_{\ell}},*},\pi^{i,N_{k_{\ell}},*})}\Rightarrow\mathcal{L}\big(\zeta,\widehat{q}(\zeta;p^{\infty,*}),\widehat{\pi}(\zeta)\big).
\end{equation*}
\end{corollary}

\begin{corollary}
\label{cor:unique_equilibrium}
If $\mathcal{P}^{\mathrm{MF},*}=\{p^{\mathrm{MF},*}\}$, then every sequence $(p^{N,*})_{N\ge1}$ with $p^{N,*}\in\mathcal{P}^{N,*}$ converges to $p^{\mathrm{MF},*}$, and
\begin{equation*}
\frac{1}{N}\sum_{i=1}^{N}\delta_{(\zeta^{i,N},q^{i,N,*},\pi^{i,N,*})}\Rightarrow\mathcal{L}\big(\zeta,\widehat{q}(\zeta;p^{\mathrm{MF},*}),\widehat{\pi}(\zeta)\big),
\end{equation*}
where $q^{i,N,*}$ and $\pi^{i,N,*}$ are defined as in \eqref{eq:selected_equilibrium_controls}.
\end{corollary}

Together, Propositions~\ref{prop:qhat_uniform_convergence}--\ref{prop:optimal_premium_convergence} and Theorem~\ref{thm:empirical_profile_convergence} justify the Stackelberg mean field game in Section~\ref{sec:smfg} as a large-population approximation of the finite-player Stackelberg--Nash game in Section~\ref{sec:finite_game}. In particular, the follower equilibrium, the leader's reduced premium objective, the optimal premium set, and Stackelberg equilibria are stable under Assumption~\ref{assump}.

\begin{remark}
A simple sufficient condition for the assumption $\mathcal{P}^{\mathrm{MF},*}=\{p^{\mathrm{MF},*}\}$ in Corollary~\ref{cor:unique_equilibrium} is population homogeneity. Indeed, if the type distribution is a Dirac mass, $\mathfrak{m}=\delta_{z_0}$, for some fixed type $z_0\in\mathcal{Z}$, then the follower retention is affine in $p$ from $0$ to $1$ on the effective premium interval $[p^{\mathrm{MF},\min},p^{\mathrm{MF},\max}]$. Hence, the reduced premium objective $\mathcal{J}^{L,\mathrm{MF}}$ is strictly concave (quadratic) on this interval, and therefore has a unique maximizer.
\end{remark}

\begin{remark}
The leader's investment control does not require a separate convergence argument, because its optimizer is the same in both the finite-player and the mean field formulations:
\begin{equation*}
\pi^{L,*}=\pi^{L,\mathrm{MF},*}=\frac{\gamma^{L}\mu^{L}}{(\sigma^{L,0})^{2}+(\sigma^{L})^{2}}.
\end{equation*}
\end{remark}

\section{Numerical Illustrations} 
\label{sec:numerics}
\subsection{Finite-Player Stackelberg--Nash Equilibria}\label{sec:numerics-finite}
We now illustrate the finite-player analysis developed in Section~\ref{sec:finite_game}. Since the investment controls are explicit (cf. \eqref{eq:pihat} and \eqref{eq:pistar_finite}), we focus on follower retention responses, the leader's premium objective, and constant Stackelberg--Nash premium and retention levels. These quantities depend on the insurers only through $(a^i,\gamma^i,\theta^i,v^{i,0},v^i)$ and on the reinsurer only through $\gamma^L$, so we leave the remaining parameters $x^i,\eta^i,\mu^i,\sigma^{i,0},\sigma^i,x^L,\mu^L,\sigma^{L,0},\sigma^L$ unspecified. Throughout this subsection, we set $\gamma^i=1$, $v^{i,0}=0.7$, $v^i=0.3$, and $\gamma^L=5$. Across the examples, we vary the relative performance parameters $\theta^i$; Figure~\ref{fig:1} also uses heterogeneous claim rates $a^i$, while Figures~\ref{fig:2} and~\ref{fig:3} use homogeneous claim rates.

Figure~\ref{fig:1} considers three insurers with $(a^1,a^2,a^3)=(1,1.30,1.35)$ and $(\theta^1,\theta^2,\theta^3)=(0.5,0,1)$. Consistent with Proposition~\ref{prop:premium_response} and Corollary~\ref{cor:premium_partition}, each $\widehat q^i(p)$ is continuous and monotonically increasing, and is affine on each subinterval of $[p^{\min},p^{\max}]=[1,1.880]$, where the follower retention configuration is fixed. Here, $p^{\min}=1$ and $p^{\max}=1.880$ follow from \eqref{eq:pminmax}, while the insurer-specific thresholds $p^{1,\min}=1$, $p^{2,\min}=1.300$, $p^{3,\min}\approx1.262$, $p^{1,\max}\approx1.413$, $p^{2,\max}=1.880$, $p^{3,\max}\approx1.513$ are obtained from the threshold continuation procedure. The thresholds illustrate the effect of relative performance concerns. Insurer~2, with $\theta^2=0$, starts retaining at $p^{2,\min}=a^2=1.30$. In contrast, insurer~3, with $\theta^3=1$, starts at $p^{3,\min}\approx1.262<a^3=1.35$, where the premium-loss spread $p-a^3$ is still negative.  Thus, insurer~3 starts retaining before insurer~2 despite having the larger expected claim rate. This early retention reflects the common-insurance exposure spillover, which is also visible in the slopes of the retention responses. On $(p^{1,\min},p^{3,\min})$, only insurer~1 retains, and $\widehat q^1$ rises affinely. Once insurer~3 begins retention at $p^{3,\min}$, the spillover through $\widehat V^N(p)$ makes $\widehat q^1$ rise more steeply; when insurer~2 begins partial retention at $p^{2,\min}$, the same spillover steepens the responses of insurers~1 and~3, with a larger effect on insurer~3 because $\theta^3>\theta^1$. By contrast, insurer~2 has $\theta^2=0$, so it is unaffected by this spillover and keeps a constant-slope retention response. For insurers~1 and~3, the full-retention thresholds $p^{1,\max}\approx1.413$ and $p^{3,\max}\approx1.513$ lie below their no-relative-performance full-retention thresholds, $1.58$ and $1.93$, respectively. Since the reinsurer anticipates this entire follower response, it chooses the premium $p^*\approx1.605$ that maximizes its reduced premium objective $\mathcal J^L$, which determines the Stackelberg--Nash equilibrium; at equilibrium, insurers~1 and~3 fully retain, while insurer~2 retains partially.

\begin{figure}[!htb]
    \centering
    \includegraphics[width=0.8\linewidth]{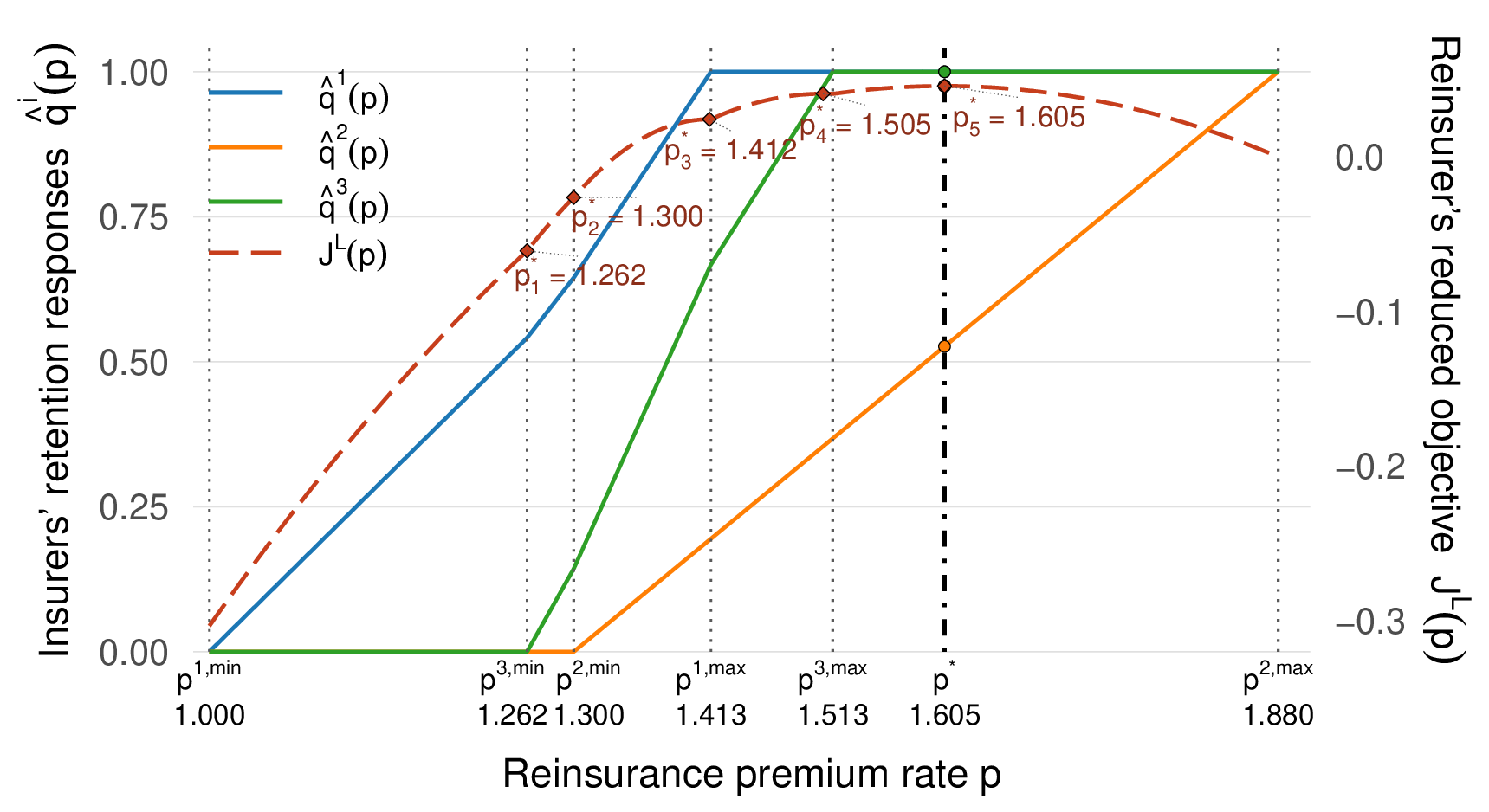}
    \caption{Follower equilibrium retentions $\widehat q^i(p)$, $i=1,2,3$, and the reinsurer's reduced objective $\mathcal J^L(p)$. The dotted vertical lines mark the insurer-specific thresholds, the red diamonds mark the intervalwise premium candidates $p_k^*$, and the black dash-dotted vertical line marks the Stackelberg--Nash equilibrium premium $p^*$. The parameters are $(a^1,a^2,a^3)=(1,1.30,1.35)$, $(\gamma^1,\gamma^2,\gamma^3)=(1,1,1)$, $(\theta^1,\theta^2,\theta^3)=(0.5,0,1)$, $(v^{1,0},v^{2,0},v^{3,0})=(0.7,0.7,0.7)$, $(v^1,v^2,v^3)=(0.3,0.3,0.3)$, and $\gamma^L=5$.}
    \label{fig:1}
\end{figure}

Figure~\ref{fig:2} illustrates the nonuniqueness of constant Stackelberg--Nash equilibria. Here $a^1=a^2=1$ and $\theta^1=0$, while $\theta^2>0$ increases across panels. Since only insurer~2 has a relative performance concern, it reaches full retention before insurer~1. 
The two intervalwise maximizers of $\mathcal J^L$ correspond to two retention configurations: both insurers retain partially, or insurer~2 fully retains while insurer~1 retains partially. The value of $\mathcal J^L$ at each maximizer varies continuously with $\theta^2$. For small $\theta^2$ (panels~(a)--(b)), the both-partial maximizer gives the larger value, whereas at $\theta^2=1$ (panel~(d)) the other maximizer gives the larger value; hence the two values coincide at some intermediate $\theta^2$, which we compute to be $\theta^2\approx0.925391$ (panel~(c)). There, two constant Stackelberg--Nash equilibria coexist at distinct premiums of approximately $1.184$ and $1.304$. Beyond this value, the equilibrium switches discontinuously from the both-partial configuration to the one in which insurer~2 fully retains; at the switch, $p^*$, $q^{1,*}$, and $q^{2,*}$ jump together.

\begin{figure}[!htb]
    \centering
    \begin{subfigure}{0.45\textwidth}
        \centering
        \includegraphics[width=\linewidth]{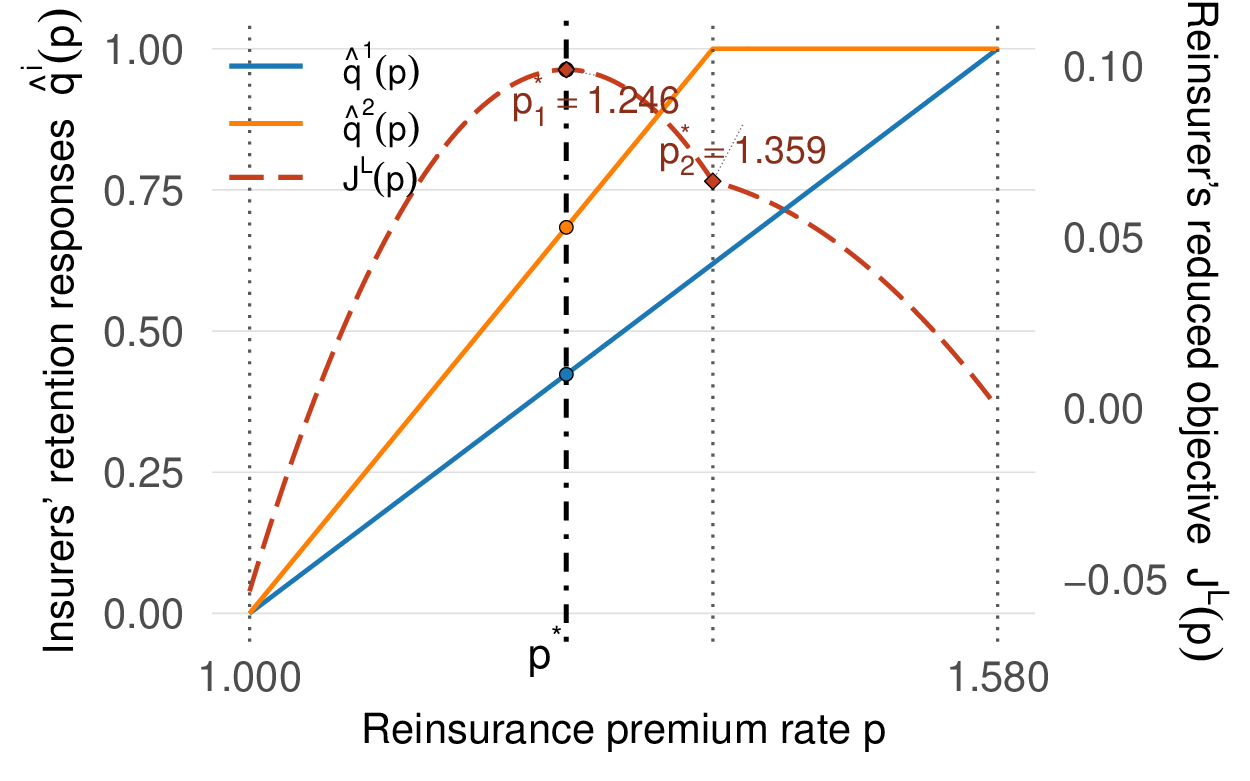}
        \caption{$\theta^2=0.5$}
    \end{subfigure}
    \hfill
    \begin{subfigure}{0.45\textwidth}
        \centering
        \includegraphics[width=\linewidth]{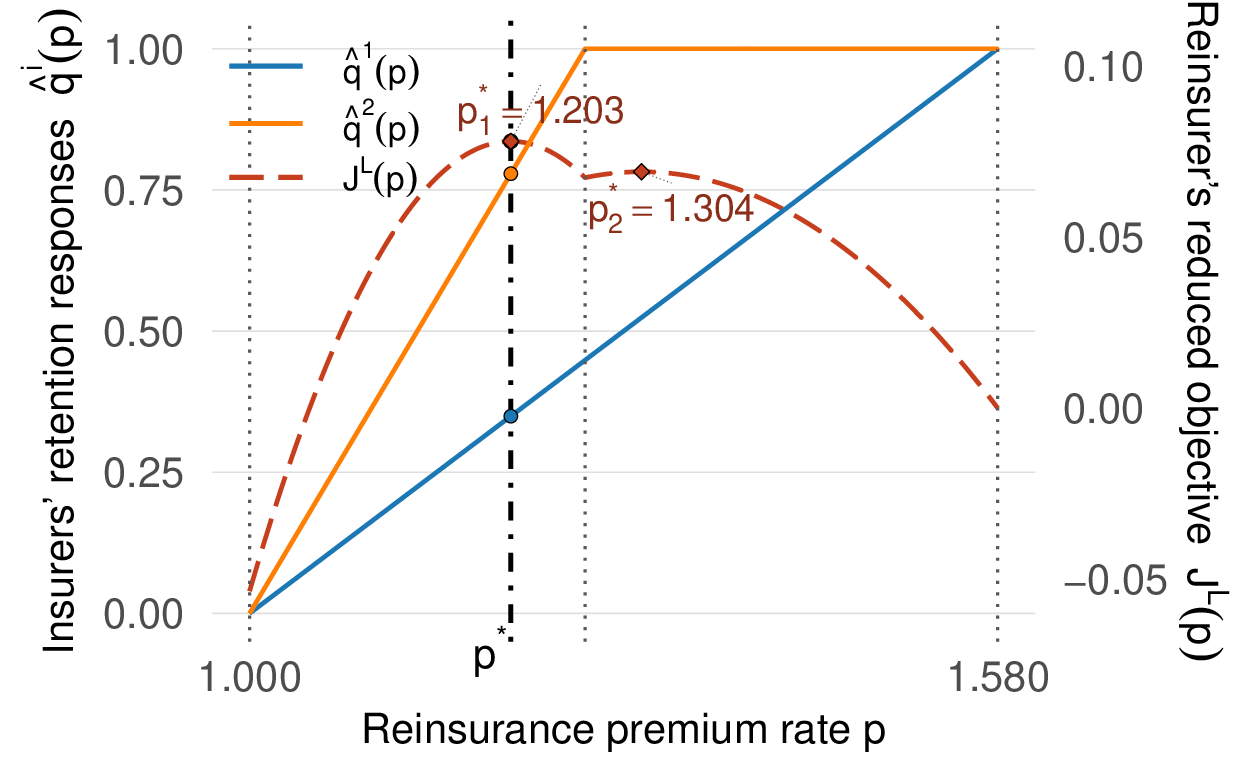}
        \caption{$\theta^2=0.8$}
    \end{subfigure}
    \hfill
    \begin{subfigure}{0.45\textwidth}
        \centering
        \includegraphics[width=\linewidth]{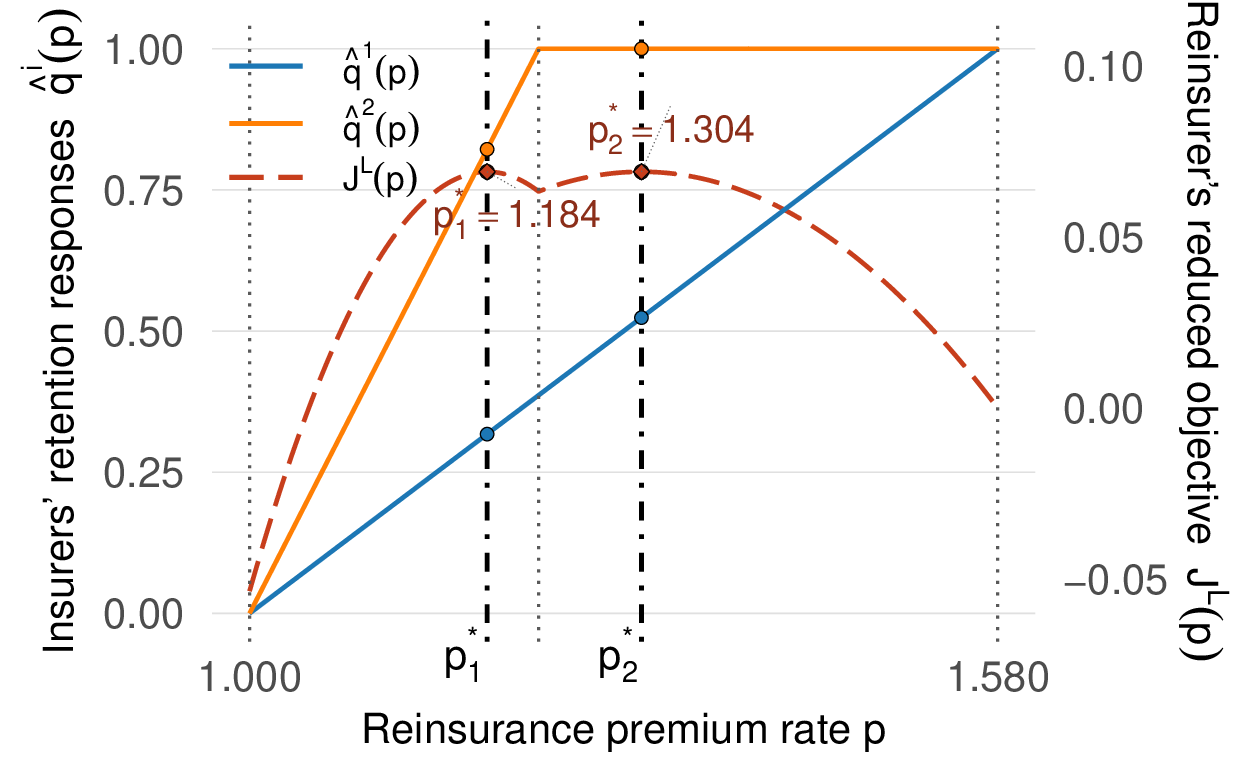}
        \caption{$\theta^2\approx0.925391$}
    \end{subfigure}
    \hfill
    \begin{subfigure}{0.45\textwidth}
        \centering
        \includegraphics[width=\linewidth]{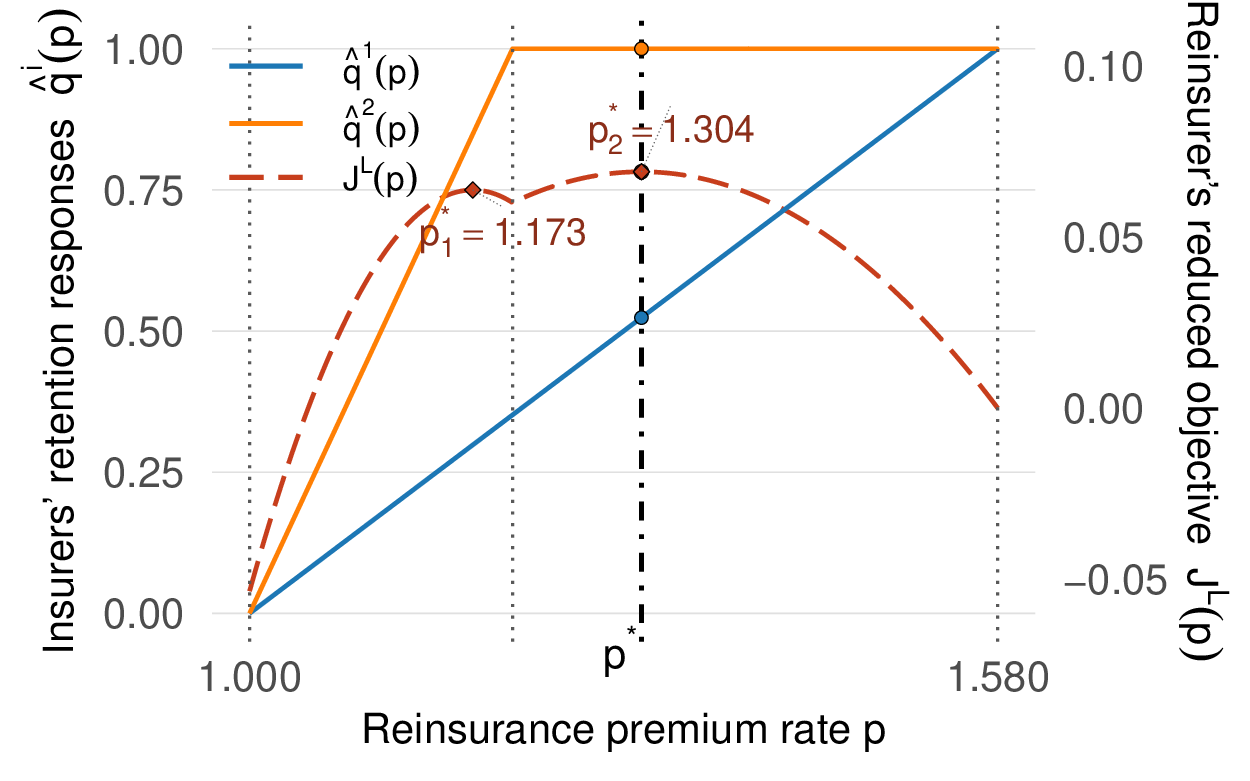}
        \caption{$\theta^2=1$}
    \end{subfigure}
    \caption{Two-insurer examples with $\theta^1=0$ and varying $\theta^2$. Each panel plots the follower equilibrium retentions $\widehat q^i(p)$, $i=1,2$, and the reinsurer's reduced objective $\mathcal J^L(p)$. The graphical conventions are the same as in Figure~\ref{fig:1}. Panel~(c) gives the critical case with multiple Stackelberg--Nash equilibrium premiums. The parameters are $(a^1,a^2)=(1,1)$, $(\gamma^1,\gamma^2)=(1,1)$, $(v^{1,0},v^{2,0})=(0.7,0.7)$, $(v^1,v^2)=(0.3,0.3)$, and $\gamma^L=5$.}
    \label{fig:2}
\end{figure}

Figure~\ref{fig:3} shows how the two-insurer Stackelberg--Nash equilibrium depends on the relative performance parameters $(\theta^1,\theta^2)$ over the full plane $[0,1]^2$; Figure~\ref{fig:2} is the one-parameter slice with $\theta^1=0$. 
Panels~(a) and~(b) report the equilibrium premium $p^*$ and insurer~1's equilibrium retention $q^{1,*}$, respectively. Relabeling the insurers leaves $p^*$ unchanged but interchanges $q^{1,*}$ and $q^{2,*}$; hence Panel~(a) is symmetric about the diagonal $\theta^1=\theta^2$, while Panel~(b) is not, and $q^{2,*}$ is the reflection of $q^{1,*}$ across the diagonal. Except near the upper-left and lower-right regions and their boundaries, the constant Stackelberg--Nash equilibrium varies continuously with $(\theta^1,\theta^2)$: moving from the lower-left region toward the upper-right region, both relative performance parameters increase, the equilibrium retentions $q^{1,*}$ and $q^{2,*}$ rise, and the equilibrium premium $p^*$ falls. In the upper-left and lower-right regions, one relative performance parameter is much larger than the other, and the insurer with the larger parameter fully retains in equilibrium. Along the boundary of each such region, two intervalwise premium candidates attain the same value of $\mathcal J^L$, so multiple constant Stackelberg--Nash equilibria coexist. Crossing the boundary changes the global maximizer from the both-partial candidate to the candidate in which one insurer fully retains, and the other retains partially; therefore, $p^*$, $q^{1,*}$, and $q^{2,*}$ jump. On the slice $\theta^1=0$, this boundary is reached at $\theta^2\approx0.925391$, which is the critical case shown in Figure~\ref{fig:2}.

\begin{figure}[!htb]
    \centering
    \begin{subfigure}{0.45\textwidth}
        \centering
        \includegraphics[width=\linewidth]{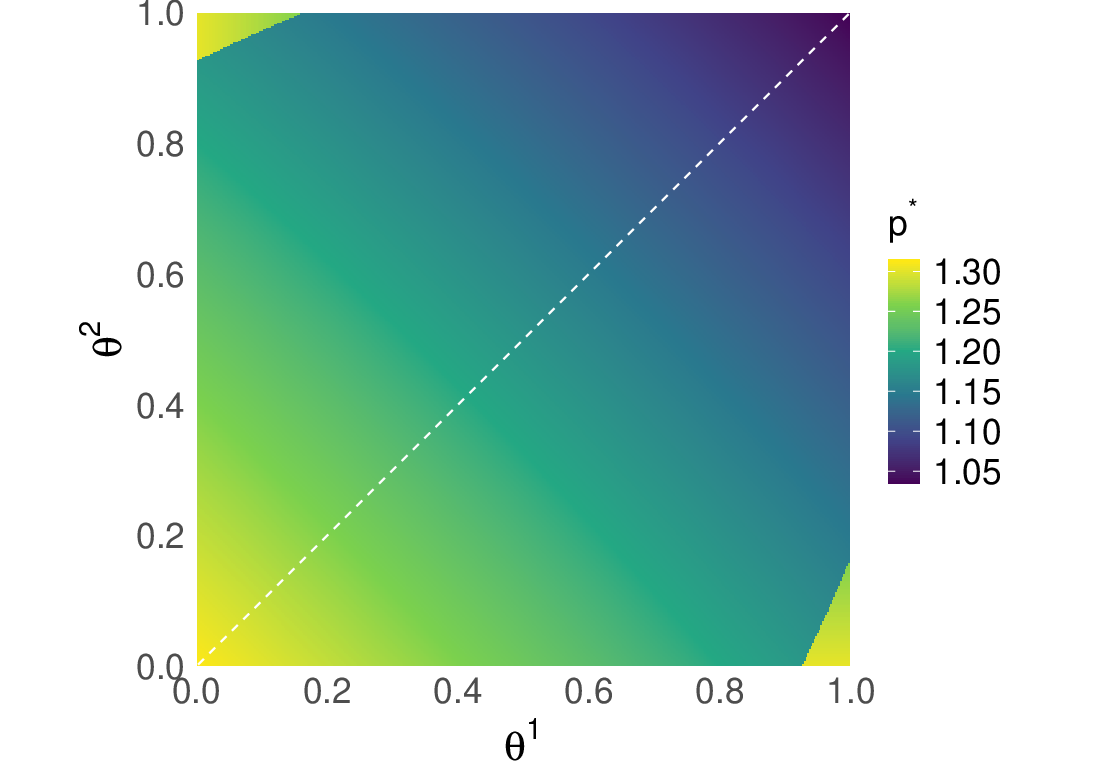}
        \caption{Equilibrium reinsurance premium rate $p^*$}
    \end{subfigure}
    \hfill
    \begin{subfigure}{0.45\textwidth}
        \centering
        \includegraphics[width=\linewidth]{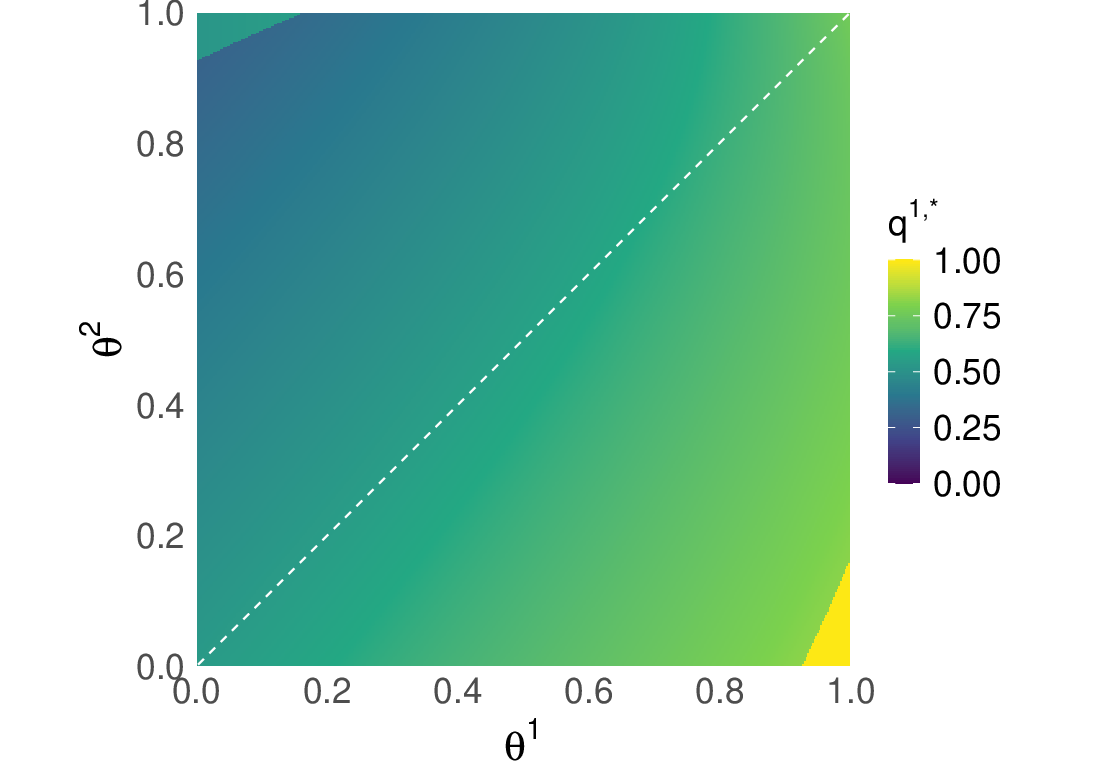}
        \caption{Equilibrium retention level $q^{1,*}$}
    \end{subfigure}
    \caption{Two-insurer Stackelberg--Nash equilibrium maps over the relative performance parameter plane. The white dashed diagonal marks the symmetric case $\theta^1=\theta^2$. The heatmaps are computed over $(\theta^1,\theta^2)\in[0,1]^2$ using a uniform $401\times401$ grid with mesh size $0.0025$. The parameters are $(a^1,a^2)=(1,1)$, $(\gamma^1,\gamma^2)=(1,1)$, $(v^{1,0},v^{2,0})=(0.7,0.7)$, $(v^1,v^2)=(0.3,0.3)$, and $\gamma^L=5$.}
    \label{fig:3}
\end{figure}

\subsection{Stackelberg Mean Field Equilibria and Finite-Player Convergence}
\label{subsec:num_mfg}
We next illustrate the mean field game of Section~\ref{sec:smfg} and the finite-player convergence of Section~\ref{sec:convergence}, using a tractable specification in which heterogeneity enters only through the relative performance parameter $\theta$, with density
\begin{equation}
\label{eq:theta_density_h}
f_{\theta}^{h}(\theta) =
\begin{cases}
h, & 0\le\theta\le0.2,\\
\tfrac{5-2h}{3}, & 0.2<\theta<0.8,\\
h, & 0.8\le\theta\le1,
\end{cases}
\end{equation}
where $0<h<2.5$. The remaining insurer parameters are fixed at $a=1$, $\gamma=1$, $v^0=0.7$, $v=0.3$. The case $h=1$ is uniform on $[0,1]$, while $h>1$ is polarized with more mass at low and high $\theta$. The choice of $\gamma^L$ varies: Figure~\ref{fig:4} uses $\gamma^L=5$ as in Section~\ref{sec:numerics-finite}, while Figure~\ref{fig:5} uses $\gamma^L \approx 1.254893$ chosen to produce multiple mean field equilibria.

Because only $\theta$ is random and its density is piecewise constant, the average common-insurance exposure $\widehat V(p)$, the follower retention response $\widehat q^\zeta(p)$, and the reduced premium objective $\mathcal J^{L,\mathrm{MF}}(p)$ admit an explicit piecewise characterization (see \ref{app:mf_closed_form}). The reduced premium objective $\mathcal J^{L,\mathrm{MF}}$ can then be maximized interval by interval on $[p^{\mathrm{MF},\min},p^{\mathrm{MF},\max}]=[1,1.580]$, rather than by a generic grid search over the entire interval.  We use this characterization to compute the mean field quantities shown in Figures~\ref{fig:4} and~\ref{fig:5}.

To illustrate the convergence results, we construct finite-player games by setting $a^i=1$, $\gamma^i=1$, $v^{i,0}=0.7$, and $v^i=0.3$, and choosing the relative performance parameters $\theta^1,\ldots,\theta^N$ by the midpoint quantile rule associated with $f_\theta^h$. The corresponding empirical distributions of insurers' types $\mathfrak m^N$ then converge weakly to the mean field type law $\mathfrak m$. We compare the finite-player reduced premium objectives, optimal premiums, and equilibrium retention distributions with their mean field counterparts.

Figure~\ref{fig:4} considers the uniform case $h=1$, with $\gamma^L=5$. Panel~(a) shows the reduced premium objectives. The mean field objective $\mathcal J^{L,\mathrm{MF}}$ is single-peaked and has a unique optimal premium $p^{\mathrm{MF},*}\approx1.189$. The finite-player objectives $\mathcal J^{L,N}$ approach $\mathcal J^{L,\mathrm{MF}}$, and their maximizers approach $p^{\mathrm{MF},*}$ as $N$ increases. Panel~(b) shows the corresponding equilibrium retention distributions. At the unique Stackelberg mean field equilibrium, all types retain partially, so the mean field retention distribution has no mass at $0$ or $1$. The empirical retention distributions induced by the finite-player Stackelberg--Nash equilibria converge to this mean field retention distribution. This example illustrates Proposition~\ref{prop:J_uniform_convergence}, Corollary~\ref{cor:optimal_premium_unique}, and Corollary~\ref{cor:unique_equilibrium} in the case with a unique Stackelberg mean field equilibrium.

\begin{figure}[!htb]
    \centering
    \begin{subfigure}{0.45\textwidth}
        \centering
        \includegraphics[width=\linewidth]{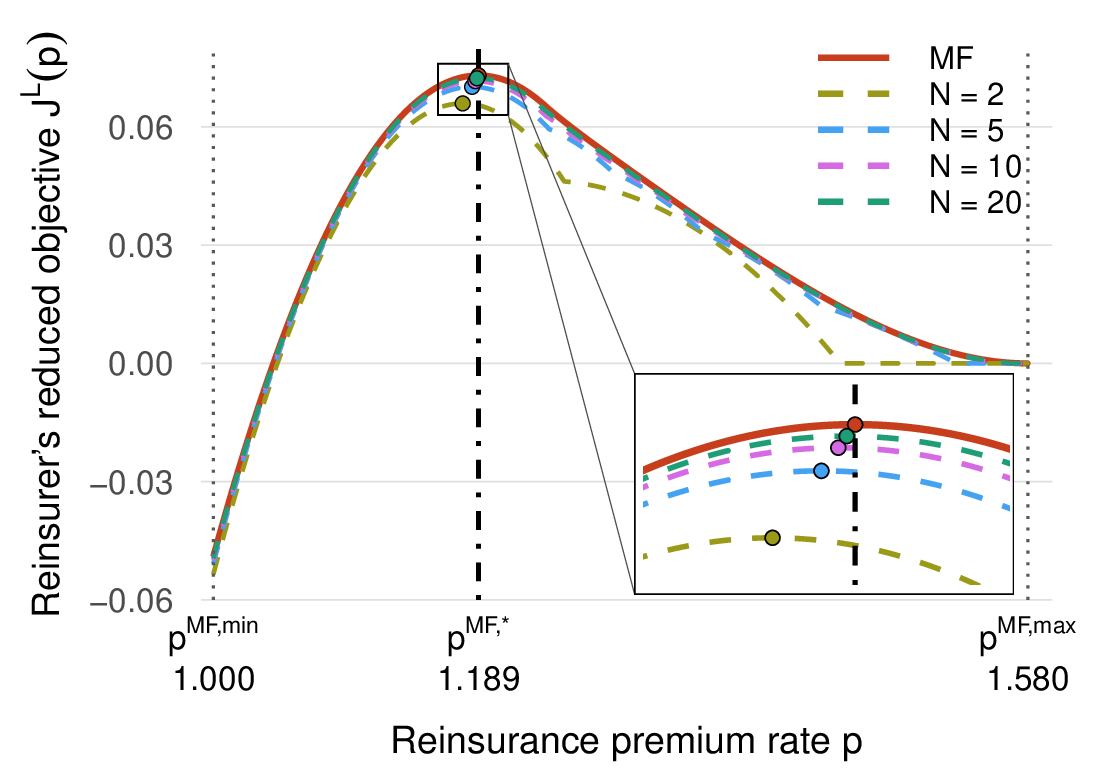}
        \caption{Convergence of reduced premium objectives}
    \end{subfigure}
    \hfill
    \begin{subfigure}{0.45\textwidth}
        \centering
        \includegraphics[width=\linewidth]{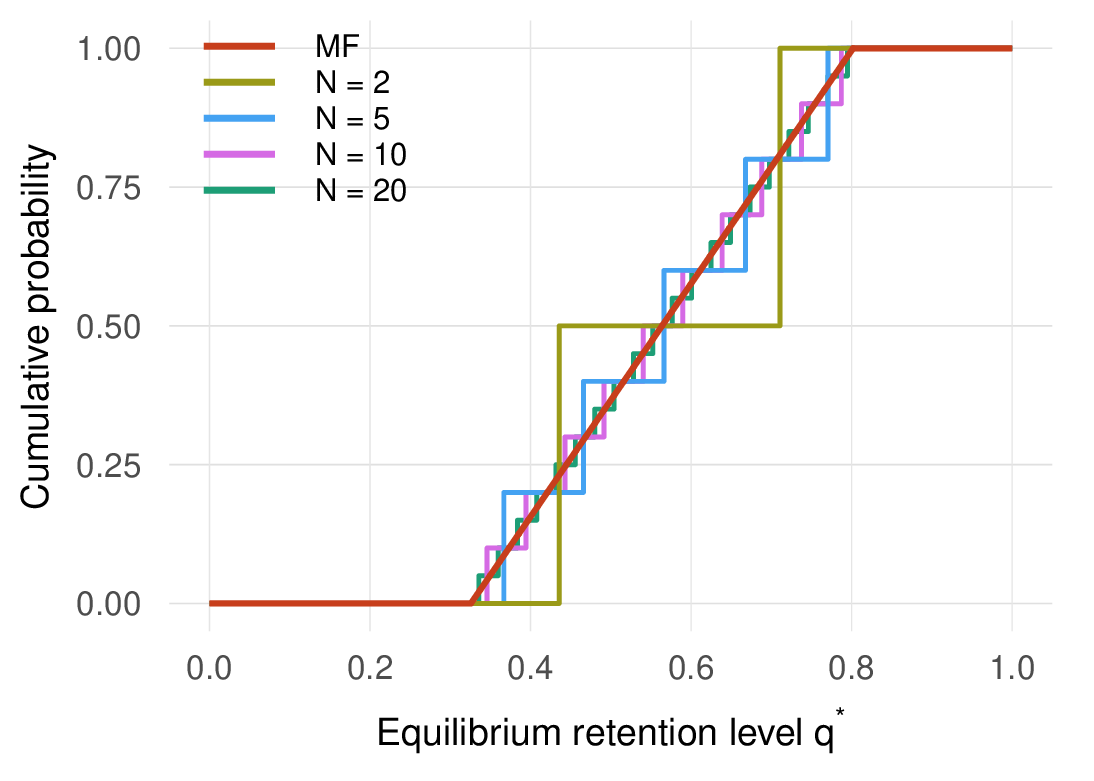}
        \caption{Convergence of equilibrium retention distributions}
    \end{subfigure}
    \caption{Mean field equilibrium and finite-player convergence under $h = 1$. In the mean field model, heterogeneity enters only through the relative performance parameter $\theta\sim\mathrm{Uniform}(0,1)$, and the remaining insurer parameters are fixed at $a=1$, $\gamma=1$, $v^0=0.7$, and $v=0.3$, with $\gamma^L=5$. For the finite-player games, each insurer has $a^i=1$, $\gamma^i=1$, $v^{i,0}=0.7$, and $v^i=0.3$, while $\theta^i=(i-0.5)/N$, $i=1,\ldots,N$ is chosen by the midpoint quantile rule.}
    \label{fig:4}
\end{figure}

Figure~\ref{fig:5} considers the polarized case $h=2.3$, with $\gamma^L\approx1.254893$. In contrast to the uniform case in Figure~\ref{fig:4}, the Stackelberg mean field equilibrium is not unique. Panel~(a) shows the reduced premium objectives. The mean field objective $\mathcal J^{L,\mathrm{MF}}$ has two maximizers, located approximately at $p=1.229$ and $p=1.294$, which induce two distinct Stackelberg mean field equilibria. The finite-player objectives $\mathcal J^{L,N}$ approach $\mathcal J^{L,\mathrm{MF}}$ as $N$ increases, and the optimal premiums of the finite-player Stackelberg--Nash equilibria shown in the figure move toward the higher mean field optimal premium. Panel~(b) compares the retention distributions induced by the two Stackelberg mean field equilibria with the empirical retention distributions induced by the finite-player Stackelberg--Nash equilibria. The two mean field equilibria have different retention configurations: at the lower-premium equilibrium, all types retain partially, whereas at the higher-premium equilibrium, the retention distribution has a jump at full retention. The empirical retention distributions move toward the retention distribution associated with the higher-premium Stackelberg mean field equilibrium. This example is consistent with Proposition~\ref{prop:optimal_premium_convergence} and Theorem~\ref{thm:empirical_profile_convergence} in the case of nonunique Stackelberg mean field equilibria.

\begin{figure}[!htb]
    \centering
    \begin{subfigure}{0.45\textwidth}
        \centering
        \includegraphics[width=\linewidth]{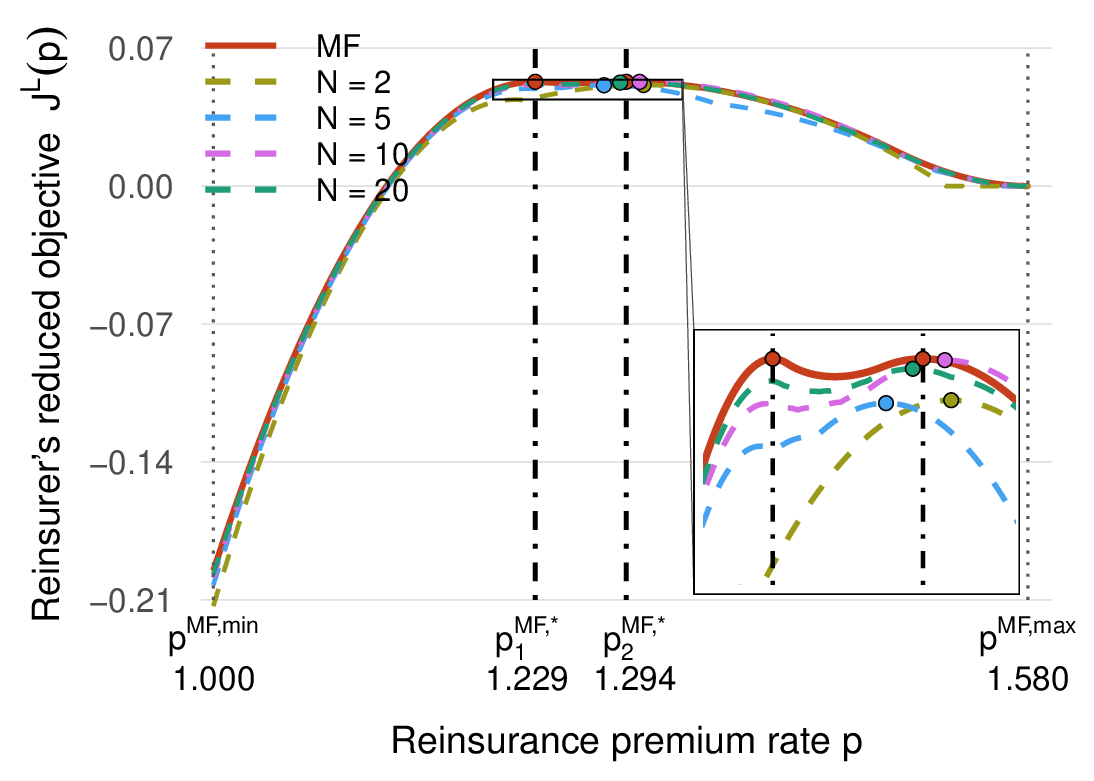}
        \caption{Convergence of reduced premium objectives}
    \end{subfigure}
    \hfill
    \begin{subfigure}{0.45\textwidth}
        \centering
        \includegraphics[width=\linewidth]{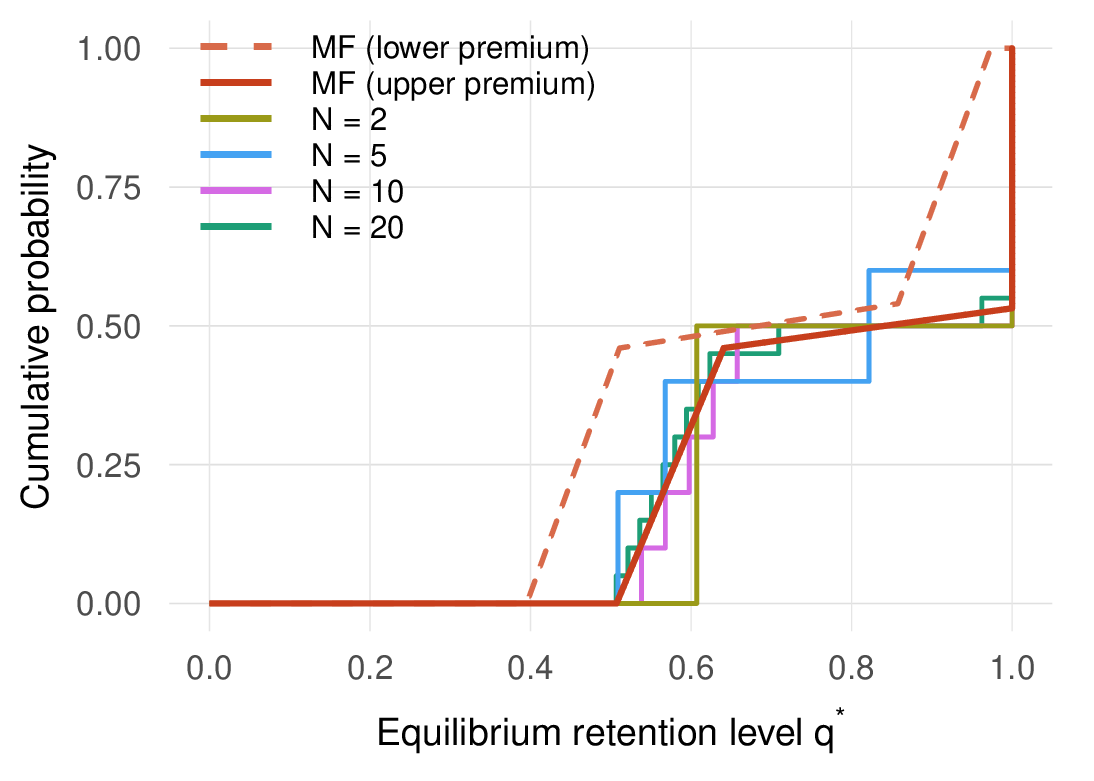}
        \caption{Convergence of equilibrium retention distributions}
    \end{subfigure}
    \caption{Mean field equilibria and finite-player convergence under $h = 2.3$. In the mean field model, heterogeneity enters only through the relative performance parameter $\theta$, whose density is given by \eqref{eq:theta_density_h} with $h=2.3$, and the remaining insurer parameters are fixed at $a=1$, $\gamma=1$, $v^0=0.7$, and $v=0.3$, with $\gamma^L\approx1.254893$. For the finite-player games, each insurer has $a^i=1$, $\gamma^i=1$, $v^{i,0}=0.7$, $v^i=0.3$, and $\theta^i$ is chosen by the midpoint quantile rule associated with the density \eqref{eq:theta_density_h}.}
    \label{fig:5}
\end{figure}

The finite-player optimal premiums shown in Figure~\ref{fig:5}(a) approach the higher of the two mean field optimal premiums, but Proposition~\ref{prop:optimal_premium_convergence} does not imply convergence to this particular value. When the mean field optimal premium is not unique, the proposition only implies that every subsequential limit of finite-player optimal premiums is a mean field optimal premium; it does not imply convergence of the full sequence or determine which of the two mean field optimal premiums is approached. To examine whether this observation depends on how the empirical type distributions are constructed, we first replaced the midpoint quantile rule in Figure~\ref{fig:5}(a) by left- and right-point quantile rules; in both cases, $p^{N,*}$ still moves toward the higher mean field optimal premium for $N=2,\ldots,50$. Figure~\ref{fig:6} then reports two further constructions for the same polarized example, with heterogeneity only through $\theta$ and $\gamma^L\approx1.254893$, again computing $p^{N,*}$ for $N=2,\ldots,50$. In Figure~\ref{fig:6}(a), the midpoint quantile rule is used as in Figure~\ref{fig:5}, but with $h$ replaced by $h_N\equiv2.3$, $h_N=2.3-0.2N^{-1/4}$ (so that $h_N\uparrow2.3$), and $h_N=2.3+0.2N^{-1/4}$ (so that $h_N\downarrow2.3$). Since $h_N\to2.3$ in all three cases, the corresponding empirical type distributions converge weakly to the same mean field type distribution, namely the one corresponding to $h=2.3$. Over the plotted range, however, $p^{N,*}$ stays near the higher mean field optimal premium for $h_N\equiv2.3$ and $h_N\downarrow2.3$, while it moves toward the lower one for $h_N\uparrow2.3$. In Figure~\ref{fig:6}(b), the relative performance parameters are generated from a nested i.i.d. sample from the distribution with the density \eqref{eq:theta_density_h} with $h=2.3$, using the first $N$ samples for each $N$. The four sequences show sizable finite-$N$ fluctuations; depending on the nested sample, $p^{N,*}$ may stay near either mean field optimal premium or move between their neighborhoods. These computations do not identify the finite-player limits. They illustrate that, even when $\mathfrak m^N$ converges weakly to the same mean field type law $\mathfrak m$, the values of $p^{N,*}$ over the plotted range may stay near different mean field optimal premiums, depending on the construction of $\mathfrak m^N$.

\begin{figure}[!htb]
    \centering
    \begin{subfigure}{0.45\textwidth}
        \centering
        \includegraphics[width=\linewidth]{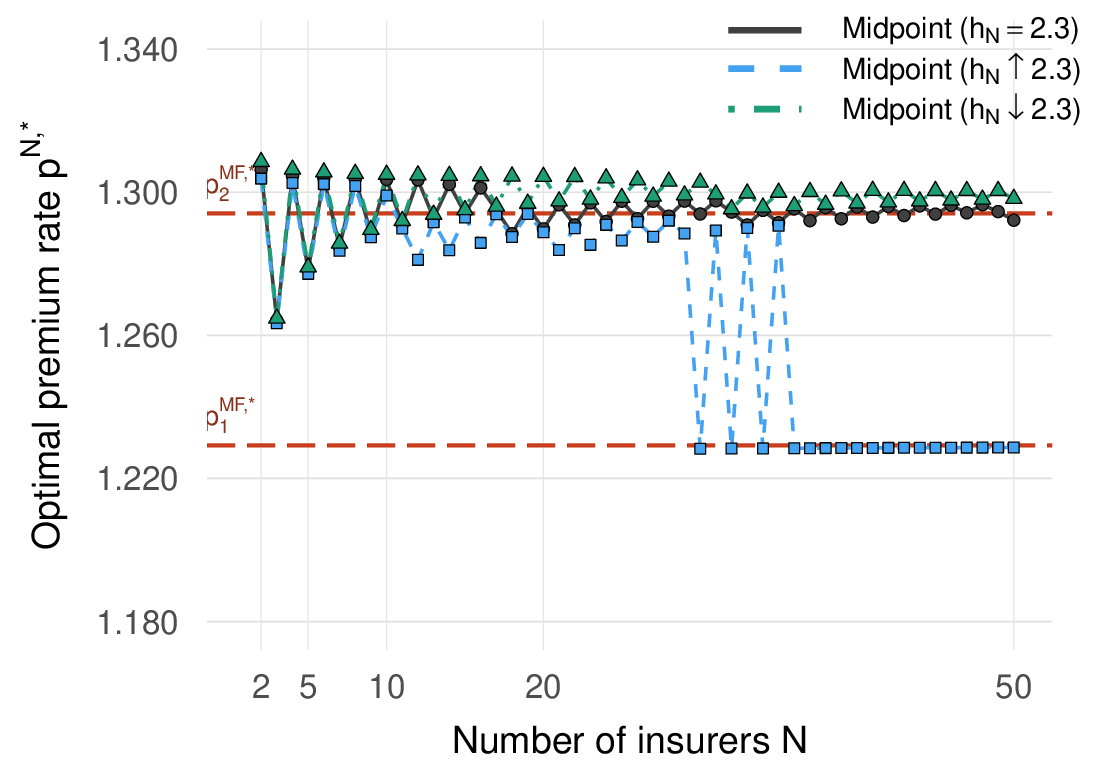}
        \caption{Optimal premiums under $N$-dependent density parameters}
    \end{subfigure}
    \hfill
    \begin{subfigure}{0.45\textwidth}
        \centering
        \includegraphics[width=\linewidth]{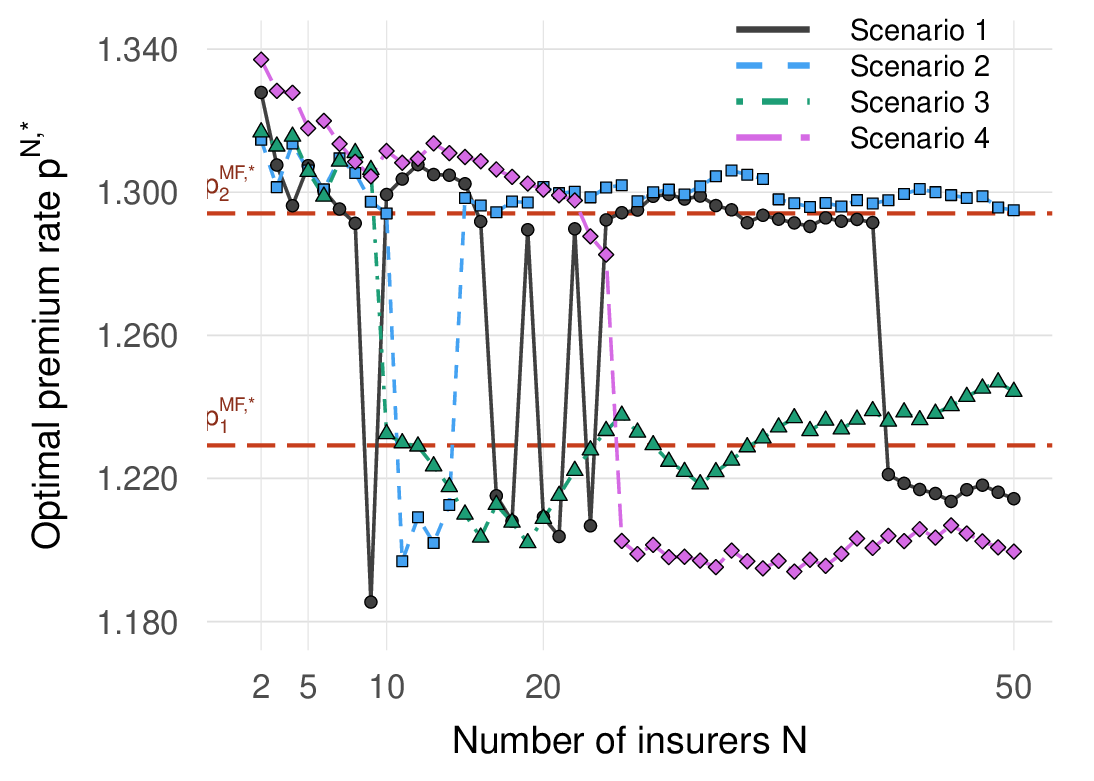}
        \caption{Optimal premiums under nested i.i.d. samples}
    \end{subfigure}
    \caption{Finite-player premium diagnostics for the nonunique mean field example in Figure~\ref{fig:5}. The dashed horizontal lines mark the two mean field optimal premiums. Panel~(a) uses the midpoint quantile rule as in Figure~\ref{fig:5}, with $h_N\equiv2.3$, $h_N=2.3-0.2N^{-1/4}$, and $h_N=2.3+0.2N^{-1/4}$. Panel~(b) uses four nested i.i.d. samples from the distribution with the density \eqref{eq:theta_density_h} with $h=2.3$, using the first $N$ samples for each $N$.}
    \label{fig:6}
\end{figure}

\section{Conclusion} 
\label{sec:conclusion}

We studied endogenous reinsurance pricing in a competitive market with one strategic reinsurer and many heterogeneous insurers, in both finite-player and mean field settings. Under constant controls, the insurers' equilibrium retention response at a fixed premium is characterized by a scalar fixed point for the average common-insurance exposure, and is monotone in the premium with a threshold structure. This structure reduces the reinsurer's problem to a one-dimensional optimization over a compact interval and yields constant Stackelberg equilibria in both settings. In the finite-player case, it also leads to a threshold continuation procedure that computes equilibrium premiums in quadratic time, without enumerating all retention configurations. The characterization further reveals a spillover channel through which relative performance concerns can lead insurers to retain risk even when reinsurance remains actuarially favorable.

We also proved convergence from the finite-player games to the mean field game under weak convergence of the empirical type distributions, without requiring the mean field equilibrium to be unique. The convergence is formulated at the level of optimal premium sets and subsequential empirical type-control distributions: every subsequential limit of finite-player equilibrium type-control distributions is induced by a Stackelberg mean field equilibrium. This provides a forward justification for the mean field approximation in the present Stackelberg reinsurance model.

Several extensions remain open. One direction is to impose solvency or ruin constraints over the contract period, motivated, for instance, by risk-based capital or Solvency~II-type requirements. Such constraints would require controlling not only terminal wealth but also the path of wealth or solvency ratios. Another direction is to relax the common-premium assumption by allowing cedent-specific premiums, which would make the leader's premium choice multidimensional. A further extension is to introduce multiple strategic reinsurers, leading to competition among leaders. Time-dependent controls and more general reinsurance contracts are also natural directions to extend. These extensions would enrich the market structure, but they would also reduce the tractability, economic interpretability, and computational simplicity of the present framework.

Uniqueness and selection of Stackelberg equilibria remain open. The leader's problems in the finite-player and mean field models may have multiple maximizers, and hence Stackelberg equilibria need not be unique. Identifying useful conditions that ensure uniqueness remains open. When the mean field equilibrium is not unique, Theorem~\ref{thm:empirical_profile_convergence} shows that subsequential limits of finite-player equilibrium type-control distributions are induced by Stackelberg mean field equilibria, but it does not determine which equilibrium is approached. Whether different sequences of finite-player type distributions $\mathfrak m^N$ converging weakly to the same mean field type law $\mathfrak m$ can lead to different finite-player limits, or select different Stackelberg mean field equilibria, remains open.

\appendix
\allowdisplaybreaks
\section{Proofs in Section~\ref{sec:finite_game}}
\subsection{Proof of Proposition \ref{prop:follower_response}}
\label{app:2.1}
Fix $p\ge0$. For fixed controls of the other insurers, insurer $i$'s reduced objective \eqref{eq:Jtilde_expanded} separates into two strictly concave quadratic functions, one in $q^i$ and the other in $\pi^i$. Thus, the projected first-order condition for $q^i\in[0,1]$ and the first-order condition for $\pi^i\in\mathbb R$ are necessary and sufficient for insurer $i$'s optimality.

For $i=1,\ldots,N$, the first-order conditions, together with the constraint $q^i\in[0,1]$, give the best responses
\begin{equation}
\label{eq:best_response}
q^{i}=\Pi_{[0,1]}\Big(\frac{(p-a^{i})+\frac{1}{\gamma^{i}}\theta^{i}v^{i,0}V_{-i}^{N}}{\frac{1}{\gamma^{i}}\big(1-\theta^{i}/N\big)\big((v^{i,0})^{2}+(v^{i})^{2}\big)}\Big),\qquad\pi^{i}=\frac{\mu^{i}+\frac{1}{\gamma^{i}}\theta^{i}\sigma^{i,0}\Sigma_{-i}^{N}}{\frac{1}{\gamma^{i}}\big(1-\theta^{i}/N\big)\big((\sigma^{i,0})^{2}+(\sigma^{i})^{2}\big)},
\end{equation}
where $V_{-i}^{N}=\frac{1}{N}\sum_{j\neq i}q^{j}v^{j,0}$, $\Sigma_{-i}^{N}=\frac{1}{N}\sum_{j\neq i}\pi^{j}\sigma^{j,0}$.
At equilibrium, these best responses must be mutually consistent. Writing $(\widehat q^{1},\widehat\pi^{1}),\ldots,(\widehat q^{N},\widehat\pi^{N})$ for an equilibrium response and using \eqref{eq:Vhat_Sigmahat_def}, we have
\begin{equation*}
\frac{1}{N}\sum_{j\neq i}\widehat{q}^{j}v^{j,0}=\widehat{V}^{N}-\frac{1}{N}\widehat{q}^{i}v^{i,0},\qquad\frac{1}{N}\sum_{j\neq i}\widehat{\pi}^{j}\sigma^{j,0}=\widehat{\Sigma}^{N}-\frac{1}{N}\widehat{\pi}^{i}\sigma^{i,0}.
\end{equation*}
Therefore the best-response identities \eqref{eq:best_response} become
\begin{equation}
\label{eq:equilibrium_best_response}
\widehat{q}^{i}=\Pi_{[0,1]}\Big(\frac{(p-a^{i})+\frac{1}{\gamma^{i}}\theta^{i}v^{i,0}\big(\widehat{V}^{N}-\frac{1}{N}\widehat{q}^{i}v^{i,0}\big)}{\frac{1}{\gamma^{i}}\big(1-\theta^{i}/N\big)\big((v^{i,0})^{2}+(v^{i})^{2}\big)}\Big),\qquad\widehat{\pi}^{i}=\frac{\mu^{i}+\frac{1}{\gamma^{i}}\theta^{i}\sigma^{i,0}\big(\widehat{\Sigma}^{N}-\frac{1}{N}\widehat{\pi}^{i}\sigma^{i,0}\big)}{\frac{1}{\gamma^{i}}\big(1-\theta^{i}/N\big)\big((\sigma^{i,0})^{2}+(\sigma^{i})^{2}\big)}.
\end{equation}

We first solve for the equilibrium investment responses. Rearranging the second identity in \eqref{eq:equilibrium_best_response} gives
\begin{equation*}
\widehat{\pi}^{i}=\frac{\mu^{i}+\frac{1}{\gamma^{i}}\theta^{i}\sigma^{i,0}\widehat{\Sigma}^{N}}{\frac{1}{\gamma^{i}}\big[(\sigma^{i,0})^{2}+(\sigma^{i})^{2}(1-\theta^{i}/N)\big]}.
\end{equation*}
Multiplying this identity by $\sigma^{i,0}/N$ and summing over $i=1,\ldots,N$ closes the system in the average common-financial exposure $\widehat\Sigma^N$ defined in \eqref{eq:Vhat_Sigmahat_def}. We obtain
\begin{equation*}
\widehat{\Sigma}^{N}=\frac{1}{N}\sum_{j=1}^{N}\frac{\gamma^{j}\sigma^{j,0}\mu^{j}}{(\sigma^{j,0})^{2}+(\sigma^{j})^{2}(1-\theta^{j}/N)}+\widehat{\Sigma}^{N}\frac{1}{N}\sum_{j=1}^{N}\frac{\theta^{j}(\sigma^{j,0})^{2}}{(\sigma^{j,0})^{2}+(\sigma^{j})^{2}(1-\theta^{j}/N)},
\end{equation*}
so $\widehat\Sigma^N$ is uniquely determined and is given by \eqref{eq:pihat}. Consequently, the investment responses $\widehat\pi^1,\ldots,\widehat\pi^N$ are uniquely determined.

We next solve for the equilibrium retention responses. The first identity in \eqref{eq:equilibrium_best_response} has the form
\begin{equation*}
\widehat{q}^{i}=\Pi_{[0,1]}\Big(\frac{A-B\widehat{q}^{i}}{C}\Big),
\end{equation*}
where $A=(p-a^{i})+\frac{1}{\gamma^{i}}\theta^{i}v^{i,0}\widehat{V}^{N}$, $B=\frac{1}{\gamma^{i}}\frac{\theta^{i}}{N}(v^{i,0})^{2}\ge0$, $C=\frac{1}{\gamma^{i}}\big(1-\theta^{i}/N\big)\big((v^{i,0})^{2}+(v^{i})^{2}\big)>0$.
Applying the elementary identity
$x=\Pi_{[0,1]}\big(\frac{A-Bx}{C}\big)\Longleftrightarrow x=\Pi_{[0,1]}\big(\frac{A}{B+C}\big)$, $ A\in\mathbb{R},\ B\ge0,\ C>0$,
which is verified by considering the three cases $x=0$, $0<x<1$, and $x=1$, we obtain
\begin{equation*}
\widehat{q}^{i}=\Pi_{[0,1]}\Big(\frac{(p-a^{i})+\frac{1}{\gamma^{i}}\theta^{i}v^{i,0}\widehat{V}^{N}}{\frac{1}{\gamma^{i}}\big[(v^{i,0})^{2}+(v^{i})^{2}(1-\theta^{i}/N)\big]}\Big)=\Pi_{[0,1]}\big(I^{i}(p,\widehat{V}^{N})\big).
\end{equation*}
Multiplying this projection formula by $v^{i,0}/N$ and summing over $i=1,\ldots,N$ closes the retention system in the average common-insurance exposure $\widehat V^N$ defined in \eqref{eq:Vhat_Sigmahat_def}. This gives the fixed-point equation \eqref{eq:fixed_point_V}.

It remains to prove that \eqref{eq:fixed_point_V} admits a unique solution. To this end, define
\begin{equation}
\label{eq:fixed_point_map}
F^{N}(p,V):=\frac{1}{N}\sum_{j=1}^{N}v^{j,0}\Pi_{[0,1]}\big(I^{j}(p,V)\big).
\end{equation}
Since $0\le\Pi_{[0,1]}\le1$, the map $V\mapsto F^N(p,V)$ maps $[0,\bar V^N]$ into itself, where $\bar{V}^{N}:=\frac{1}{N}\sum_{j=1}^{N}v^{j,0}$. Moreover, since $\Pi_{[0,1]}$ is $1$-Lipschitz,
\begin{equation*}
|F^{N}(p,V)-F^{N}(p,V')|\le\frac{1}{N}\sum_{j=1}^{N}v^{j,0}|I^{j}(p,V)-I^{j}(p,V')|=\Big(\frac{1}{N}\sum_{j=1}^{N}\frac{\theta^{j}(v^{j,0})^{2}}{(v^{j,0})^{2}+(v^{j})^{2}(1-\theta^{j}/N)}\Big)|V-V'|.
\end{equation*}
The coefficient multiplying $|V-V'|$ is strictly less than one, since each summand is at most one, with equality possible only in the excluded case $N=\theta^1=1$. Hence $F^N(p,\cdot)$ is a contraction on $[0,\bar V^N]$, and by Banach's fixed point theorem, it has a unique fixed point. This proves the existence and uniqueness of $\widehat V^N$, and then \eqref{eq:qhat_projection} uniquely determines $\widehat q^1,\ldots,\widehat q^N$. Restoring the dependence on the fixed premium $p$ gives the notation $\widehat q^i(p)$ in \eqref{eq:qhat_projection}.

\subsection{Proof of Proposition~\ref{prop:premium_response}}
\label{app:2.2}
We first prove that $p\mapsto\widehat V^N(p)$ is continuous and monotonically increasing. Recall the fixed-point map in \eqref{eq:fixed_point_map}. Since $\Pi_{[0,1]}$ is $1$-Lipschitz, for any $p,p'\ge0$ and $V,V'\in[0,\bar V^N]$,
\begin{equation*}
|F^{N}(p,V)-F^{N}(p',V')|\le\frac{1}{N}\sum_{j=1}^{N}v^{j,0}\Big|\frac{\theta^{j}v^{j,0}(V-V')+\gamma^{j}(p-p')}{(v^{j,0})^{2}+(v^{j})^{2}(1-\theta^{j}/N)}\Big|\le K_{1}^{N}|V-V'|+K_{2}^{N}|p-p'|,
\end{equation*}
where $
K_{1}^{N}:=\frac{1}{N}\sum_{j=1}^{N}\frac{\theta^{j}(v^{j,0})^{2}}{(v^{j,0})^{2}+(v^{j})^{2}(1-\theta^{j}/N)}\in[0,1)$, $K_{2}^{N}:=\frac{1}{N}\sum_{j=1}^{N}\frac{\gamma^{j}v^{j,0}}{(v^{j,0})^{2}+(v^{j})^{2}(1-\theta^{j}/N)}>0$.
Using the fixed-point relation $\widehat{V}^{N}(p)=F^{N}(p,\widehat{V}^{N}(p))$, we obtain
\begin{equation*}
|\widehat{V}^{N}(p)-\widehat{V}^{N}(p')|\le K_{1}^{N}|\widehat{V}^{N}(p)-\widehat{V}^{N}(p')|+K_{2}^{N}|p-p'|.
\end{equation*}
Hence $
|\widehat{V}^{N}(p)-\widehat{V}^{N}(p')|\le\frac{K_{2}^{N}}{1-K_{1}^{N}}|p-p'|$, and $p\mapsto\widehat V^N(p)$ is continuous.

To prove monotonicity, let $p_1\le p_2$. Since $F^N$ is monotonically increasing in $p$,
\begin{equation*}
\widehat{V}^{N}(p_{1})=F^{N}(p_{1},\widehat{V}^{N}(p_{1}))\le F^{N}(p_{2},\widehat{V}^{N}(p_{1})).
\end{equation*}
Since $V\mapsto F^N(p_2,V)$ is monotonically increasing, repeated application gives
\begin{equation*}
\widehat{V}^{N}(p_{1})\le F^{N}\big(p_{2},\widehat{V}^{N}(p_{1})\big)\le F^{N}\big(p_{2},F^{N}(p_{2},\widehat{V}^{N}(p_{1}))\big)\le\cdots.
\end{equation*}
The iterates converge to $\widehat V^N(p_2)$, the unique fixed point of $V=F^N(p_2,V)$, by the contraction property of $V\mapsto F^N(p_2,V)$. Hence $
\widehat V^N(p_1)\le\widehat V^N(p_2)$,
and $p\mapsto\widehat V^N(p)$ is monotonically increasing. Since $I^i(p,V)$ and $\Pi_{[0,1]}$ are continuous and monotonically increasing in their arguments, \eqref{eq:qhat_projection} implies that each map $p\mapsto\widehat q^i(p)$ is continuous and monotonically increasing on $[0,\infty)$.

We now compute the population-wide thresholds.
By \eqref{eq:qhat_projection} and the definition of $\widehat V^N(p)$,
\begin{equation*}
\widehat{q}^{i}(p)=0,\quad i=1,\ldots,N\quad\Longleftrightarrow\quad\widehat{V}^{N}(p)=0\quad\Longleftrightarrow\quad I^{i}(p,0)\le0,\quad i=1,\ldots,N.
\end{equation*}
Since $I^i(p,0)\le0$ if and only if $p\le a^i$, we obtain 
\begin{equation*}
p^{\min}:=\max\{p\ge0:\widehat{q}^{i}(p)=0,\ i=1,\dots,N\}=\min_{1\le i\le N}a^{i}.
\end{equation*}
Similarly,
\begin{equation*}
\widehat{q}^{i}(p)=1,\quad i=1,\ldots,N\quad\Longleftrightarrow\quad\widehat{V}^{N}(p)=\bar{V}^{N}\quad\Longleftrightarrow\quad I^{i}(p,\bar{V}^{N})\ge1,\quad i=1,\ldots,N.
\end{equation*}
Solving $I^i(p,\bar V^N)\ge1$ for every $i$ gives
\begin{equation*}
p\ge a^{i}+\frac{1}{\gamma^{i}}\big[(v^{i,0})^{2}+(v^{i})^{2}(1-\theta^{i}/N)\big]-\frac{1}{\gamma^{i}}\theta^{i}v^{i,0}\bar{V}^{N},\qquad i=1,\ldots,N.
\end{equation*}
Therefore
\begin{equation*}
p^{\max}:=\min\{p\ge0:\widehat{q}^{i}(p)=1,\ i=1,\dots,N\}=\max_{1\le i\le N}\Big\{ a^{i}+\frac{1}{\gamma^{i}}\big[(v^{i,0})^{2}+(v^{i})^{2}(1-\theta^{i}/N)\big]-\frac{1}{\gamma^{i}}\theta^{i}v^{i,0}\bar{V}^{N}\Big\}.
\end{equation*}
This proves \eqref{eq:pminmax}.

Since $p^{\min}$ and $p^{\max}$ are well-defined, and each map $p\mapsto\widehat q^i(p)$ is continuous and monotonically increasing, the definitions of $p^{\min}$ and $p^{\max}$ imply $p^{\min}<p^{\max}$, and 
\begin{equation*}
\widehat{q}^{i}(p)=0\quad\text{for }p\le p^{\min},\qquad\widehat{q}^{i}(p)=1\quad\text{for }p\ge p^{\max},\qquad i=1,\ldots,N.
\end{equation*}

For each $i=1,\ldots,N$, since $\widehat q^i(p)=0$ for $p\le p^{\min}$ and $\widehat q^i(p)=1$ for $p\ge p^{\max}$, the sets in \eqref{eq:premium_thresholds} are nonempty and $p^{i,\min},p^{i,\max}\in[p^{\min},p^{\max}]$. By the continuity and monotonicity of $p\mapsto\widehat q^i(p)$, $\widehat q^i(p^{i,\min})=0$ and $\widehat q^i(p^{i,\max})=1$, hence $p^{i,\min}<p^{i,\max}$. The monotonicity of $p\mapsto\widehat q^i(p)$ gives \eqref{eq:qhat_threshold_structure}.

\subsection{Proof of Proposition \ref{prop:leader_equilibrium}}
\label{app:2.3}
By \eqref{eq:Jtilde_leader} and \eqref{eq:leader_premium_objective}, the reduced leader objective separates as
\begin{equation*}
\widetilde{J}^{L}(p,\pi^{L})=x^{L}+T\Big[\mu^{L}\pi^{L}-\frac{1}{2\gamma^{L}}\big((\sigma^{L,0})^{2}+(\sigma^{L})^{2}\big)(\pi^{L})^{2}\Big]+T\mathcal{J}^{L}(p).
\end{equation*}
Thus $\pi^L$ and $p$ can be optimized separately. The part depending on $\pi^L$ is a strictly concave quadratic function, and its unique maximizer is precisely \eqref{eq:pistar_finite}.

It remains to optimize $\mathcal J^L(p)$ over $p\ge0$. Since each equilibrium retention $p\mapsto \widehat{q}^{i}(p)$ is continuous, the reduced premium objective $\mathcal{J}^{L}$ is continuous on $[0,\infty)$. Furthermore, by Proposition~\ref{prop:premium_response}, $\widehat{q}^{i}(p)=0$ for $p\le p^{\min}$, $\widehat{q}^{i}(p)=1$ for $p\ge p^{\max}$,  $i=1,\ldots,N$.
Consequently, on $[0,p^{\min}]$,
\begin{equation*}
\mathcal{J}^{L}(p)=\frac{1}{N}\sum_{i=1}^{N}(p-a^{i})-\frac{1}{2\gamma^{L}}\Big[\Big(\frac{1}{N}\sum_{i=1}^{N}v^{i,0}\Big)^{2}+\frac{1}{N^{2}}\sum_{i=1}^{N}(v^{i})^{2}\Big],
\end{equation*}
so $\mathcal J^L$ is affine and strictly increasing there. On $[p^{\max},\infty)$, $
\mathcal{J}^{L}(p)=0$.
Hence, no premium outside $[p^{\min},p^{\max}]$ can improve $\mathcal{J}^{L}$, and therefore
\begin{equation*}
\sup_{p\ge0}\mathcal{J}^{L}(p)=\sup_{p\in[p^{\min},\,p^{\max}]}\mathcal{J}^{L}(p).
\end{equation*}
Since $\mathcal{J}^{L}$ is continuous and $[p^{\min},\,p^{\max}]$ is compact, the extreme value theorem implies that the argmax set in \eqref{eq:pstar_argmax} is nonempty. Let $p^*$ be any maximizer, and define
\begin{equation*}
q^{i,*}:=\widehat{q}^{i}(p^*),\qquad \pi^{i,*}:=\widehat \pi^{i},\qquad i=1,\dots,N.
\end{equation*}
Then $(q^{1,*},\pi^{1,*}),\dots,(q^{N,*},\pi^{N,*})$ is a constant follower Nash equilibrium induced by $p^*$, while $(p^*,\pi^{L,*})$ is optimal for the reinsurer. Therefore, $(p^{*},\pi^{L,*},q^{1,*},\pi^{1,*},\dots,q^{N,*},\pi^{N,*})$ is a constant Stackelberg--Nash equilibrium.

\section{Proofs in Section~\ref{sec:smfg}}
\subsection{Proof of Proposition~\ref{prop:CMFER}}
\label{app:3.1}
Fix $p\ge0$. For a candidate average 
$\overline{X}_{T}=C+VW_{T}^{0}+\Sigma B_{T}^{0},$
the typewise reduced criterion \eqref{eq:Jtilde_MF_expanded} separates into two strictly concave quadratic functions, one in $q^\zeta$ and the other in $\pi^\zeta$. Hence, for almost every type $\zeta$, the projected first-order condition for $q^\zeta\in[0,1]$ and the first-order condition for $\pi^\zeta\in\mathbb R$ are necessary and sufficient for optimality. Thus, the typewise best response to this candidate average is
\begin{equation}
\label{eq:MF_best_response}
q^{\zeta}=\Pi_{[0,1]}\big(I^{\zeta}(p,V)\big),\qquad\pi^{\zeta}=\frac{\mu+\frac{1}{\gamma}\theta\sigma^{0}\Sigma}{\frac{1}{\gamma}\big((\sigma^{0})^{2}+\sigma^{2}\big)}.
\end{equation}
For this best response to be a follower mean field equilibrium response, the candidate average must coincide with the population average generated by the response. By the consistency condition \eqref{eq:MF_consistency_condition} and the affine representation \eqref{eq:MF_average_affine_form}, this is equivalent to
\begin{equation*}
V=\mathbb{E}[q^{\zeta}v^{0}],\qquad\Sigma=\mathbb{E}[\pi^{\zeta}\sigma^{0}].
\end{equation*}

We first solve the consistency equation for the average common-financial exposure. Substituting the second identity in \eqref{eq:MF_best_response} into $\Sigma=\mathbb E[\pi^\zeta\sigma^0]$ gives
\begin{equation*}
\Sigma=\mathbb{E}\Big[\frac{\gamma\mu\sigma^{0}}{(\sigma^{0})^{2}+\sigma^{2}}\Big]+\Sigma\,\mathbb{E}\Big[\frac{\theta(\sigma^{0})^{2}}{(\sigma^{0})^{2}+\sigma^{2}}\Big].
\end{equation*}
Since $0\le\theta\le1$ and $\sigma>0$ a.s., $
\mathbb{E}\big[\frac{\theta(\sigma^{0})^{2}}{(\sigma^{0})^{2}+\sigma^{2}}\big]<1$.
Therefore, the consistency equation for $\Sigma$ has the unique solution
\begin{equation*}
\widehat{\Sigma}=\frac{\mathbb{E}\big[\gamma\mu\sigma^{0}/((\sigma^{0})^{2}+\sigma^{2})\big]}{1-\mathbb{E}\big[\theta(\sigma^{0})^{2}/((\sigma^{0})^{2}+\sigma^{2})\big]}.
\end{equation*}

We next solve the consistency equation for the average common-insurance exposure. Define
\begin{equation}
\label{eq:MF_fixed_point_map}
F(p,V):=\mathbb{E}\big[v^{0}\Pi_{[0,1]}\big(I^{\zeta}(p,V)\big)\big].
\end{equation}
Then the consistency equation for $V$ is $V=F(p,V)$. Since $0\le\Pi_{[0,1]}\le1$, the map $V\mapsto F(p,V)$ sends $[0,\bar V]$ into itself, where $\bar V:=\mathbb E[v^0]$. Moreover, since $\Pi_{[0,1]}$ is $1$-Lipschitz, for any $V,V'\in[0,\bar V]$,
\begin{equation*}
|F(p,V)-F(p,V')|\le\mathbb{E}\big[v^{0}|I^{\zeta}(p,V)-I^{\zeta}(p,V')|\big]=\mathbb{E}\Big[\frac{\theta(v^{0})^{2}}{(v^{0})^{2}+v^{2}}\Big]|V-V'|.
\end{equation*}
The coefficient multiplying $|V-V'|$ is strictly less than one, because $0\le\theta\le1$ and $v>0$ a.s. Therefore, $V\mapsto F(p,V)$ is a contraction on $[0,\bar V]$ and has a unique fixed point. Denote this fixed point by $\widehat V$.

Substituting the unique solutions $\widehat V$ and $\widehat\Sigma$ into \eqref{eq:MF_best_response} gives
\begin{equation*}
\widehat{q}^{\zeta}=\Pi_{[0,1]}\big(I^{\zeta}(p,\widehat{V})\big),\qquad\widehat{\pi}^{\zeta}=\frac{\mu+\frac{1}{\gamma}\theta\sigma^{0}\widehat{\Sigma}}{\frac{1}{\gamma}\big((\sigma^{0})^{2}+\sigma^{2}\big)},
\end{equation*}
and by the consistency equations, these quantities satisfy
\begin{equation*}
\widehat{V}=\mathbb{E}[\widehat{q}^{\zeta}v^{0}],\qquad\widehat{\Sigma}=\mathbb{E}[\widehat{\pi}^{\zeta}\sigma^{0}].
\end{equation*}
Restoring the dependence on the fixed premium $p$ gives the notation $\widehat V(p)$ and $\widehat q^\zeta(p)$ in \eqref{eq:Vhat_Sigmahat_MF}--\eqref{eq:qhat_MF}, and the fixed-point equation $V=F(p,V)$ above is precisely \eqref{eq:fixed_point_V_MF}.

By compactness of $\mathcal Z$, the pair $(\widehat q^\zeta,\widehat\pi^\zeta)$ belongs to $\mathcal A^{\mathrm{MF}}$. By construction, it satisfies Definition~\ref{def:CSMFE}, and hence is a constant follower mean field equilibrium response. Uniqueness follows from the uniqueness of the typewise best response \eqref{eq:MF_best_response} for fixed $V$ and $\Sigma$, together with the uniqueness of the solutions to the consistency equations for $V$ and $\Sigma$.

\subsection{Proof of Proposition~\ref{prop:premium_response_CMFER}}
\label{app:3.2}
We first prove that $p\mapsto\widehat V(p)$ is continuous and monotonically increasing. Recall the fixed-point map in \eqref{eq:MF_fixed_point_map}. Since $\Pi_{[0,1]}$ is $1$-Lipschitz, for any $p,p'\ge0$ and $V,V'\in[0,\bar V]$,
\begin{equation*}
|F(p,V)-F(p',V')|\le\mathbb{E}\Big[v^{0}\Big|\frac{\theta v^{0}(V-V')+\gamma(p-p')}{(v^{0})^{2}+v^{2}}\Big|\Big]\le K_{1}|V-V'|+K_{2}|p-p'|,
\end{equation*}
where $K_{1}:=\mathbb{E}\big[\frac{\theta(v^{0})^{2}}{(v^{0})^{2}+v^{2}}\big]\in[0,1)$,  $K_{2}:=\mathbb{E}\big[\frac{\gamma v^{0}}{(v^{0})^{2}+v^{2}}\big]>0$.
Using the fixed-point relation $\widehat V(p)=F(p,\widehat V(p))$, we obtain
\begin{equation*}
|\widehat{V}(p)-\widehat{V}(p')|\le K_{1}|\widehat{V}(p)-\widehat{V}(p')|+K_{2}|p-p'|.
\end{equation*}
Hence $
|\widehat{V}(p)-\widehat{V}(p')|\le\frac{K_{2}}{1-K_{1}}|p-p'|$,
and $p\mapsto\widehat V(p)$ is continuous.

To prove monotonicity, let $p_1\le p_2$. Since $F$ is monotonically increasing in $p$,
\begin{equation*}
\widehat{V}(p_{1})=F(p_{1},\widehat{V}(p_{1}))\le F(p_{2},\widehat{V}(p_{1})).
\end{equation*}
Since $V\mapsto F(p_2,V)$ is monotonically increasing, repeated application gives
\begin{equation*}
\widehat{V}(p_{1})\le F(p_{2},\widehat{V}(p_{1}))\le F\big(p_{2},F(p_{2},\widehat{V}(p_{1}))\big)\le\cdots.
\end{equation*}
The iterates converge to $\widehat V(p_2)$, the unique fixed point of $V=F(p_2,V)$, by the contraction property of $V\mapsto F(p_2,V)$. Hence $
\widehat{V}(p_{1})\le\widehat{V}(p_{2})$,
and $p\mapsto\widehat V(p)$ is monotonically increasing. Since $I^\zeta(p,V)$ and $\Pi_{[0,1]}$ are continuous and monotonically increasing in their arguments, \eqref{eq:qhat_MF} implies that, for almost every type $\zeta$, the map $p\mapsto\widehat q^\zeta(p)$ is continuous and monotonically increasing on $[0,\infty)$.

We now compute the population-wide thresholds. Since $v^0>0$ a.s. and $\widehat q^\zeta(p)\in[0,1]$, by \eqref{eq:qhat_MF} and the definition of $\widehat V(p)$,
\begin{equation*}
\widehat{q}^{\zeta}(p)=0\ \text{for a.e. }\zeta\quad\Longleftrightarrow\quad\widehat{V}(p)=0\quad\Longleftrightarrow\quad I^{\zeta}(p,0)\le0\ \text{for a.e. }\zeta.
\end{equation*}
Since $I^\zeta(p,0)\le0$ if and only if $p\le a$, we obtain
\begin{equation*}
p^{\mathrm{MF},\min}:=\sup\{p\ge0:\widehat{q}^{\zeta}(p)=0\text{ for a.e. }\zeta\}=\mathrm{ess\,inf}\,a.
\end{equation*}
Similarly,
\begin{equation*}
\widehat{q}^{\zeta}(p)=1\ \text{for a.e. }\zeta\quad\Longleftrightarrow\quad\widehat{V}(p)=\bar{V}\quad\Longleftrightarrow\quad I^{\zeta}(p,\bar{V})\ge1\ \text{for a.e. }\zeta.
\end{equation*}
Solving $I^\zeta(p,\bar V)\ge1$ gives $p\ge a+\frac{1}{\gamma}((v^{0})^{2}+v^{2})-\frac{1}{\gamma}\theta v^{0}\bar{V}$, a.s.. Therefore
\begin{equation*}
p^{\mathrm{MF},\max}:=\inf\{p\ge0:\widehat{q}^{\zeta}(p)=1\text{ for a.e. }\zeta\}=\mathrm{ess\,sup}\Big(a+\frac{1}{\gamma}((v^{0})^{2}+v^{2})-\frac{1}{\gamma}\theta v^{0}\bar{V}\Big).
\end{equation*}
This proves \eqref{eq:pminmax_MF}.

Since $p^{\mathrm{MF},\min}$ and $p^{\mathrm{MF},\max}$ are well-defined, and  $p\mapsto\widehat q^\zeta(p)$ is continuous and monotonically increasing for almost every type $\zeta$, the definitions of $p^{\mathrm{MF},\min}$ and $p^{\mathrm{MF},\max}$ imply $p^{\mathrm{MF},\min}<p^{\mathrm{MF},\max}$,
\begin{equation*}
\widehat{q}^{\zeta}(p)=0\quad\text{for }p\le p^{\mathrm{MF},\min}\text{ and a.e. }\zeta,\qquad\widehat{q}^{\zeta}(p)=1\quad\text{for }p\ge p^{\mathrm{MF},\max}\text{ and a.e. }\zeta.
\end{equation*}

It remains to prove the type-dependent threshold structure. For almost every type $\zeta$, define
\begin{equation*}
p^{\zeta,\min}:=\sup\{p\ge0:\widehat{q}^{\zeta}(p)=0\},\qquad p^{\zeta,\max}:=\inf\{p\ge0:\widehat{q}^{\zeta}(p)=1\}.
\end{equation*}
Since $p\mapsto\widehat q^\zeta(p)$ is continuous and $\sigma(\zeta)$-measurable for each $p$, the thresholds $p^{\zeta,\min}$ and $p^{\zeta,\max}$ are $\sigma(\zeta)$-measurable random variables. Moreover, the population-wide regimes above imply
\begin{equation*}
p^{\zeta,\min},p^{\zeta,\max}\in[p^{\mathrm{MF},\min},p^{\mathrm{MF},\max}]\qquad\text{a.s.}
\end{equation*}
By the continuity of $p\mapsto\widehat q^\zeta(p)$,
\begin{equation*}
\widehat{q}^{\zeta}(p^{\zeta,\min})=0,\qquad\widehat{q}^{\zeta}(p^{\zeta,\max})=1
\end{equation*}
for almost every type $\zeta$, hence $p^{\zeta,\min}<p^{\zeta,\max}$ a.s. The monotonicity of $p\mapsto\widehat q^\zeta(p)$ then gives, for almost every type $\zeta$,
\begin{equation*}
\widehat{q}^{\zeta}(p)=0\quad\text{for }p\le p^{\zeta,\min},\qquad0<\widehat{q}^{\zeta}(p)<1\quad\text{for }p^{\zeta,\min}<p<p^{\zeta,\max},\qquad\widehat{q}^{\zeta}(p)=1\quad\text{for }p\ge p^{\zeta,\max}.
\end{equation*}

\subsection{Proof of Proposition~\ref{prop:CSMFE_existence}}
\label{app:3.3}
By \eqref{eq:Jtilde_leader_MF} and \eqref{eq:leader_premium_objective_MF}, the reduced leader objective separates as
\begin{equation*}
\widetilde{J}^{L,\mathrm{MF}}(p,\pi^{L})=x^{L}+T\Big[\mu^{L}\pi^{L}-\frac{1}{2\gamma^{L}}\big((\sigma^{L,0})^{2}+(\sigma^{L})^{2}\big)(\pi^{L})^{2}\Big]+T\mathcal{J}^{L,\mathrm{MF}}(p).
\end{equation*}
Thus $\pi^L$ and $p$ can be optimized separately. The part depending on $\pi^L$ is a strictly concave quadratic function, and its unique maximizer is precisely \eqref{eq:pistar_mf}.

It remains to optimize $\mathcal J^{L,\mathrm{MF}}(p)$ over $p\ge0$. We first note that $\mathcal J^{L,\mathrm{MF}}$ is continuous. Indeed, if $p_n\to p$, then Proposition~\ref{prop:premium_response_CMFER} gives $\widehat q^\zeta(p_n)\to\widehat q^\zeta(p)$ for almost every type $\zeta$. Since $0\le\widehat q^\zeta\le1$ and $\mathcal Z$ is compact, the bounded convergence theorem gives
\begin{equation*}
\mathbb{E}[(1-\widehat{q}^{\zeta}(p_{n}))(p_{n}-a)]\longrightarrow\mathbb{E}[(1-\widehat{q}^{\zeta}(p))(p-a)],\qquad\mathbb{E}[(1-\widehat{q}^{\zeta}(p_{n}))v^{0}]\longrightarrow\mathbb{E}[(1-\widehat{q}^{\zeta}(p))v^{0}].
\end{equation*}
Therefore $\mathcal J^{L,\mathrm{MF}}$ is continuous.

Furthermore, by Proposition~\ref{prop:premium_response_CMFER}, $
\widehat{q}^{\zeta}(p)=0$ for $p\le p^{\mathrm{MF},\min}$ and a.e. $\zeta$, $\widehat{q}^{\zeta}(p)=1$ for $p\ge p^{\mathrm{MF},\max}$ and a.e. $\zeta$.
Consequently, on $[0,p^{\mathrm{MF},\min}]$,
\begin{equation*}
\mathcal{J}^{L,\mathrm{MF}}(p)=\mathbb{E}[p-a]-\frac{1}{2\gamma^{L}}\big(\mathbb{E}[v^{0}]\big)^{2},
\end{equation*}
so $\mathcal J^{L,\mathrm{MF}}$ is affine and strictly increasing there. On $[p^{\mathrm{MF},\max},\infty)$, $\mathcal{J}^{L,\mathrm{MF}}(p)=0$.
Hence, no premium outside $[p^{\mathrm{MF},\min},p^{\mathrm{MF},\max}]$ can improve $\mathcal J^{L,\mathrm{MF}}$, and therefore
\begin{equation*}
\sup_{p\ge0}\mathcal{J}^{L,\mathrm{MF}}(p)=\sup_{p\in[p^{\mathrm{MF},\min},\,p^{\mathrm{MF},\max}]}\mathcal{J}^{L,\mathrm{MF}}(p).
\end{equation*}
Since $\mathcal J^{L,\mathrm{MF}}$ is continuous and $[p^{\mathrm{MF},\min},p^{\mathrm{MF},\max}]$ is compact, the extreme value theorem implies that the argmax set in \eqref{eq:pstar_argmax_MF} is nonempty.

Let $p^{\mathrm{MF},*}$ be any maximizer in \eqref{eq:pstar_argmax_MF}, and define
\begin{equation*}
q^{\zeta,*}:=\widehat{q}^{\zeta}(p^{\mathrm{MF},*}),\qquad\pi^{\zeta,*}:=\widehat{\pi}^{\zeta}.
\end{equation*}
Then $(q^{\zeta,*},\pi^{\zeta,*})$ is the constant follower mean field equilibrium response induced by $p^{\mathrm{MF},*}$, while $(p^{\mathrm{MF},*},\pi^{L,\mathrm{MF},*})$ is optimal for the reinsurer. Therefore, $(p^{\mathrm{MF},*},\pi^{L,\mathrm{MF},*},q^{\zeta,*},\pi^{\zeta,*})$ is a constant Stackelberg mean field equilibrium.

\section{Proofs in Section~\ref{sec:convergence}}
\setcounter{theorem}{0}
\renewcommand{\thetheorem}{C}
The following lemma will be used repeatedly to upgrade weak convergence of the empirical type distributions to uniform convergence of type averages.
\begin{lemma}
\label{lem:uniform_empirical_integral}
Let $\mathcal{Y}$ be a compact metric space, and let $H:\mathcal{Z}\times\mathcal{Y}\to\mathbb{R}$ be continuous. Then
\begin{equation}
\label{eq:uniform_empirical_integral}
\sup_{y\in\mathcal{Y}}\Big|\int_{\mathcal{Z}}H(z,y)\,\mathfrak{m}^{N}(dz)-\int_{\mathcal{Z}}H(z,y)\,\mathfrak{m}(dz)\Big|\longrightarrow0.
\end{equation}
\end{lemma}

\begin{proof}
Since $\mathcal{Z}\times\mathcal{Y}$ is compact and $H$ is continuous, $H$ is uniformly continuous. Hence, for every $\varepsilon>0$, there exist finitely many points $y_1,\ldots,y_m\in\mathcal{Y}$ such that, for every $y\in\mathcal{Y}$, there is $k\in\{1,\ldots,m\}$ satisfying $\sup_{z\in\mathcal{Z}}|H(z,y)-H(z,y_{k})|\le\varepsilon.$ Thus, for such $y$ and $k$,
\begin{equation*}
\Big|\int_{\mathcal{Z}}H(z,y)\,\mathfrak{m}^{N}(dz)-\int_{\mathcal{Z}}H(z,y)\,\mathfrak{m}(dz)\Big|\le2\varepsilon+\Big|\int_{\mathcal{Z}}H(z,y_{k})\,\mathfrak{m}^{N}(dz)-\int_{\mathcal{Z}}H(z,y_{k})\,\mathfrak{m}(dz)\Big|.
\end{equation*}
Taking the supremum over $y\in\mathcal{Y}$ gives
\begin{equation*}
\sup_{y\in\mathcal{Y}}\Big|\int_{\mathcal{Z}}H(z,y)\,\mathfrak{m}^{N}(dz)-\int_{\mathcal{Z}}H(z,y)\,\mathfrak{m}(dz)\Big|\le2\varepsilon+\max_{1\le k\le m}\Big|\int_{\mathcal{Z}}H(z,y_{k})\,\mathfrak{m}^{N}(dz)-\int_{\mathcal{Z}}H(z,y_{k})\,\mathfrak{m}(dz)\Big|.
\end{equation*}
For each fixed $k$, the map $z\mapsto H(z,y_k)$ is continuous on $\mathcal{Z}$. Hence, by $\mathfrak m^N\Rightarrow\mathfrak m$,
\begin{equation*}
\int_{\mathcal{Z}}H(z,y_{k})\,\mathfrak{m}^{N}(dz)\longrightarrow\int_{\mathcal{Z}}H(z,y_{k})\,\mathfrak{m}(dz).
\end{equation*}
Since the maximum is taken over finitely many $k$, the maximum term converges to zero. Therefore,
\begin{equation*}
\limsup_{N\to\infty}\sup_{y\in\mathcal{Y}}\Big|\int_{\mathcal{Z}}H(z,y)\,\mathfrak{m}^{N}(dz)-\int_{\mathcal{Z}}H(z,y)\,\mathfrak{m}(dz)\Big|\le2\varepsilon.
\end{equation*}
Letting $\varepsilon\downarrow0$ proves \eqref{eq:uniform_empirical_integral}.
\end{proof}

\subsection{Proof of Proposition~\ref{prop:qhat_uniform_convergence}}
\label{app:4.1}
We first prove 
\begin{equation}
\label{eq:FN_to_F_uniform}
\sup_{(p,V)\in[\underline{p},\bar{p}]\times[0,\bar{v^{0}}]}\big|F^{N}(p,V)-F(p,V)\big|\longrightarrow0.
\end{equation}
For $(p,V)\in[\underline{p},\bar{p}]\times[0,\bar{v^{0}}]$,
\begin{align}
\label{eq:FN_F_decomposition}
\big|F^{N}(p,V)-F(p,V)\big| \le &\ \int_{\mathcal{Z}}v_{z}^{0}\big|\Pi_{[0,1]}\big(I^{N}(z;p,V)\big)-\Pi_{[0,1]}\big(I(z;p,V)\big)\big|\,\mathfrak{m}^{N}(dz)
\nonumber\\
& + \Big|\int_{\mathcal{Z}}v_{z}^{0}\,\Pi_{[0,1]}\big(I(z;p,V)\big)\,(\mathfrak{m}^{N}-\mathfrak{m})(dz)\Big|.
\end{align}
Since $\Pi_{[0,1]}$ is $1$-Lipschitz,
\begin{equation*}
\int_{\mathcal{Z}}v_{z}^{0}\big|\Pi_{[0,1]}\big(I^{N}(z;p,V)\big)-\Pi_{[0,1]}\big(I(z;p,V)\big)\big|\,\mathfrak{m}^{N}(dz)\ \le\ \bar{v^{0}}\sup_{(z,p,V)\in\mathcal{Z}\times[\underline{p},\bar{p}]\times[0,\bar{v^{0}}]}\big|I^{N}(z;p,V)-I(z;p,V)\big|.
\end{equation*}
By the bounds in Assumption~\ref{assump}, the common numerator $(p-a_{z})+\frac{1}{\gamma_{z}}\theta_{z}v_{z}^{0}V$ of $I^{N}(z;p,V)$ and $I(z;p,V)$ is uniformly bounded on $\mathcal{Z}\times[\underline p,\bar p]\times[0,\bar v^{0}]$, and their denominators are uniformly bounded away from zero. Since $\theta_z\in[0,1]$, we have $\sup_{z\in\mathcal Z}|\theta_z/N|\le 1/N$, and hence $1-\theta_z/N\to1$ uniformly on $\mathcal Z$. 
Since,
\begin{equation*}
I^{N}(z;p,V)-I(z;p,V)=\big((p-a_{z})+\frac{1}{\gamma_{z}}\theta_{z}v_{z}^{0}V\big)\Big(\frac{1}{\frac{1}{\gamma_{z}}\big((v_{z}^{0})^{2}+(v_{z})^{2}(1-\theta_{z}/N)\big)}-\frac{1}{\frac{1}{\gamma_{z}}\big((v_{z}^{0})^{2}+(v_{z})^{2}\big)}\Big),
\end{equation*}
we have
\begin{equation}
\label{eq:IN_to_I_uniform}
\sup_{(z,p,V)\in\mathcal{Z}\times[\underline{p},\bar{p}]\times[0,\bar{v^{0}}]}\big|I^{N}(z;p,V)-I(z;p,V)\big|\longrightarrow0.
\end{equation}
Thus, the first term on the right-hand side of \eqref{eq:FN_F_decomposition} converges to zero uniformly in $(p,V)$.

For the second term on the right-hand side of \eqref{eq:FN_F_decomposition}, apply Lemma~\ref{lem:uniform_empirical_integral} with $\mathcal{Y}=[\underline{p},\bar{p}]\times[0,\bar{v^{0}}]$ and $H(z,(p,V))=v_{z}^{0}\,\Pi_{[0,1]}\big(I(z;p,V)\big)$. The map $H$ is continuous on $\mathcal{Z}\times[\underline p,\bar p]\times[0,\bar{v^{0}}]$. Hence
\begin{equation*}
\sup_{(p,V)\in[\underline{p},\bar{p}]\times[0,\bar{v^{0}}]}\Big|\int_{\mathcal{Z}}v_{z}^{0}\Pi_{[0,1]}\big(I(z;p,V)\big)\,(\mathfrak{m}^{N}-\mathfrak{m})(dz)\Big|\longrightarrow0.
\end{equation*}
This proves \eqref{eq:FN_to_F_uniform}.

We next prove the uniform contraction estimate used below. Since $\Pi_{[0,1]}$ is $1$-Lipschitz,
\begin{align*}
\big|F^{N}(p,V)-F^{N}(p,V')\big| & \le \int_{\mathcal{Z}}v_{z}^{0}\big|I^{N}(z;p,V)-I^{N}(z;p,V')\big|\,\mathfrak{m}^{N}(dz)
\\
& = \int_{\mathcal{Z}}\frac{\theta_{z}(v_{z}^{0})^{2}}{(v_{z}^{0})^{2}+(v_{z})^{2}(1-\theta_{z}/N)}\,\mathfrak{m}^{N}(dz)\,|V-V'|.
\end{align*}
Since $N\ge2$ and $\theta_z\in[0,1]$, we have $1-\theta_{z}/N\ge1/2$. Hence, by the bounds in Assumption~\ref{assump}, for every $z\in\mathcal Z$,
\begin{equation*}
\frac{\theta_{z}(v_{z}^{0})^{2}}{(v_{z}^{0})^{2}+(v_{z})^{2}(1-\theta_{z}/N)}\le\frac{(v_{z}^{0})^{2}}{(v_{z}^{0})^{2}+\frac{1}{2}(\underline{v})^{2}}\le\frac{(\bar{v^{0}})^{2}}{(\bar{v^{0}})^{2}+\frac{1}{2}(\underline{v})^{2}}=:\rho_{v}<1.
\end{equation*}
Therefore,
\begin{equation}
\label{eq:FN_uniform_contraction}
\big|F^{N}(p,V)-F^{N}(p,V')\big|\le\rho_{v}\,|V-V'|
\end{equation}
holds for every $N\ge2$.

We now prove \eqref{eq:VhatN_to_Vhat_uniform}. Since $0\le \Pi_{[0,1]}\le1$ and $0<v_z^0\le \bar{v^{0}}$ on $\mathcal{Z}$, both $\widehat V^N(p)$ and $\widehat V(p)$ satisfy
\begin{equation*}
0\le\widehat{V}^{N}(p)\le\bar{v^{0}},\qquad0\le\widehat{V}(p)\le\bar{v^{0}},\qquad p\in[\underline{p},\bar{p}].
\end{equation*}
Fix $p\in[\underline{p},\bar{p}]$. By the fixed-point identities,
\begin{align*}
\big|\widehat{V}^{N}(p)-\widehat{V}(p)\big| & = \big|F^{N}(p,\widehat{V}^{N}(p))-F(p,\widehat{V}(p))\big|
\\
& \le \big|F^{N}(p,\widehat{V}^{N}(p))-F^{N}(p,\widehat{V}(p))\big|+\big|F^{N}(p,\widehat{V}(p))-F(p,\widehat{V}(p))\big|.
\end{align*}
By \eqref{eq:FN_uniform_contraction}, $\big|F^{N}(p,\widehat{V}^{N}(p))-F^{N}(p,\widehat{V}(p))\big|\le\rho_{v}\big|\widehat{V}^{N}(p)-\widehat{V}(p)\big|,$ and hence $\big|\widehat{V}^{N}(p)-\widehat{V}(p)\big|\le\frac{1}{1-\rho_{v}}\big|F^{N}(p,\widehat{V}(p))-F(p,\widehat{V}(p))\big|.$ Taking the supremum over $p\in[\underline{p},\bar{p}]$ gives
\begin{align*}
\sup_{p\in[\underline{p},\bar{p}]}\big|\widehat{V}^{N}(p)-\widehat{V}(p)\big| & \le \frac{1}{1-\rho_{v}}\sup_{p\in[\underline{p},\bar{p}]}\big|F^{N}(p,\widehat{V}(p))-F(p,\widehat{V}(p))\big|
\\
& \le \frac{1}{1-\rho_{v}}\sup_{(p,V)\in[\underline{p},\bar{p}]\times[0,\bar{v^{0}}]}\big|F^{N}(p,V)-F(p,V)\big|.
\end{align*}
The right-hand side converges to zero by \eqref{eq:FN_to_F_uniform}. This proves \eqref{eq:VhatN_to_Vhat_uniform}.

Finally, we prove \eqref{eq:qhat_uniform_convergence}. Since $\Pi_{[0,1]}$ is $1$-Lipschitz,
\begin{align*}
\big|\widehat{q}^{N}(z;p)-\widehat{q}(z;p)\big| & \le \big|I^{N}(z;p,\widehat{V}^{N}(p))-I(z;p,\widehat{V}(p))\big|
\\
& \le \big|I^{N}(z;p,\widehat{V}^{N}(p))-I(z;p,\widehat{V}^{N}(p))\big|+\big|I(z;p,\widehat{V}^{N}(p))-I(z;p,\widehat{V}(p))\big|.
\end{align*}
The first term converges to zero uniformly on $\mathcal{Z}\times[\underline{p},\bar{p}]$ by \eqref{eq:IN_to_I_uniform}. For the second term, the bounds in Assumption~\ref{assump} give a constant $C>0$ such that
\begin{equation*}
\big|I(z;p,V)-I(z;p,V')\big|\le C|V-V'|
\end{equation*}
for all $z\in\mathcal{Z}$, $p\in[\underline{p},\bar{p}]$, and $V,V'\in[0,\bar{v^{0}}]$. Hence \eqref{eq:VhatN_to_Vhat_uniform} implies \eqref{eq:qhat_uniform_convergence}.

\subsection{Proof of Proposition~\ref{prop:pihat_uniform_convergence}}
\label{app:4.2}
Since $N\ge2$ and $\theta_z\in[0,1]$, we have $1-\theta_{z}/N\ge1/2$. Hence, by the bounds in Assumption~\ref{assump}, for every $z\in\mathcal Z$,
\begin{equation*}
\frac{\theta_{z}(\sigma_{z}^{0})^{2}}{(\sigma_{z}^{0})^{2}+(\sigma_{z})^{2}(1-\theta_{z}/N)}\le\frac{(\sigma_{z}^{0})^{2}}{(\sigma_{z}^{0})^{2}+\frac{1}{2}(\underline{\sigma})^{2}}\le\frac{(\bar{\sigma^{0}})^{2}}{(\bar{\sigma^{0}})^{2}+\frac{1}{2}(\underline{\sigma})^{2}}=:\rho_{\sigma}<1. 
\end{equation*}
Hence, the denominator of $\widehat\Sigma^N$ is uniformly bounded away from zero. 

The numerator and denominator integrands of $\widehat\Sigma^N$ converge uniformly on $\mathcal Z$ to those of $\widehat\Sigma$, because $\sup_{z\in\mathcal{Z}}|\theta_{z}/N|\le1/N\to0$ and Assumption~\ref{assump} keeps denominators uniformly away from zero. Together with $\mathfrak m^N\Rightarrow\mathfrak m$, this implies that the numerator and denominator of $\widehat\Sigma^N$ converge to those of $\widehat\Sigma$. Therefore, \eqref{eq:SigmahatN_to_Sigmahat} holds. 

Finally, the uniform convergence \eqref{eq:pihat_uniform_convergence} follows from \eqref{eq:pihat_type_form}, \eqref{eq:SigmahatN_to_Sigmahat}, the bound $\sup_{z\in\mathcal Z}|\theta_z/N|\le1/N$, and the denominator bounds implied by Assumption~\ref{assump}.

\subsection{Proof of Proposition~\ref{prop:profile_distribution_convergence}}
\label{app:4.3}
We first prove the uniform convergence of the retention responses along $(p^N)_{N\ge1}$. By \eqref{eq:qhat_uniform_convergence},
\begin{equation*}
\sup_{z\in\mathcal{Z}}\big|\widehat{q}^{N}(z;p^{N})-\widehat{q}(z;p^{N})\big|\longrightarrow0.
\end{equation*}
As shown at the beginning of \ref{app:3.2}, the map $p\mapsto\widehat V(p)$ is continuous. Together with the explicit formula for $\widehat q(z;p)$ and the continuity of $I$, this implies that $(z,p)\mapsto\widehat q(z;p)$ is continuous on the compact set $\mathcal Z\times[\underline p,\bar p]$, and hence uniformly continuous. Since $p^N\to p^\infty$, $\sup_{z\in\mathcal{Z}}\big|\widehat{q}(z;p^{N})-\widehat{q}(z;p^{\infty})\big|\longrightarrow0.$ Therefore,
\begin{equation}
\label{eq:qhat_uniform_convergent}
\sup_{z\in\mathcal{Z}}\big|\widehat{q}^{N}(z;p^{N})-\widehat{q}(z;p^{\infty})\big|\le\sup_{z\in\mathcal{Z}}\big|\widehat{q}^{N}(z;p^{N})-\widehat{q}(z;p^{N})\big|+\sup_{z\in\mathcal{Z}}\big|\widehat{q}(z;p^{N})-\widehat{q}(z;p^{\infty})\big|\longrightarrow0.
\end{equation}

Now we prove the weak convergence \eqref{eq:profile_distribution_convergence}. Let $\Phi \in C_b(\mathcal Z \times [0,1] \times \mathbb R)$ be an arbitrary bounded continuous test function. By \eqref{eq:pihat_type_form}, \eqref{eq:SigmahatN_to_Sigmahat}, and the bounds in Assumption~\ref{assump}, there exists $M>0$ such that
\begin{equation*}
\sup_{z\in\mathcal{Z}}|\widehat{\pi}^{N}(z)|\le M
\quad\text{for all sufficiently large }N,\qquad
\sup_{z\in\mathcal{Z}}|\widehat{\pi}(z)|\le M.
\end{equation*}
Hence, for all sufficiently large $N$, the empirical distributions and the limiting law in \eqref{eq:profile_distribution_convergence} are supported on the compact set $\mathcal Z\times[0,1]\times[-M,M]$. Thus, their integrals against $\Phi$ depend only on the restriction $\Psi:=\Phi|_{\mathcal Z\times[0,1]\times[-M,M]}$. Since $\mathcal Z\times[0,1]\times[-M,M]$ is compact, $\Psi$ is uniformly continuous. Using \eqref{eq:qhat_uniform_convergent} and \eqref{eq:pihat_uniform_convergence}, we obtain
\begin{equation}
\label{eq:test_function_convergence}
\sup_{z\in\mathcal{Z}}\big|\Psi\big(z,\widehat{q}^{N}(z;p^{N}),\widehat{\pi}^{N}(z)\big)-\Psi\big(z,\widehat{q}(z;p^{\infty}),\widehat{\pi}(z)\big)\big|\longrightarrow0.
\end{equation}
Therefore,
\begin{align*}
& \Big|\int\Psi\,d\Big(\frac{1}{N}\sum_{i=1}^{N}\delta_{(\zeta^{i,N},\widehat{q}^{N}(\zeta^{i,N};p^{N}),\widehat{\pi}^{N}(\zeta^{i,N}))}\Big)-\int\Psi\,d\mathcal{L}\big(\zeta,\widehat{q}(\zeta;p^{\infty}),\widehat{\pi}(\zeta)\big)\Big|
\\
& \le \int_{\mathcal{Z}}\Big|\Psi\big(z,\widehat{q}^{N}(z;p^{N}),\widehat{\pi}^{N}(z)\big)-\Psi\big(z,\widehat{q}(z;p^{\infty}),\widehat{\pi}(z)\big)\Big|\,\mathfrak{m}^{N}(dz)
\\
& \quad + \Big|\int_{\mathcal{Z}}\Psi\big(z,\widehat{q}(z;p^{\infty}),\widehat{\pi}(z)\big)\,(\mathfrak{m}^{N}-\mathfrak{m})(dz)\Big|.
\end{align*}
The first term converges to zero by \eqref{eq:test_function_convergence}. The second term converges to zero by \eqref{eq:empirical_type_convergence}, because the map $z\mapsto\Psi\big(z,\widehat q(z;p^{\infty}),\widehat\pi(z)\big)$ is continuous on $\mathcal Z$. Since this holds for every $\Phi\in C_b(\mathcal Z\times[0,1]\times\mathbb R)$, this proves \eqref{eq:profile_distribution_convergence}.

\subsection{Proof of Proposition~\ref{prop:J_uniform_convergence}}
\label{app:4.4}
For any continuous function $b:\mathcal Z\to\mathbb R$, we first prove 
\begin{equation}
\label{eq:ceded_integral_convergence}
\sup_{p\in[\underline{p},\bar{p}]}\Big|\int_{\mathcal{Z}}b(z)(1-\widehat{q}^{N}(z;p))\,\mathfrak{m}^{N}(dz)-\int_{\mathcal{Z}}b(z)(1-\widehat{q}(z;p))\,\mathfrak{m}(dz)\Big|\longrightarrow0.
\end{equation}
The difference in \eqref{eq:ceded_integral_convergence} is bounded by
\begin{equation*}
\int_{\mathcal{Z}}|b(z)|\,|\widehat{q}^{N}(z;p)-\widehat{q}(z;p)|\,\mathfrak{m}^{N}(dz)+\Big|\int_{\mathcal{Z}}b(z)(1-\widehat{q}(z;p))\,(\mathfrak{m}^{N}-\mathfrak{m})(dz)\Big|.
\end{equation*}
The first term converges to zero uniformly in $p$ by \eqref{eq:qhat_uniform_convergence} and the boundedness of $b$ on $\mathcal Z$ (since $b$ is continuous on the compact set $\mathcal Z$). The second term converges to zero uniformly in $p$ by Lemma~\ref{lem:uniform_empirical_integral}, applied to $\mathcal{Y}=[\underline{p},\bar{p}]$ and $H(z,p)=b(z)(1-\widehat q(z;p))$, which is continuous on $\mathcal Z\times[\underline p,\bar p]$. This proves \eqref{eq:ceded_integral_convergence}.

Taking $b(z)=1$ and $b(z)=a_z$ in \eqref{eq:ceded_integral_convergence}, and using that $p$ is bounded on $[\underline p,\bar p]$, gives
\begin{equation}
\label{eq:ceded_surplus_convergence}
\sup_{p\in[\underline{p},\bar{p}]}\Big|\int_{\mathcal{Z}}(1-\widehat{q}^{N}(z;p))(p-a_{z})\,\mathfrak{m}^{N}(dz)-\int_{\mathcal{Z}}(1-\widehat{q}(z;p))(p-a_{z})\,\mathfrak{m}(dz)\Big|\longrightarrow0.
\end{equation}
Taking $b(z)=v_z^0$ in \eqref{eq:ceded_integral_convergence} gives
\begin{equation*}
\sup_{p\in[\underline{p},\bar{p}]}\Big|\int_{\mathcal{Z}}(1-\widehat{q}^{N}(z;p))v_{z}^{0}\,\mathfrak{m}^{N}(dz)-\int_{\mathcal{Z}}(1-\widehat{q}(z;p))v_{z}^{0}\,\mathfrak{m}(dz)\Big|\longrightarrow0.
\end{equation*}
This implies
\begin{align}
\label{eq:ceded_exposure_convergence}
& \sup_{p\in[\underline{p},\bar{p}]}\Big|\Big(\int_{\mathcal{Z}}(1-\widehat{q}^{N}(z;p))v_{z}^{0}\,\mathfrak{m}^{N}(dz)\Big)^{2}-\Big(\int_{\mathcal{Z}}(1-\widehat{q}(z;p))v_{z}^{0}\,\mathfrak{m}(dz)\Big)^{2}\Big|
\nonumber\\
& \le 2\bar{v^{0}}\sup_{p\in[\underline{p},\bar{p}]}\Big|\int_{\mathcal{Z}}(1-\widehat{q}^{N}(z;p))v_{z}^{0}\,\mathfrak{m}^{N}(dz)-\int_{\mathcal{Z}}(1-\widehat{q}(z;p))v_{z}^{0}\,\mathfrak{m}(dz)\Big|\longrightarrow0.
\end{align}

Furthermore, since $0\le1-\widehat q^N(z;p)\le1$ and $v_z\le\bar v$ on $\mathcal Z$,
\begin{equation*}
\sup_{p\in[\underline{p},\bar{p}]}\frac{1}{N}\int_{\mathcal{Z}}(1-\widehat{q}^{N}(z;p))^{2}v_{z}^{2}\,\mathfrak{m}^{N}(dz)\le\frac{\bar{v}^{2}}{N}\longrightarrow0.
\end{equation*}
Combining this with \eqref{eq:ceded_surplus_convergence} and \eqref{eq:ceded_exposure_convergence} proves \eqref{eq:J_uniform_convergence}.

\subsection{Proof of Proposition~\ref{prop:optimal_premium_convergence}}
\label{app:4.5}
We first prove the subsequential characterization. Let $(N_k)_{k\ge1}$ be a strictly increasing sequence, choose $p^{N_k,*}\in\mathcal P^{N_k,*}$, and suppose that $p^{N_k,*}\to p^{\infty,*}$. For any $p\in[\underline p,\bar p]$, the optimality of $p^{N_k,*}$ gives
\begin{equation}
\label{eq:finite_premium_optimality}
\mathcal{J}^{L,N_{k}}(p^{N_{k},*})\ge\mathcal{J}^{L,N_{k}}(p),
\end{equation}
and the result \eqref{eq:J_uniform_convergence} gives, $\mathcal{J}^{L,N_{k}}(p)\to\mathcal{J}^{L,\mathrm{MF}}(p).$ Moreover, by the triangle inequality, 
\begin{equation*}
\big|\mathcal{J}^{L,N_{k}}(p^{N_{k},*})-\mathcal{J}^{L,\mathrm{MF}}(p^{\infty,*})\big|\le\big|\mathcal{J}^{L,N_{k}}(p^{N_{k},*})-\mathcal{J}^{L,\mathrm{MF}}(p^{N_{k},*})\big|+\big|\mathcal{J}^{L,\mathrm{MF}}(p^{N_{k},*})-\mathcal{J}^{L,\mathrm{MF}}(p^{\infty,*})\big|,
\end{equation*}
in which the first term converges to zero by \eqref{eq:J_uniform_convergence}, and the second term converges to zero by the continuity of $\mathcal{J}^{L,\mathrm{MF}}$ and $p^{N_k,*}\to p^{\infty,*}$. Hence
\begin{equation*}
\mathcal{J}^{L,N_{k}}(p^{N_{k},*})\to\mathcal{J}^{L,\mathrm{MF}}(p^{\infty,*}).
\end{equation*}
Passing to the limit in \eqref{eq:finite_premium_optimality} yields
\begin{equation*}
\mathcal{J}^{L,\mathrm{MF}}(p^{\infty,*})\ge\mathcal{J}^{L,\mathrm{MF}}(p).
\end{equation*}
Since $p\in[\underline{p},\bar{p}]$ was arbitrary, $p^{\infty,*}$ maximizes $\mathcal{J}^{L,\mathrm{MF}}$ over $[\underline{p},\bar{p}]$. Thus, $p^{\infty,*}\in\mathcal{P}^{\mathrm{MF},*}$.

It remains to prove \eqref{eq:optimal_premium_convergence}. Suppose, to the contrary, that \eqref{eq:optimal_premium_convergence} fails. Then there exist $\varepsilon>0$, a strictly increasing sequence $(N_k)_{k\ge1}$, and points $p^{N_k,*}\in\mathcal P^{N_k,*}$ such that
\begin{equation}
\label{eq:distance_contradiction}
\mathrm{dist}\big(p^{N_{k},*},\mathcal{P}^{\mathrm{MF},*}\big)\ge\varepsilon,\qquad k\ge1.
\end{equation}
Since $p^{N_k,*}\in[\underline{p},\bar{p}]$ and $[\underline{p},\bar{p}]$ is compact, there exists a further subsequence $(N_{k_\ell})_{\ell\ge1}$ such that $p^{N_{k_\ell},*}\to p^{\infty,*}$ for some $p^{\infty,*}\in[\underline{p},\bar{p}]$. By the subsequential characterization proved above, $p^{\infty,*}\in\mathcal{P}^{\mathrm{MF},*}$. Therefore,
\begin{equation*}
\mathrm{dist}\big(p^{N_{k_{\ell}},*},\mathcal{P}^{\mathrm{MF},*}\big)\le|p^{N_{k_{\ell}},*}-p^{\infty,*}|\to0,
\end{equation*}
which contradicts \eqref{eq:distance_contradiction}. Thus, \eqref{eq:optimal_premium_convergence} holds.

\subsection{Proof of Theorem~\ref{thm:empirical_profile_convergence}}
\label{app:4.6}
Suppose that \eqref{eq:empirical_profile_assumption} holds along a subsequence $(N_k)_{k\ge1}$. Since $p^{N_k,*}\in[\underline p,\bar p]$ and $[\underline p,\bar p]$ is compact, there exist a further subsequence $(N_{k_\ell})_{\ell\ge1}$ and some $p^{\infty,*}\in[\underline p,\bar p]$ such that $p^{N_{k_\ell},*}\to p^{\infty,*}$. By Proposition~\ref{prop:optimal_premium_convergence}, $p^{\infty,*}\in\mathcal P^{\mathrm{MF},*}$. Moreover, \eqref{eq:empirical_type_convergence} gives $\mathfrak{m}^{N_{k_\ell}}\Rightarrow\mathfrak{m}$. Therefore, applying Proposition~\ref{prop:profile_distribution_convergence} along this further subsequence gives 
\begin{equation*}
\frac{1}{N_{k_{\ell}}}\sum_{i=1}^{N_{k_{\ell}}}\delta_{(\zeta^{i,N_{k_{\ell}}},q^{i,N_{k_{\ell}},*},\pi^{i,N_{k_{\ell}},*})}\Rightarrow\mathcal{L}\big(\zeta,\widehat{q}(\zeta;p^{\infty,*}),\widehat{\pi}(\zeta)\big).
\end{equation*}
On the other hand, the original subsequence in \eqref{eq:empirical_profile_assumption} converges weakly to $\Lambda$, so every further subsequence has the same weak limit. Therefore, $\Lambda=\mathcal{L}\big(\zeta,\widehat{q}(\zeta;p^{\infty,*}),\widehat{\pi}(\zeta)\big),$ which proves \eqref{eq:empirical_weak_limit}.

\section{Explicit mean field equilibrium under $\theta$-only heterogeneity}
\label{app:mf_closed_form}
We explicitly characterize the mean field follower retention response $\widehat q^\zeta(p)$ and the reduced premium objective $\mathcal J^{L,\mathrm{MF}}(p)$ under the specification of Section~\ref{subsec:num_mfg}, where heterogeneity enters only through $\theta$, whose density is $f_\theta^h$ in \eqref{eq:theta_density_h}, while $a,\gamma,v^0,v$ are constant. This characterization is used to compute the mean field reduced premium objective, optimal premiums, and retention distributions in Figures~\ref{fig:4}--\ref{fig:5}. Since $\widehat q^\zeta(p)$ and $\mathcal J^{L,\mathrm{MF}}(p)$ are determined by $\widehat V(p)$, it suffices to characterize $\widehat V(p)$ by solving the fixed-point equation \eqref{eq:fixed_point_V_MF}.

By Proposition~\ref{prop:premium_response_CMFER}, $p^{\mathrm{MF},\min}=a$ and $p^{\mathrm{MF},\max}=a+\frac1\gamma((v^0)^2+v^2)$. Within $[p^{\mathrm{MF},\min},p^{\mathrm{MF},\max}]$, we first consider the premium range in which all types retain partially. More precisely, there exists a value $p^{\mathrm{MF},c}$ such that, for $p\in(p^{\mathrm{MF},\min},p^{\mathrm{MF},c})$, $I^\zeta(p,\widehat V(p))\in(0,1)$ for almost every $\theta$. In this range, solving \eqref{eq:fixed_point_V_MF}, together with $\mathbb E[\theta]=0.5$, gives
\begin{equation*}
\widehat{V}(p)=\widehat{V}_{1}(p):=\frac{v^{0}(p-a)}{\frac{1}{\gamma}\big(0.5\,(v^{0})^{2}+v^{2}\big)}.
\end{equation*}
Under this expression, $I^\zeta(p,\widehat V_1(p))$ is increasing in both $p$ and $\theta$. Hence the endpoint $p^{\mathrm{MF},c}$ is determined by the type with $\theta=1$ first reaching full retention, or equivalently,
$I^\zeta(p^{\mathrm{MF},c},\widehat V_1(p^{\mathrm{MF},c}))\big|_{\theta=1}=1$. Solving this equation gives
\begin{equation}
\label{eq:pMFc_def}
p^{\mathrm{MF},c}=a+\frac{1}{\gamma}\big((v^{0})^{2}+v^{2}\big)\frac{0.5\,(v^{0})^{2}+v^{2}}{1.5\,(v^{0})^{2}+v^{2}}.
\end{equation}
Thus, $\widehat V(p)=\widehat V_1(p)$ for $p^{\mathrm{MF},\min}<p<p^{\mathrm{MF},c}$. For the parameters $a=1$, $\gamma=1$, $v^0=0.7$, and $v=0.3$ used in Figures~\ref{fig:4}--\ref{fig:5}, this gives $p^{\mathrm{MF},c}\approx1.236$.

We next consider the remaining range $p^{\mathrm{MF},c}\le p\le p^{\mathrm{MF},\max}$, and write the solution of \eqref{eq:fixed_point_V_MF} on this range as $\widehat V_2(p)$. Since $I^\zeta(p,V)$ is increasing in $\theta$ when $V>0$, the projection in \eqref{eq:fixed_point_V_MF} reaches its upper bound for sufficiently large values of $\theta$. Hence, the partially retaining and fully retaining types are separated by a cutoff. Let $\theta^c(p)$ denote this cutoff, defined by $I^\zeta(p,\widehat V_2(p))\big|_{\theta=\theta^c(p)}=1$. Solving this equation gives
\begin{equation}
\label{eq:theta_cutoff}
\theta^{c}(p)=\frac{(v^{0})^{2}+v^{2}-\gamma(p-a)}{v^{0}\widehat{V}_{2}(p)}.
\end{equation}
Using \eqref{eq:theta_cutoff}, the incentive can be rewritten as $I^\zeta(p,\widehat V_2(p))=1-\frac{v^0\widehat V_2(p)}{(v^0)^2+v^2}(\theta^c(p)-\theta)$. Hence the projection term in \eqref{eq:fixed_point_V_MF} is
\begin{equation*}
\Pi_{[0,1]}\big(I^{\zeta}(p,\widehat{V}_{2}(p))\big)=1-\frac{v^{0}\widehat{V}_{2}(p)}{(v^{0})^{2}+v^{2}}\big(\theta^{c}(p)-\theta\big)_{+},\qquad(y)_{+}:=\max\{y,0\}.
\end{equation*}
Taking the expectation in \eqref{eq:fixed_point_V_MF} then involves the term $\mathbb E[(\theta^c(p)-\theta)_+]$. For $x\in[0,1]$, define
\begin{equation*}
G_{h}(x):=\mathbb{E}[(x-\theta)_{+}]=\int_{0}^{x}(x-\theta)f_{\theta}^{h}(\theta)\,d\theta=\frac{h}{2}x^{2}+\frac{5}{6}(1-h)[(x-0.2)_{+}]^{2}-\frac{5}{6}(1-h)[(x-0.8)_{+}]^{2}.
\end{equation*}
Substituting the projection term into \eqref{eq:fixed_point_V_MF} and rearranging gives
\begin{equation}
\label{eq:V2}
\widehat{V}_{2}(p)=\frac{v^{0}\big((v^{0})^{2}+v^{2}\big)}{(v^{0})^{2}+v^{2}+(v^{0})^{2}G_{h}(\theta^{c}(p))}.
\end{equation}
Equations \eqref{eq:theta_cutoff} and \eqref{eq:V2} characterize $\widehat V_2(p)$. Since $G_h$ is quadratic on $[0,0.2]$, $[0.2,0.8]$, and $[0.8,1]$, $\widehat V_2(p)$ is obtained piecewise by solving the resulting quadratic equations.

Therefore, on $[p^{\mathrm{MF},\min},p^{\mathrm{MF},\max}]$,
\begin{equation*}
\widehat{V}(p)=\begin{cases}
\widehat{V}_{1}(p), & p^{\mathrm{MF},\min}<p<p^{\mathrm{MF},c},\\
\widehat{V}_{2}(p), & p^{\mathrm{MF},c}\le p\le p^{\mathrm{MF},\max}.
\end{cases}
\end{equation*}
This characterization of $\widehat V(p)$ determines the follower retention response $\widehat q^\zeta(p)=\Pi_{[0,1]}(I^\zeta(p,\widehat V(p)))$. Since $a$ and $v^0$ are constant in the present specification, $\widehat V(p)=\mathbb E[\widehat q^\zeta(p)v^0]$ implies $\mathbb E[1-\widehat q^\zeta(p)]=1-\widehat V(p)/v^0$, hence \eqref{eq:leader_premium_objective_MF} gives
\begin{equation*}
\mathcal{J}^{L,\mathrm{MF}}(p)=(p-a)\big(1-\widehat{V}(p)/v^{0}\big)-\frac{1}{2\gamma^{L}}\big(v^{0}-\widehat{V}(p)\big)^{2}.
\end{equation*}
Thus, $\widehat q^\zeta(p)$ and $\mathcal J^{L,\mathrm{MF}}(p)$ are explicitly determined once $\widehat V(p)$ is known. The reduced premium objective is piecewise algebraic, so its maximizers are obtained by comparing the maximizers on each piece.

For the uniform case used in Figure~\ref{fig:4}, $h=1$, so $f_\theta^h$ is uniform on $[0,1]$ and $G_h(x)=G_1(x)=x^2/2$. In this case, \eqref{eq:theta_cutoff} and \eqref{eq:V2} yield
\begin{equation*}
\widehat{V}_{2}(p)=\frac{1}{2}v^{0}+\frac{1}{2}\sqrt{(v^{0})^{2}-\frac{2\big((v^{0})^{2}+v^{2}-\gamma(p-a)\big)^{2}}{(v^{0})^{2}+v^{2}}}.
\end{equation*}
The resulting reduced premium objective is single-peaked. For the polarized case in Figure~\ref{fig:5}, with $h=2.3$ and $\gamma^L\approx1.254893$, the piecewise form of $G_h$ leads to a piecewise algebraic reduced premium objective with two maximizers.

\bibliographystyle{plain}
\bibliography{refs}

@article{LasryLions2006a,
  author  = {Jean-Michel Lasry and Pierre-Louis Lions},
  title   = {Jeux {\`a} champ moyen. {I} -- {L}e cas stationnaire},
  journal = {Comptes Rendus Math{\'e}matique},
  volume  = {343},
  number  = {9},
  pages   = {619--625},
  year    = {2006},
  doi     = {10.1016/j.crma.2006.09.019},
}

@article{LasryLions2006b,
  author  = {Jean-Michel Lasry and Pierre-Louis Lions},
  title   = {Jeux {\`a} champ moyen. {II} -- {H}orizon fini et contr{\^o}le optimal},
  journal = {Comptes Rendus Math{\'e}matique},
  volume  = {343},
  number  = {10},
  pages   = {679--684},
  year    = {2006},
  doi     = {10.1016/j.crma.2006.09.018},
}

@article{LasryLions2007,
  author  = {Lasry, Jean-Michel and Lions, Pierre-Louis},
  title   = {Mean field games},
  journal = {Japanese Journal of Mathematics},
  year    = {2007},
  volume  = {2},
  number  = {1},
  pages   = {229--260},
  doi     = {10.1007/s11537-007-0657-8}
}

@article{HuangMalhameCaines2006,
  author  = {Huang, Minyi and Malham{\'e}, Roland P. and Caines, Peter E.},
  title   = {Large population stochastic dynamic games: Closed-loop {M}c{K}ean--{V}lasov systems and the {N}ash certainty equivalence principle},
  journal = {Communications in Information and Systems},
  year    = {2006},
  volume  = {6},
  number  = {3},
  pages   = {221--252},
  doi     = {10.4310/CIS.2006.v6.n3.a5}
}

@article{HuangCainesMalhame2007,
  author  = {Minyi Huang and Peter E. Caines and Roland P. Malham{\'e}},
  title   = {Large-Population Cost-Coupled {LQG} Problems with Nonuniform Agents: Individual-Mass Behavior and Decentralized $\varepsilon$-Nash Equilibria},
  journal = {IEEE Transactions on Automatic Control},
  volume  = {52},
  number  = {9},
  pages   = {1560--1571},
  year    = {2007},
  doi     = {10.1109/TAC.2007.904450},
}

@misc{ZhangEtAl2024ReinsuranceMFG,
  author = {Zhang, Shuhua and Qian, Shenghua and Wang, Xinyu and Kang, Xinyi},
  title  = {A {S}tackelberg Mean Field Game Approach to Optimal Reinsurance and Investment: One Reinsurer and Competitive Insurers},
  year   = {2024},
  note   = {{S}SRN Working Paper},
  doi    = {10.2139/ssrn.4944970}
}

@article{ChenShen2019,
  author  = {Chen, Lv and Shen, Yang},
  title   = {Stochastic {S}tackelberg differential reinsurance games under time-inconsistent mean-variance framework},
  journal = {Insurance: Mathematics and Economics},
  year    = {2019},
  volume  = {88},
  pages   = {120--137},
  doi     = {10.1016/j.insmatheco.2019.06.006}
}

@article{LiYoung2022,
  author  = {Li, Danping and Young, Virginia R.},
  title   = {Stackelberg differential game for reinsurance: Mean-variance framework and random horizon},
  journal = {Insurance: Mathematics and Economics},
  year    = {2022},
  volume  = {102},
  pages   = {42--55},
  doi     = {10.1016/j.insmatheco.2021.11.006}
}

@article{BoWangZhou2024,
  author  = {Bo, Lijun and Wang, Shihua and Zhou, Chao},
  title   = {A mean field game approach to optimal investment and risk control for competitive insurers},
  journal = {Insurance: Mathematics and Economics},
  year    = {2024},
  volume  = {116},
  pages   = {202--217},
  doi     = {10.1016/j.insmatheco.2024.03.002}
}

@article{GuanHu2022,
  author  = {Guan, Guohui and Hu, Xiang},
  title   = {Time-Consistent Investment and Reinsurance Strategies for Mean-Variance Insurers in ${N}$-Agent and Mean-Field Games},
  journal = {North American Actuarial Journal},
  year    = {2022},
  volume  = {26},
  number  = {4},
  pages   = {537--569},
  doi     = {10.1080/10920277.2021.2014891}
}

@article{HeEtAl2023Robust,
  author  = {He, Yong and He, Lin and Chen, Dengsheng and Liu, Zhezhi},
  title   = {Mean field and $n$-insurers games for robust optimal reinsurance-investment in correlated markets},
  journal = {Journal of Industrial and Management Optimization},
  year    = {2023},
  volume  = {19},
  number  = {9},
  pages   = {6806--6825},
  doi     = {10.3934/jimo.2022240}
}

@article{BaiEtAl2022Hybrid,
  author  = {Bai, Yanfei and Zhou, Zhongbao and Xiao, Helu and Gao, Rui and Zhong, Feimin},
  title   = {A hybrid stochastic differential reinsurance and investment game with bounded memory},
  journal = {European Journal of Operational Research},
  year    = {2022},
  volume  = {296},
  pages   = {717--737},
  doi     = {10.1016/j.ejor.2021.04.046}
}

@article{LiangXiaZou2024,
  author  = {Liang, Zongxia and Xia, Yi and Zou, Bin},
  title   = {A two-layer stochastic game approach to reinsurance contracting and competition},
  journal = {Insurance: Mathematics and Economics},
  year    = {2024},
  volume  = {119},
  pages   = {226--237},
  doi     = {10.1016/j.insmatheco.2024.09.002}
}

@article{CaoEtAl2025,
  author  = {Cao, Jingyi and Li, Dongchen and Young, Virginia R. and Zou, Bin},
  title   = {Co-opetition in reinsurance markets: When {P}areto meets {S}tackelberg and {N}ash},
  journal = {Insurance: Mathematics and Economics},
  year    = {2025},
  volume  = {125},
  pages   = {103133},
  doi     = {10.1016/j.insmatheco.2025.103133}
}

@article{Iglehart1969,
  author  = {Iglehart, Donald L.},
  title   = {Diffusion approximations in collective risk theory},
  journal = {Journal of Applied Probability},
  year    = {1969},
  volume  = {6},
  number  = {2},
  pages   = {285--292},
  doi     = {10.2307/3211999}
}

@article{Browne1995,
  author  = {Browne, Sid},
  title   = {Optimal Investment Policies for a Firm with a Random Risk Process: Exponential Utility and Minimizing the Probability of Ruin},
  journal = {Mathematics of Operations Research},
  year    = {1995},
  volume  = {20},
  number  = {4},
  pages   = {937--958},
  doi     = {10.1287/moor.20.4.937}
}

@article{AsmussenTaksar1997,
  author  = {Asmussen, S{\o}ren and Taksar, Michael},
  title   = {Controlled diffusion models for optimal dividend pay-out},
  journal = {Insurance: Mathematics and Economics},
  year    = {1997},
  volume  = {20},
  number  = {1},
  pages   = {1--15},
  doi     = {10.1016/S0167-6687(96)00017-0}
}

@article{EspinosaTouzi2015,
  author  = {Espinosa, Gilles-Edouard and Touzi, Nizar},
  title   = {Optimal investment under relative performance concerns},
  journal = {Mathematical Finance},
  year    = {2015},
  volume  = {25},
  number  = {2},
  pages   = {221--257},
  doi     = {10.1111/mafi.12034}
}

@article{LackerZariphopoulos2019,
  author  = {Lacker, Daniel and Zariphopoulou, Thaleia},
  title   = {Mean field and $n$-agent games for optimal investment under relative performance criteria},
  journal = {Mathematical Finance},
  year    = {2019},
  volume  = {29},
  number  = {4},
  pages   = {1003--1038},
  doi     = {10.1111/mafi.12206}
}

@article{GerberPafumi1998Utility,
  title   = {Utility functions: From risk theory to finance},
  author  = {Gerber, Hans U. and Pafumi, G{\'e}rard},
  journal = {North American Actuarial Journal},
  volume  = {2},
  number  = {3},
  pages   = {74--91},
  year    = {1998},
  doi     = {10.1080/10920277.1998.10595671}
}

\end{document}